\documentclass[manyauthors]{fundam}

\usepackage{amssymb,paralist,theorem,paralist}
\usepackage[format=default,labelfont=bf]{caption}
\usepackage{placeins}
\usepackage{microtype}
\usepackage{yfonts}
\usepackage[utf8]{inputenc}
\usepackage{scalerel}
\usepackage{stmaryrd}
\usepackage{bbold}
\usepackage{fancyhdr}

\usepackage{graphicx}
\usepackage{tikz}
\usepackage{subcaption}
\usepackage{tikz-qtree}
\usepackage{rotating}

\usetikzlibrary{arrows,matrix}
\usetikzlibrary{matrix,calc}
\usetikzlibrary{shapes.geometric}
\usetikzlibrary{positioning}
\usetikzlibrary{calligraphy}
\usetikzlibrary{fit}

\pgfdeclarelayer{background}
\pgfsetlayers{background,main}

\usetikzlibrary{decorations.pathreplacing}
\tikzset{curlybrace/.style={decoration=brace,decorate}}
\usetikzlibrary{automata}
\usetikzlibrary{calc,positioning}
\usetikzlibrary{decorations.pathmorphing}
\usetikzlibrary{shapes}
\tikzset{trinode/.style={draw,triangle,minimum width=2.0cm}}
\tikzset{snake/.style={decorate,decoration=snake}}
\tikzset{curlybrace/.style={decoration=brace,decorate}}
\tikzset{triangle/.style={regular polygon,regular polygon sides=3}}
\tikzset{edge from parent path={(\tikzparentnode) -- (\tikzchildnode.north)}}
\usetikzlibrary{arrows.meta}

\usepackage{wrapfig}
\usepackage[rightcaption]{sidecap}
\usepackage[colorinlistoftodos]{todonotes}
\usepackage{mathtools}
\usepackage{enumitem}
\usepackage{multicol} 
\usepackage{todonotes}
\usepackage{imakeidx}         

\usepackage{hyperref}

\newtheorem{observation}[definition]{Observation}

\newcommand{\cA}{\mathcal A}
\newcommand{\cB}{\mathcal B}
\newcommand{\cC}{\mathcal C}

\newcommand{\C}{\mathrm{C}}
\newcommand{\T}{\mathrm{T}}

\newcommand{\im}{\mathrm{im}}

\newcommand{\supp}{\mathrm{supp}}

\newcommand{\e}{\mathrm{e}} 

\DeclareMathOperator{\n}{-}

\newcommand{\Pol}{\mathrm{Pol}}
\newcommand{\sfPol}{\mathsf{Pol}}

\newcommand{\Rec}{\mathrm{Rec}}

\newcommand{\budRec}{\mathrm{bud}\mathord{\n}\mathrm{Rec}}
\newcommand{\budwta}{\mathrm{wta}}

\newcommand{\A}{\mathsf{A}}
\newcommand{\B}{\mathsf{B}}

\newcommand{\sfFtwo}{\mathsf{F}_2}

\newcommand{\sfM}{\mathsf{M}}

\newcommand{\Ratnum}{\mathsf{Rat}}

\newcommand{\St}{\mathsf{S}}

\newcommand{\sfT}{\mathsf{T}}

\newcommand{\V}{\mathsf{V}}

\newcommand{\0}{\mathbb{0}}
\newcommand{\1}{\mathbb{1}}

\newcommand{\h}{\mathrm{h}}
\newcommand{\sfh}{\mathsf{h}}

\newcommand{\height}{\mathrm{height}}

\newcommand{\rk}{\mathrm{rk}}

\newcommand{\sem}[1]{[\![#1]\!]}

\DeclareMathOperator*{\bigplus}{\scalerel*{+}{\sum}}

\newcommand{\nf}{\mathrm{nf}}

\newcommand{\state}{\mathrm{state}}

\newcommand{\id}{\mathrm{id}}

\newcommand{\ttop}{\mathrm{top}}

\usepackage[most]{tcolorbox}

\tikzset{small circle/.style={circle, draw=black, inner sep=0pt,outer sep=0pt, minimum size=2.5pt}}

\newcommand{\sfMon}{\mathsf{Mon}}
\newcommand{\rmMon}{\mathrm{Mon}}

\newcommand{\mysucc}{\mathrm{succ}}

\title{Characterization of deterministically recognizable weighted tree languages over commutative semifields by 
  finitely generated and cancellative scalar algebras}

\date{\today}

\author{Zolt\'an F\"ul\"op\thanks{Project no TKP2021-NVA-09 has been implemented with the support provided by the Ministry of Culture and Innovation of Hungary from the National Research, Development and Innovation Fund, financed under the TKP2021-NVA funding scheme.}\corresponding\\Institute of Informatics\\University of Szeged\\ Hungary \and
Heiko Vogler\\Faculty of Computer Science\\Technische Universit\"at Dresden\\ Germany}

\begin{document}
\maketitle
\address{fulop@inf.u-szeged.hu}

\runninghead{Z. F\"ul\"op, H. Vogler}{Characterization of deterministically recognizable weighted tree languages}

\begin{abstract} Due to the works of S. Bozapalidis and A. Alexandrakis,
  there is a well-known characterization of  recognizable weighted tree languages over fields in terms of finite-dimensionality of syntactic vector spaces. Here we prove a characterization of bottom-up deterministically recognizable weighted tree languages over commutative semifields in terms of the requirement that the respective m-syntactic scalar algebras are finitely generated. The concept of scalar algebra is introduced in this paper; it is obtained from the concept of vector space by disregarding the addition of vectors.
  Moreover, we prove a minimization theorem for bottom-up-deterministic weighted tree automata and we construct the minimal automaton.
  \end{abstract}

  \section{Introduction}

  In formal language theory, the Myhill-Nerode theorem provides an algebraic characterization of the set of recognizable string languages:
  a language is recognizable if and only if  its syntactic congruence over the free monoid of all strings has finite index \cite{myh57,ner58} (cf. also \cite[Thm.~3.9]{hopull79}, \cite[Thm.~16.3]{koz97}). A corresponding characterization of the set of recognizable tree languages (the  Myhill-Nerode theorem for trees) was given in \cite{magmor70,ste92,koz92} and \cite[Thm.~2.7.1]{gecste97}:
  a tree language over a ranked alphabet $\Sigma$ is recognizable if and only if its syntactic congruence (or: Myhill-Nerode congruence) over the initial $\Sigma$-term algebra $\sfT_\Sigma$ has finite index.

  Later, also for weighted tree languages which are recognizable by weighted tree automata (wta) such an algebraic characterization was proved. A $(\Sigma,\B)$-weighted tree language is a mapping $r: \T_\Sigma \to B$  which associates to each tree over some ranked alphabet $\Sigma$  a value in an algebra  $\B=(B,\oplus,\otimes,\0,\1)$ of weights; this algebra has two binary operations (a summation $\oplus$ and a multiplication $\otimes$) with unit elements $\0$ and~$\1$ each.
  Weighted tree automata were investigated for several classes of weight algebras:  wta over the bounded lattice with  the real number interval $[0,1]$ and $\max$ and $\min$ as operations  \cite{inafuk75}, wta over fields \cite{berreu82},  wta over semirings \cite{aleboz87,esikui03,fulvog09new}, wta over pairs of t-conorm and t-norm on the unit interval $[0,1]$ \cite{bozlou10}, and wta over strong bimonoids \cite{rad10} (cf. \cite{fulvog24} for a survey). Moreover, wta over weight algebras in which the multiplication is generalized are wta over distributive multioperator monoids \cite{kui98,mal05a}, wta over multioperator monoids \cite{fulmalvog09,stuvogfue09,fulstuvog12}, and wta over tree valuation monoids \cite{drogoemaemei11,teiost15}. Weighted tree languages which are recognizable by wta are called recognizable weighted tree languages.

  The first generalization of the Myhill-Nerode theorem to $(\Sigma,\B)$-weighted tree languages, where $\B$ is a field,  was given in \cite{bozale89} and \cite{boz91} (cf. Theorem~\ref{th:MN-fields-det}).  There the algebra used for the characterization is a $(\Sigma,\B)$-vector space. A $(\Sigma,\B)$-vector space is a vector space over the field $\B$ \cite{axl24,mooyaq98} together with a $\Sigma$-algebra on its carrier set; the operations of the $\Sigma$-algebra are compatible with the scalar multiplication and the addition of the vector space (cf. e.g. \cite[Sec.~2.7.7]{fulvog24}). In \cite{bozale89, boz91} and \cite{fulste11} such algebras are called $\B$-$\Sigma$ algebras and $\B\Sigma$-algebras, respectively.
  The characterization result is as follows:  a $(\Sigma,\B)$-weighted tree language $r$ is recognizable if and only if the quotient of the initial $(\Sigma,\B)$-vector space of polynomial
weighted tree languages induced by the syntactic congruence $\approx_r$ of $r$ is finite-dimensional. We will call this theorem  ``B-A's theorem'' and we refer also to \cite[Thm.~18.4.1]{fulvog24} for an alternative proof.

Then, in  \cite{bor03,bor04b}, the  generalization of the Myhill-Nerode theorem to weighted tree languages  aimed at the characterization of the set of $(\Sigma,\B)$-weighted tree languages recognizable by bottom-up-deterministic wta, denoted by $\budRec(\Sigma,\B)$, where $\B$ is a commutative semifield. The algebra used for the characterization
is a $\Sigma$-algebra (as in the unweighted case). The characterization result says that a  $(\Sigma,\B)$-weighted tree language $r$ is bottom-up deterministically recognizable if and only if its Myhill-Nerode congruence $\equiv_r$ on the $\Sigma$-term algebra $\sfT_\Sigma$ has finite index, where $\equiv_r$ is a generalization of the corresponding concept in the unweighted case. In \cite[Cor.~3(a)]{bor03} it is shown how to define, for each bottom-up deterministic wta, an equivalent one which is minimal with respect to the number of states. For a similar characterization of bottom-up deterministically recognizable weighted tree languages, we also refer to \cite{mal08c} where the set of weight algebras was extended to  multiplicatively cancellative and commutative semirings (under some mild changes of the scenario).

As the first main result of our paper, we prove a B-A-like theorem also for the set $\budRec(\Sigma,\B)$, where $\B=(B,\oplus,\otimes,\0,\1)$ is a commutative semifield (cf. Theorem~\ref{thm:budRec-mono-det}). However, for the characterization we use the new concept of $(\Sigma,\B)$-scalar algebra. These algebras are in between $(\Sigma,\B)$-vector spaces used in  \cite{bozale89,boz91} and  $\Sigma$-algebras used in 
\cite{bor03,bor04b} and they are obtained as follows. First,  we introduce the concept of $\B$-scalar algebra,
which is obtained from the concept of $\B$-vector space by disregarding the addition of vectors.
Dropping the addition from $\B$-vector spaces is suggested by the more or less well known fact that, in the computation of the semantics of a bottom-up-deterministic $(\Sigma,\B)$-wta, the summation $\oplus$ of the weight algebra $\B$ does not play a role  (cf. Corollary~\ref{cor:budetwta-addition-irrelevant-det-new}). 

Formally, a \emph{$\B$-scalar algebra} is an algebra $(V,0)$ where $V$ is a set and $0 \in V$; moreover, there is a scalar multiplication $\cdot: B \times V \to V$ which satisfies,
 for every $b,b' \in B$ and $v \in V$, the axioms:
    \begin{eqnarray}
&(b \otimes b') \cdot v = b \cdot (b' \cdot v)  \\
&\1 \cdot v = v  \\
&      b \cdot 0 = \0 \cdot v = 0  \enspace.
    \end{eqnarray}
    Alternatively, a $\B$-scalar algebra might be considered as a left group action \cite{ker99,hum04} extended by an annihilating zero. Second, in analogy to $(\Sigma,\B)$-vector spaces, a \emph{$(\Sigma,\B)$-scalar algebra} is obtained by equipping a $\B$-scalar algebra  with a $\Sigma$-algebra where the operations are compatible with the scalar multiplication.

It seems that a $\B$-scalar algebra is a weak algebraic structure. However, the existence of a
pair-independent generating set $H$ of a cancellative $\B$-scalar algebra $(V,0)$ guarantees a valuable property, viz., that each $v \in V\setminus\{0\}$ can be decomposed in a unique way as $v=d \cdot u$ for some $d \in B\setminus\{\0\}$ and $u \in H\setminus\{0\}$ (cf. Lemma~\ref{lm:crucial-canc-pair-ind-imply-uniqueness-det}(2)). This is similar to the fact that, in a $\B$-vector space with basis $H$, each vector can be represented in a unique way as a linear combination of the base vectors.
And this analogy inspired us to call a pair-independent generating set of a cancellative $\B$-scalar algebra a \emph{scalar basis}.
Among others, this property makes $(\Sigma,\B)$-scalar algebras  adequate to formulate a deterministic version of  B-A's theorem.

For our characterization result, we consider the $(\Sigma,\B)$-scalar algebra $(\rmMon(\Sigma,\B),\widetilde{\0},\ttop)$ where $(\rmMon(\Sigma,\B),\widetilde{\0})$ is the $\B$-scalar algebra  of monomial $(\Sigma,\B)$-weighted tree languages, which we denote by $\sfMon(\Sigma,\B)$. Each such monomial is a mapping $r: \T_\Sigma \to B$ for which there exist a $b \in B$ and a $\Sigma$-tree~$\xi$ such that $r(\xi)=b$ and $r(\zeta)=\0$ for each $\Sigma$-tree $\zeta \ne \xi$; we denote this monomial by $b.\xi$. The zero $\widetilde{\0}$ is the constant weighted tree language which maps each $\Sigma$-tree to $\0$.
The interpretation  $\ttop$ assigns to each $k$-ary $\sigma \in \Sigma$ the usual top-concatenation $\ttop(\sigma)$ with $\sigma$. We show that $(\rmMon(\Sigma,\B),\widetilde{\0},\ttop)$ is initial in the set of all $(\Sigma,\B)$-scalar algebras (cf. Theorem~\ref{lm:Mon-SigmaB-initial-det}).
For a given weighted tree language $r: \T_\Sigma \to B$, we factorize $(\rmMon(\Sigma,\B),\widetilde{\0},\ttop)$ by the m-syntactic congruence  $\sim_r$ of $r$ which is defined by
\[
b_1.\xi_1 \sim_r b_2.\xi_2 \ \text{ if } \ \text{for each $\Sigma$-context $c$}:  \ b_1 \otimes r(c[\xi_1]) = b_2 \otimes r(c[\xi_2])\enspace,
  \]
  where $c[\xi_i]$ denotes the tree obtained by plugging $\xi_i$ into $c$. The m-syntactic congruence $\sim_r$ is the restriction of $\approx_r$ (used in B-A's theorem) to \underline{m}onomial $(\Sigma,\B)$-weighted tree languages, which is the reason for the prefix ``m-''. Our characterization theorem is as follows:

\

{\bf Theorem \ref{thm:budRec-mono-det}} Let $\Sigma$ be a ranked alphabet and $\B=(B,\oplus,\otimes,\0,\1)$ a commutative semifield. Moreover, let $r: \T_\Sigma \to B$. The following three statements are equivalent.
\begin{compactenum}
 \item[(A)] $r \in \budRec(\Sigma,\B)$.
  \item[(B)] There exists a congruence  $\sim$ on the $(\Sigma,\B)$-scalar algebra $(\rmMon(\Sigma,\B),\widetilde{\0},\ttop)$ such that  $\sim$ respects~$r$ and the $\B$-scalar algebra $\sfMon(\Sigma,\B)/_\sim$ is finitely generated and cancellative.
 \item[(C)] The $\B$-scalar algebra $\sfMon(\Sigma,\B)/_{\sim_r}$ is finitely generated.
\end{compactenum}

\

In our second main result we show how to construct, for each bottom-up-deterministic $(\Sigma,\B)$-wta, an equivalent bottom-up-deterministic $(\Sigma,\B)$-wta which is minimal with respect to the number of states (cf. Theorem~\ref{thm:minimization-theorem-new-2}). For the construction, we assume a reasonable decidability condition and a computability condition, both on $\sim_r$.

In Table~\ref{table:overview} we summarize the main ingredients of the mentioned characterizations of  (a)~\cite{bozale89,boz91}, (b) \cite{bor03,bor04b}, and (c) the present paper, where we abbreviate bottom-up-deterministic by bu-deterministic.
\begin{table}
  \centering
{\tiny
\begin{tabular}{l|c|c|c|}
 & (a) \cite{bozale89,boz91} & (b) \cite{bor03,bor04b} & (c) present paper \\[2mm]
\hline
weighted tree language $r$ &   & &  \\
recognized by: &  wta & bu-deterministic wta & bu-deterministic wta \\[2mm]
\hline
type of algebra used &   & &  \\
for the characterization: &   $(\Sigma,\B)$-vector space &  $\Sigma$-algebra & $(\Sigma,\B)$-scalar algebra\\[2mm]
\hline
congruence used &   & & \\                                            
for the characterization: &  syntactic congruence $\approx_r$ &  Myhill-Nerode congruence $\equiv_r$  &  m-syntactic congruence $\sim_r$\\[2mm]
\hline
the characterization  & the quotient of the initial  & the quotient of the initial &  the quotient of the initial \\
of recognizability: &   $(\Sigma,\B)$-vector space  by $\approx_r$ & $\Sigma$-term algebra $\sfT_\Sigma$ by $\equiv_r$ & $(\Sigma,\B)$-scalar algebra  by $\sim_r$\\
& is finite-dimensional & is finite  & is finitely generated
\end{tabular}
}
\caption{\label{table:overview} Summary of characterizations of  (a)~\cite{bozale89,boz91}, (b) \cite{bor03,bor04b}, and (c) the present paper, where we abbreviate bottom-up-deterministic by bu-deterministic.}
\end{table}
As can be seen, the main differences between the algebraic characterizations in  \cite{bor03,bor04b} and in the present paper are as follows: (i) in the choice of the algebra: we use $(\Sigma,\B)$-scalar algebras, and (ii) as the congruence used for characterization: we use the m-syntactic congruence $\sim_r$.
Therefore our results are rather in the style of \cite{bozale89,boz91}. It might be a benefit of our approach that $(\Sigma,\B)$-scalar algebras are very close to group actions, which are well known in algebra.

The paper is organized as follows. In Section \ref{sec:preliminaries-det}, we recall some notions from universal algebra, trees, weighted sets and weighted tree languages, and well-founded induction. In Section~\ref{sec:wta-det}, we recall the concept of wta and  we prove properties of bottom-up-deterministic wta.
For a better understanding of our characterization theorem, in Section~\ref{sect:B-A-BLs-result-det} we recall the original B-A's theorem (cf. Theorem~\ref{th:MN-fields-det}). In Section~\ref{subsect:mono-spaces-det}, we introduce the concepts of $\B$-scalar algebra and of $(\Sigma,\B)$-scalar algebra and we show some important properties of $\B$-scalar algebras. Moreover, we define the particular $(\Sigma,\B)$-scalar algebra of monomial weighted tree languages.
In Section~\ref{sec:syntactic-congruence-det}, we define the m-syntactic congruence $\sim_r$ for any weighted tree language~$r$ and show some of its important properties.
In Section~\ref{sec:result-det}, we prepare and prove our characterization theorem for bottom-up-deterministic wta over commutative semifields (cf. Theorem~\ref{thm:budRec-mono-det}),  we prove our minimization theorem (cf. Theorem~\ref{thm:minimization-theorem-new-2}), and we give a characterization of minimality (cf. Corollary~\ref{minimal-iff-degrees-equal}).

\section{Preliminaries}
\label{sec:preliminaries-det}

\paragraph{Basic notions.} We denote by $\mathbb{N}$ the set $\{0,1,2,\ldots\}$ of natural numbers; moreover, we let $\mathbb{N}_+= \mathbb{N} \setminus \{0\}$. For each $k \in \mathbb{N}$, we let $[k]= \{1,\ldots,k\}$. In particular, $[0]= \emptyset$.
    
Let $A$ and $B$ be sets. We denote by $|A|$ the cardinality of $A$. 
     We denote by $B^A$ the set of all mappings $f: A \to B$. For each mapping $f: A \to B$, we denote by $\im(f)$ the set $\{f(a)\mid a\in A\}$. For each $A'\subseteq A$ 
    the \emph{restriction of $f$ to $A'$} is the mapping $f|_{A'}: A' \to B$ defined by $f|_{A'}(a)=f(a)$ for each $a\in A'$.
     
    Let $k\in \mathbb{N}$ and $f: A^k \to A$ be a $k$-ary operation on $A$ and $A' \subseteq A$. We say that \emph{$A'$ is closed under $f$} if, for every $a_1,\ldots,a_k \in A'$, we  have $f(a_1,\ldots,a_k) \in A'$.
    We denote by $\mathrm{Ops}(A)$ the set of all (finitary) operations on $A$. The identity mapping on $A$ is denoted by~$\id_A$.
    
Let $A$ be a set and  $\rho\subseteq A\times A$  a  relation on $A$.
For every $a_1,a_2\in A$ we write $a_1\rho a_2$ to denote that $(a_1,a_2)\in\rho$. Let $\rho$ be an equivalence relation. For each $a\in A$, the \emph{equivalence class of $a$ (modulo~$\rho$)}, denoted by $[a]_\rho$, is the set $\{b\in A \mid a\rho b\}$. For each $A'\subseteq A$, we put $A'/_\rho =\{[a]_\rho \mid a\in A'\}$. We say that \emph{$\rho$ has finite index} if the set $A/_\rho$ is finite.
 If this is the case, then the number $|A/_\rho|$ is the \emph{index of $\rho$}.

\paragraph{Universal algebras.}
We assume that the reader is familiar with basic concepts of universal algebra, like subalgebra, congruence relation (for short: congruence), quotient algebra, homomorphism, kernel, and free algebra with a generating set \cite{gra68,gogthawagwri77,bursan81,wec92}. We only recall some notions.

  A {\em universal algebra} is a pair $\A=(A,\theta)$ where $A$ is a set and $\theta: I \to \mathrm{Ops}(A)$ for some set $I$.
  The \emph{type of $\A$} (or: \emph{signature of $\A$}) is the mapping $\tau: I \to \mathbb{N}$ such that, for each $i\in I$, the number $\tau(i)$ is  the arity of the operation $\theta(i)$. Then we also say that \emph{$\A$ is of type~$\tau$}. For each mapping $\tau: I \to \mathbb{N}$, we denote by~$\mathrm{K}(\tau)$ the \emph{set of all universal algebras of type $\tau$}.

  For the rest of this paragraph, let $\A=(A,\theta)$ be a universal algebra of type $\tau:I\to \mathbb{N}$.

Let $B \subseteq A$ and $\mathsf{B} = (B,\theta_B)$ be a universal algebra of type $\tau$. If, for every $i \in I$ and $b_1,\ldots,b_{\tau(i)}\in B$, the equality 
$\theta_{B}(i)(b_1,\ldots,b_{\tau(i)})=\theta(i)(b_1,\ldots,b_{\tau(i)})$
holds, then we say that $\mathsf{B}$ is a {\em subalgebra of $\A$}. In order to avoid complex subscripts, we denote $ (B,\theta_B)$ also by $(B,\theta)$.

For each $H \subseteq A$ and $O\subseteq \im(\theta)$, we denote by $\langle H \rangle_{O}$ the smallest subset of $A$ which contains $H$ and is closed under the operations in
$O$. We say that $H$ is a \emph{generating set of $\A$} (or: \emph{$H$ generates $\A$}) if $\langle H \rangle_{\im(\theta)}=A$. If this is the case and $H$ is finite, then $\A$ is \emph{finitely generated}.

Now let $I' \subseteq I$. Then the universal algebra $(A,\theta')$ with $\theta'=\theta|_{I'}$ is a \emph{reduct of $\A$}.

Free algebras and initial algebras play a fundamental role in our paper.
Let $\tau: I \to \mathbb{N}$ be a mapping and let  $\mathcal{K} \subseteq \mathrm{K}(\tau)$ be a set of universal algebras of type $\tau$. Moreover,  let $H \subseteq A$.
The algebra~$\A$ is called \emph{free in~$\mathcal{K}$ with generating set $H$} if the following conditions hold:
\begin{compactenum}
  \item[(a)] $\A \in \mathcal{K}$,
  \item[(b)] $H$ generates $\A$, and
\item[(c)] for every universal algebra $\A'=(A',\theta')$ in $\mathcal{K}$ and mapping $f\colon H \rightarrow A'$, there exists an extension of $f$ to an algebra homomorphism $h\colon A \rightarrow A'$ from $\A$ to $\A'$.
  \end{compactenum}
If the extension $h'$ in (c) exists, then it is unique (cf. e.g. \cite[Lm.~3.3.1]{baanip98} and \cite[Thm.~II.10.7]{bursan81}).
If $\A$ is free  in~$\mathcal{K}$ with generating set $H=\emptyset$, then, for each~$\A'=(A',\theta')$ in $\mathcal{K}$, there exists exactly one algebra homomorphism $h\colon A \rightarrow A'$  from $\A$ to $\A'$. In this case $\A$ is called \emph{initial in~$\mathcal{K}$}.

Let $\sim$ a congruence on $\A$. The \emph{quotient algebra of $\A$ (with respect to $\sim$)} is the universal algebra $\A/_\sim = (A/_\sim,\theta/_\sim)$ of type $\tau$  where $\theta/_\sim: I \to \mathrm{Ops}(A/_\sim)$ such that, for each $i \in I$ with $\tau(i)=k$ for some $k \in \mathbb{N}$, and for every $a_1,\ldots,a_k \in A$ we have   
\(
(\theta/\!_\sim)(i)([a_1]_\sim,\ldots,[a_k]_\sim) = [\theta(i)(a_1,\ldots,a_k)]_\sim
\).
Then we call the mapping $\pi_\sim: A \to A/_\sim$ with $\pi_\sim(a)=[a]_\sim$ for each $a \in A$ the \emph{canonical mapping}.

\begin{lemma} \label{thm:canonical-map-of-congr-is-hom} \rm \cite[\S~7, Lm.~2]{gra68} Let $\A=(A,\theta)$ be a universal algebra and $\sim$ a congruence on $\A$. The canonical mapping $\pi_\sim: A \to A/_\sim$ is an algebra homomorphism from $\A$ to $\A/\!_\sim$. 
  \end{lemma}

\begin{lemma}\rm\cite[§7, Lm.~3]{gra68} \label{lm:hom-image=subalgebra} Let $(A_1,\theta_1)$ and  $(A_2,\theta_2)$ be two universal algebras of the same type. Moreover, let $h: A_1 \rightarrow A_2$ be an algebra homomorphism. Then $\im(h)$ is closed under the operations in $\im(\theta_2)$. Hence $(\im(h),\theta_2|_{\im(h)})$ is a subalgebra of $(A_2,\theta_2)$.
\end{lemma}

\begin{lemma} \label{thm:congruence-restricted-to-subalgebra} \rm \cite[\S~7, Lm.~4]{gra68} and \cite[Lm.~6.17]{bursan81} Let $\A=(A,\theta)$ be a universal algebra and $\sim$ a congruence on $\A$. Moreover, let $\B=(B,\theta)$ be a subalgebra of $\A$. Then the relation $\sim\cap \,(B\times B)$ is a congruence on $\B$.
\end{lemma}

Let $\A_1$ and $\A_2$ be universal algebras of the same type. If there exists a bijective algebra homomorphism from $\A_1$ to $\A_2$, then $\A_1$ and $\A_2$ are \emph{isomorphic}, for short: $\A_1 \cong \A_2$.

\begin{theorem}\label{thm:kernel-is-congruence} {\rm \cite[\S 7, Lm. 6]{gra68} and \cite[\S 11,~Thm.~1]{gra68}} Let $\A_1=(A_1,\theta_1)$ and $\A_2=(A_2,\theta_2)$ be universal algebras of the same type.
Moreover, let $h$ be an algebra homomorphism from $\A_1$ to $\A_2$.
Then the kernel of $h$, i.e., the set $\mathrm{ker}(h)=\{(a,a') \in A_1 \times A_2 \mid h(a)=h(a')\}$, is a congruence relation on $\A_1$. Moreover, if $h$ is surjective, then $\A_1/_{\ker(h)} \cong \A_2$.
\end{theorem}

\begin{corollary} \label{cor:image-of-hom-isomorphic-to-quotient-of-kernel}\rm Let $\A_1=(A_1,\theta_1)$ and $\A_2=(A_2,\theta_2)$ be universal algebras of the same type. Moreover, let $h$ be an algebra homomorphism from $\A_1$ to $\A_2$. Then $\A_1/_{\ker(h)} \cong (\im(h),\theta_2|_{\im(h)})$.
\end{corollary}
\begin{proof} By Lemma \ref{lm:hom-image=subalgebra}, $(\im(h),\theta_2|_{\im(h)})$ is an algebra of the same type as $\A_2$. Then $h': A_1 \to \im(h)$ defined, for each $a \in A_1$ by $h'(a)=h(a)$ is a homomorphism from $\A_1$ to $(\im(h),\theta_2|_{\im(h)})$. Clearly, $\ker(h)=\ker(h')$. Since $h'$ is surjective, by Theorem \ref{thm:kernel-is-congruence}, we obtain that $\A_1/_{\ker(h)} \cong (\im(h),\theta_2|_{\im(h)})$.
  \end{proof}

Let $\sim$ and $\approx$ be equivalence relations on $A$ such that $\sim \subseteq \approx$. We define the relation $\approx\!/_\sim$ on $A/_\sim$ such that, for every $[a_1]_\sim,[a_2]_\sim\in A/_\sim$, we let $[a_1]_\sim \approx\!/_\sim \,[a_2]_\sim$ if $a_1\approx a_2$. Clearly, the relation $\approx\!/_\sim$ is well defined because $\sim \subseteq \approx$ and it is an equivalence relation.

\begin{theorem}\label{thm:snd-isom-theorem} {\rm \cite[\S 11, Lm.~2]{gra68} and \cite[\S 11, Thm.~4]{gra68}}  
  Let $\A$ be a universal algebra. Moreover, let $\sim$ and $\approx$ be congruences on $\A$ such that $\sim \subseteq \approx$. Then the following two statements hold.
  \begin{compactenum}
  \item[(1)] The relation $\approx\!/_\sim$ is a congruence on $\A/_\sim$.
  \item[(2)] $(\A/_\sim)/_{(\approx/_\sim)} \cong \A/_\approx$.
    \end{compactenum}
\end{theorem}

\paragraph{Commutative semifields and fields.}
    A \emph{commutative semifield } is an algebra  $\B=(B,\oplus,\otimes,\mathbb{0},\mathbb{1})$ where $(B,\oplus,\mathbb{0})$ is a commutative monoid and $(B\setminus\{\0\},\otimes,\mathbb{1})$ is an Abelian group, $\otimes$ distributes over $\oplus$ (from both sides), and  $\mathbb{0}$ is an annihilator for $\otimes$, i.e.,  $b \otimes \mathbb{0} = \mathbb{0} \otimes b= \mathbb{0}$ holds for every $b \in B$. As usual, $\otimes$ has higher binding priority than $\oplus$. 
 It is easy to see that each commutative semifield is \emph{zero-divisor free}, i.e., $a\otimes b =\0$ implies that $a=\0$ or $b=\0$, for every $a,b\in B$.
    A \emph{field} is a commutative semifield $\B=(B,\oplus,\otimes,\mathbb{0},\mathbb{1})$ in which $(B,\oplus,\0)$ is an Abelian group.
For example, $\Ratnum = (\mathbb{Q},+,\cdot,0,1)$ is the field of rational numbers.

  \paragraph{Vector spaces.}  We assume that the reader is familiar with basic concepts concerning vector spaces like linear mappings, linear dependency, linear independency,  and basis \cite{axl24,mooyaq98}.  We only recall some notions. Let $\B=(B,\oplus,\otimes,\0,\1)$ be a field. A \emph{$\B$-vector space} is a commutative group $(V,+,0)$ together with a mapping  $\cdot: B\times V \rightarrow V$, called \emph{scalar multiplication}. The scalar multiplication satisfies certain algebraic laws on the compatibility with the operations $\oplus$ and $\otimes$ of $\B$ (cf. \cite[Def.~1.12]{axl24} and \cite[p.~319 and Thm.~5.1]{mooyaq98}).

  Let $\V=(V,+,0)$ and $\V'=(V',+',0')$ be $\B$-vector spaces. A mapping $f: V\to V'$ is \emph{linear (from $\V$ to $\V'$)} if, for every $b_1,b_2 \in B$ and $u,v\in V$, the equation $f(b_1\cdot u+b_2\cdot v)=b_1 \cdot f(u)+'b_2 \cdot f(v)$
holds. The $\B$-vector spaces $\V$ and $\V'$  are  \emph{isomorphic} if there exists a bijective linear mapping from $\V$ to $\V'$.
A {\em linear form (over $\V$)} is a linear mapping from the $\B$-vector space $\V=(V,+,0)$ to the $\B$-vector space  $(B,\oplus,\0)$.  We note that, by viewing $\B$-vector spaces as universal algebras, linear mappings and linear forms are nothing else but algebra homomorphisms.

    \paragraph{Ranked alphabets, trees, and contexts.}
 A {\em ranked alphabet} is a pair $(\Sigma,\rk)$ where $\Sigma$ is a non-empty and finite set and $\rk: \Sigma \rightarrow \mathbb{N}$ is a mapping, called \emph{rank  mapping},  such that $\rk^{-1}(0)\not= \emptyset$. For each $k\in \mathbb{N}$, we denote the set $\rk^{-1}(k)$ by $\Sigma^{(k)}$. Sometimes we write $\sigma^{(k)}$ to indicate that $\sigma \in \Sigma^{(k)}$. Moreover, we abbreviate $(\Sigma,\rk)$ by $\Sigma$ and assume that the rank mapping is known or irrelevant.

 Let $\Sigma$ be a ranked alphabet and let  $X$ be a set disjoint with $\Sigma$.
     The set of \emph{$\Sigma$-terms (or $\Sigma$-trees) over $X$, } denoted by $\T_\Sigma(X)$, is the smallest set $T$ such that (i) $\Sigma^{(0)} \cup X \subseteq T$ and (ii) for every $k \in \mathbb{N}_+$, $\sigma \in \Sigma^{(k)}$, and $\xi_1,\ldots, \xi_k \in T$, we have $\sigma(\xi_1,\ldots,\xi_k) \in T$. We abbreviate $\T_\Sigma(\emptyset)$ by $\T_\Sigma$. Due to the condition $\Sigma^{(0)}\ne\emptyset$ we have $\T_\Sigma\ne \emptyset$.

Obviously, for each $\xi \in \T_\Sigma$ there exist unique $k\in \mathbb{N}$, $\sigma\in\Sigma^{(k)}$, and $\xi_1,\ldots,\xi_k \in \T_\Sigma$ such that $\xi=\sigma(\xi_1,\ldots,\xi_k)$. (In case $k=0$ we identify $\sigma()$ and $\sigma$.) In order to avoid the repetition of all these quantifications, we henceforth only write ``$\xi=\sigma(\xi_1,\ldots,\xi_k)$'' when we consider a tree in $\T_\Sigma$ with the above mentioned quantifications.

The mapping $\height: \T_\Sigma \to \mathbb{N}$ is defined by induction as usual by $\height(\alpha) = 0$ for each $\alpha \in \Sigma^{(0)}$ and $\height(\sigma(\xi_1,\ldots,\xi_k)) = 1 + \max( \height(\xi_i) \mid i \in [k])$ for each $k \in \mathbb{N}_+$, $\sigma \in \Sigma^{(k)}$, and $\xi_1,\ldots,\xi_k \in \T_\Sigma$.

Let $z\not\in \Sigma$ be a symbol. We will consider $z$ as a variable.
  We denote by $\C_\Sigma$ the set of all trees in $\T_\Sigma(\{z\})$ in which $z$ occurs exactly once. 
  The elements of $\C_\Sigma$ are called \emph{$\Sigma$-contexts} (or, simply, \emph{contexts}). In particular, $z$ is a context. A context $c \in \C_\Sigma$ is called \emph{elementary} if there exist $k \in \mathbb{N}_+$, $\sigma \in \Sigma^{(k)}$, $i \in [k]$, and $\xi_1,\ldots,\xi_{i-1},\xi_{i+1},\ldots,\xi_k \in \T_\Sigma$ such that
$e =  \sigma(\xi_1,\ldots,\xi_{i-1},z,\xi_{i+1},\ldots,\xi_k)$.
We denote by $\e\C_\Sigma$ the set of all elementary $\Sigma$-contexts.
Similarly to the specification of a  tree, we drop obvious quantifications and specify  an elementary context by writing only
``$e \in \e\C_\Sigma$ with $e=\sigma(\xi_1,\ldots,\xi_{i-1},z,\xi_{i+1},\ldots,\xi_k)$.''

For every context $c \in \C_\Sigma$ and tree $\xi \in \T_\Sigma(\{z\})$, we denote by $c[\xi]$ the tree obtained by substituting the  only occurrence of $z$ in $c$ by $\xi$. For every $c_1,c_2 \in \C_\Sigma$, we denote $c_1[c_2]$ also by $c_1 \circ_z c_2$. Since $c_1\circ_z c_2 \in \C_\Sigma$, we think of $\circ_z$ as a binary operation on $\C_\Sigma$. In fact, $\circ_z$ is associative and $(\C_\Sigma,\circ_z,z)$ is a monoid.

\begin{lemma}\rm \cite[Prop.~9.1]{berreu82} (also cf. \cite[Lm.~6.2.1]{fulvog24} \label{lm:freely-generated-context-det}  The monoid $(\C_\Sigma,\circ_z,z)$ is free in the set of all monoids with generating set $\e\C_\Sigma$.
\end{lemma}

\paragraph{$\Sigma$-algebras.} A \emph{$\Sigma$-algebra} is a universal algebra $(A,\theta)$ where $\theta: \Sigma \to \mathrm{Ops}(A)$ and, for every $k \in \mathbb{N}$ and $\sigma \in \Sigma^{(k)}$, the operation $\theta(\sigma)$ is $k$-ary. The \emph{$\Sigma$-term algebra} is the $\Sigma$-algebra $\sfT_\Sigma=(\T_\Sigma,\ttop_\Sigma)$ where,
for every $k \in \mathbb{N}$, $\sigma \in \Sigma^{(k)}$, and $\xi_1,\ldots,\xi_k \in \T_\Sigma$, we have $\ttop_\Sigma(\sigma)(\xi_1,\ldots,\xi_k) = \sigma(\xi_1,\ldots,\xi_k)$. 

\begin{theorem}\label{thm:initial-iso-det} {\rm \cite[Prop.~2.3]{gogthawagwri77} and \cite[Thm. 1.3.11]{gecste84}} The $\Sigma$-term algebra $\sfT_\Sigma$ is  initial in the set of all $\Sigma$-algebras.
\end{theorem}

\paragraph{Weighted  sets and weighted tree languages.}\label{subsect:weighted-sets-and-operations-det}
Let $\B$ be a commutative semifield and $A$ be a set.  We call each mapping $r: A \to B$ a \emph{$\B$-weighted set}. The {\em support of $r$}, denoted by $\supp_\B(r)$ (for short: $\supp(r)$), is defined by $\supp_\B(r) = \{a \in A \, | \, r(a) \not= \mathbb{0}\}$. 

Let $Q$ be a finite set. We call the elements of the weighted set $B^Q$ \emph{vectors}. For every vector $v\in B^Q$ and $q\in Q$, we denote $v(q)$ also by $v_q$. We denote by $\0^Q$ the vector which maps each $q\in Q$ to $\0$, i.e., $v_q=\0$ for each each $q\in Q$. Moreover, for every $q\in Q$ and $b \in B$, we denote by $b_q$ the vector in $B^Q$ defined for each $p \in Q$ by $(b_q)_p=b$ if $p=q$ and $\0$ otherwise. Hence $\0_q=\0^Q$ for each $q\in Q$. We call  $\1_q$ the \emph{$q$-unit vector}.

We define the mapping $\cdot : B\times B^Q \to B^Q$, called \emph{scalar multiplication} (of vectors) such that, for every $b\in B$, 
$v\in B^Q$, and $q\in Q$, we let $(b\cdot v)_q=b\otimes v_q$.

A weighted set of the form $r:\T_\Sigma \to B$ is called a \emph{$(\Sigma,\B)$-weighted tree language}. We call $r$ 
\begin{compactitem}
\item \emph{polynomial} if $\supp(r)$ is finite.
\item \emph{monomial} if $\supp(r) \subseteq \{\xi\}$ for some $\xi \in \T_\Sigma$; then we denote $r$ also by $r(\xi).\xi$.
\item \emph{constant} if there exists a $b \in B$ such that, for every $\xi \in \T_\Sigma$, we have $r(\xi) = b$; then we denote $r$ by $\widetilde{b}$.
\end{compactitem}
We note that $\widetilde{\0}$ is a monomial because $\supp(\widetilde{\0})=\emptyset$ and $\T_\Sigma\ne \emptyset$. Moreover,  $\widetilde{\0} = \0.\xi$ for each $\xi \in \T_\Sigma$.

\paragraph{Well-founded induction.}

A \emph{reduction system} \cite{hue80,baanip98} is a pair $(A,\succ)$ where $A$ is a set and $\succ \subseteq A \times A$. A reduction system $(A,\succ)$ is \emph{terminating} if there does not exist an $\mathbb{N}$-indexed family $(a_i \mid i\in \mathbb{N})$ over $A$ such that, for each $i \in \mathbb{N}$, the condition $a_i \succ a_{i+1}$ holds. 

Let $(A,\succ)$ be a terminating reduction system.  Moreover, let $P \subseteq A$ be a subset. We will abbreviate the fact that $a \in P$ by $P(a)$. Then the following holds:
\begin{equation}
\Big( (\forall a \in A): \ [(\forall a' \in A): (a \succ a') \Rightarrow P(a')] \Rightarrow P(a)\Big) \Rightarrow \Big((\forall a \in A): P(a)\Big)\label{equ:well-founded-induction-det} \enspace.
\end{equation}
The formula \eqref{equ:well-founded-induction-det} is called the principle of \emph{proof by well-founded induction on $(A,\succ)$} (for short: \emph{proof by induction on $(A,\succ)$}) \cite{baanip98}.
Often, the proof of the premise of \eqref{equ:well-founded-induction-det} is split into two parts. 
The first part is the \emph{induction base} (for short: I.B.):
\begin{equation*}
(\forall a \in \nf_{\succ}(A)):  P(a) \label{equ:wfi-base-det}
\end{equation*}
where $\nf_\succ(A) = \{a \in A \mid \neg (\exists b \in A): a \succ b\}$ is the set of \emph{normal forms of $\succ$}.
The second part is the  \emph{induction step} (for short: I.S.):
\begin{equation*}\label{eq:wfi-step-det}
\Big( (\forall a \in A \setminus  \nf_{\succ}(A)): \ [(\forall a' \in A): (a \succ a') \Rightarrow P(a')] \Rightarrow P(a)\Big) \enspace.
\end{equation*}
Usually, the subformula $(\forall a' \in A): (a \succ a') \Rightarrow P(a')$ is called \emph{induction hypothesis} (for short: I.H.).

For instance, for each ranked alphabet $\Sigma$, the reduction system $(\T_\Sigma,\succ_\Sigma)$ is terminating where $\succ_\Sigma$ is defined as follows. For every $\xi,\xi' \in \T_\Sigma$, we let $\xi \succ_\Sigma \xi'$ if $\xi$ has the form $\sigma(\xi_1,\ldots,\xi_k)$ and $\xi' = \xi_i$ for some $i \in [k]$. Then $\nf_{\succ_\Sigma}(\T_\Sigma) = \Sigma^{(0)}$. We abbreviate the notion ``proof by well-founded induction on $(\T_\Sigma,\succ_\Sigma)$'' by ``proof by induction on $\T_\Sigma$''. 

As another example, the reduction system $(\C_\Sigma,\succ_{\C_\Sigma})$ is terminating where, for every $c_1, c_2 \in \C_\Sigma$, we let $c_1 \succ_{\C_\Sigma} c_2$ if there exists an elementary context $e =  \sigma(\xi_1,\ldots,\xi_{i-1},z,\xi_{i+1},\ldots,\xi_k)$ such that $c_1 = e[c_2]$. Then $\nf_{\succ_{\C_\Sigma}}(\C_\Sigma) = \{z\}$.

\paragraph{Conventions.}
In the rest of this paper $\B=(B,\oplus,\otimes,\0,\1)$ denotes an arbitrary commutative semifield, if not specified otherwise, and we let $B^{-\0}$ denote the set $B \setminus\{\0\}$. Moreover, we let $\Sigma$ denote an arbitrary ranked alphabet, if not specified otherwise.
  
Whenever we write that we can construct an object, like a weighted tree automaton (cf., e.g., Lemma \ref{lm:wta(r,sim,H)}(3), Theorems \ref{lm:construction-of-simple-det} and \ref{thm:minimization-theorem-new-2}), we assume that the input objects of the construction are effectively given and we can compute freely with them, in particular, with the operations of the weight algebra $\B=(B,\oplus,\otimes,\0,\1)$. Also, we assume that we can decide equality of elements in $B$.


\section{Weighted tree automata}
\label{sec:wta-det}

In this section, we recall the definition of weighted tree automaton from \cite{fulvog24}.
 For each wta $\cA$, we define the transformation monoid of $\cA$, which allows to express in a succinct manner the image of trees of the form $c[\xi]$ for a context $c$ and a tree $\xi$ (cf. Lemma~\ref{lm:hcACchAxi=hAcxi-det}). In Lemma~\ref{lm:properties-hA-of-budet-wta-det-new} we show that the semantics of a bottom-up-deterministic wta can be expressed without the use of the addition of the underlying weight algebra.

\subsection{The basic model}
\label{subs:basic-model-det}

A \emph{weighted tree automaton over $\Sigma$ and~$\B$} (or: $(\Sigma,\B)$-wta, or: wta) is a triple $\cA = (Q,\delta,F)$ where 
\begin{compactitem}
\item $Q$ is a finite nonempty set (\emph{states}) such that $Q \cap \Sigma = \emptyset$,
\item $\delta=(\delta_k\mid k\in\mathbb{N})$ is a family of mappings $\delta_k: Q^k\times \Sigma^{(k)}\times Q \to B$ (\emph{transition mappings})  where we consider $Q^k$ as the set of strings over $Q$ of length $k$, and 
\item $F: Q \rightarrow B$ is a mapping (\emph{root weight vector}).
\end{compactitem}
For each $k \in \mathbb{N}$, we view $Q^k$ as set of strings over $Q$ of length $k$; thus, in particular, $Q^0 = \{\varepsilon\}$ and $\varepsilon$ denotes the empty string.
We say that $\cA$ is
\begin{compactitem}
  \item \emph{bottom-up-deterministic} (for short: bu-deterministic) if for every $k \in \mathbb{N}$, $\sigma \in \Sigma^{(k)}$, and $w \in Q^k$ there exists at most one $q \in Q$ such that $\delta_k(w,\sigma,q) \not = \mathbb{0}$, and
\item  \emph{total} if for every $k \in \mathbb{N}$, $\sigma \in \Sigma^{(k)}$, and $w \in Q^k$ there exists at least one state $q$ such that $\delta_k(w,\sigma,q) \not= \mathbb{0}$.
\end{compactitem}

We recall the initial algebra semantics of a $(\Sigma,\B)$-wta $\cA = (Q,\delta,F)$ .
The {\em vector algebra of $\cA$} is the $\Sigma$-algebra $\V(\cA)=(B^Q,\delta_\cA)$ where, for every $k \in \mathbb{N}$ and $\sigma \in \Sigma^{(k)}$, the $k$-ary operation $\delta_\cA(\sigma): B^Q \times \cdots \times B^Q \to B^Q$ is defined, for every $v_1,\ldots,v_k \in B^Q$ and $q \in Q$, by
\[
\delta_\cA(\sigma)(v_1,\dots,v_k)_q 
  = \bigoplus_{q_1\cdots q_k \in Q^k} \Big(\bigotimes_{i\in[k]} (v_i)_{q_i}\Big) \otimes \delta_k(q_1\cdots q_k,\sigma,q) \enspace.
\]
Thus, in particular, for $\sigma \in \Sigma^{(0)}$, we have $\delta_\cA(\sigma)()_q= \delta_0(\varepsilon,\sigma,q)$, because $Q^0 = \{\varepsilon\}$, $[0]=\emptyset$, and $\bigotimes_{i \in \emptyset} f_\emptyset = \1$ where $f_\emptyset$ denotes the mapping of type $\emptyset \to B$.

Since the $\Sigma$-term algebra $\sfT_\Sigma$ is initial, there exists a unique 
$\Sigma$-algebra homomorphism from $\sfT_\Sigma$ to the vector algebra $\V(\cA)$; we denote this homomorphism by $\h_\cA$. Then, for every $\xi = \sigma(\xi_1,\ldots,\xi_k)$ in $\T_\Sigma$ and $q\in Q$, we have
\begin{align*}
  \h_\cA(\sigma(\xi_1,\ldots,\xi_k))_q &= \h_\cA\big(\ttop_\Sigma(\sigma)(\xi_1,\ldots,\xi_k)\big)_q
  = \delta_\cA(\sigma)\big(\h_\cA(\xi_1),\ldots,\h_\cA(\xi_k)\big)_q\\
                                      &=   \bigoplus_{q_1 \cdots q_k \in Q^k} \Big( \bigotimes_{i\in [k]} \h_\cA(\xi_i)_{q_i}\Big) \otimes \delta_k(q_1\cdots q_k,\sigma,q) \enspace.
  \end{align*}
In the particular case that $k=0$, this means that $\h_\cA(\sigma)_q = \delta_0(\varepsilon,\sigma,q)$ .

The \emph{initial algebra semantics of $\cA$} (for short: the \emph{semantics of $\cA$}), denoted by $\sem{\cA}$, is the weighted tree language $\sem{\cA}: \T_\Sigma \rightarrow B$ defined by  
\begin{equation*}
  \sem{\cA}(\xi) =  \h_\cA(\xi) \cdot F\enspace, 
\end{equation*}
where $\cdot$ is the usual scalar product of two $Q$-vectors over $B$, i.e., $
\sem{\cA}(\xi) = \bigoplus_{q \in Q} \h_\cA(\xi)_q \otimes F_q$.

Let $r: \T_\Sigma \to B$. We say that $r$ is \emph{recognizable over $\Sigma$ and $\B$} (and \emph{bu-deterministically recognizable  over $\Sigma$ and $\B$}) if there exists a $(\Sigma,\B)$-wta $\cA$ (and a bu-deterministic  $(\Sigma,\B)$-wta $\cA$, respectively) such that $\sem{\cA}= r$. We denote by $\Rec(\Sigma,\B)$ (and $\budRec(\Sigma,\B)$) the set of all weighted tree languages which are recognizable over $\Sigma$ and $\B$ (and bu-recognizable recognizable over $\Sigma$ and $\B$, respectively).

  \begin{example}\rm \label{ex:bu-det+total-wta-det} We consider the ranked alphabet $\Sigma = \{\sigma^{(2)}, \alpha^{(0)}\}$ and the field $\Ratnum = (\mathbb{Q},+,\cdot,0,1)$ of rational numbers. Moreover, we define $r: \T_\Sigma \to \mathbb{Q}$ such that, for each $\xi \in \T_\Sigma$, we let
    \[
      r(\xi) = \begin{cases} 2 \cdot 2^{\#_\alpha(\xi)} & \text{ if $\#_\alpha(\xi)$ is even}\\
       3 \cdot 2^{\#_\alpha(\xi)} & \text{ otherwise} \end{cases}
   \]
      where $\#_\alpha(\xi)$ denotes the number of occurrences of $\alpha$ in $\xi$.

      Next we define the $(\Sigma,\Ratnum)$-wta $\cA=(Q,\delta,F)$ with the intention that $\sem{\cA}=r$. 
   We let 
    \begin{compactitem}
    \item $Q = \{e,o\}$, where $e$ and $o$ are standing for {\em e}ven and {\em o}dd, respectively,
    \item $\delta_0(\varepsilon,\alpha,o) = 2$,  $\delta_0(\varepsilon,\alpha,e)=0$ and
      \[
        \delta_2(q_1q_2,\sigma,o) = \begin{cases} 1 & \text{ if $q_1\not= q_2$}\\
          0 & \text{ otherwise }
        \end{cases}
       \ \ \text{ and } \ \
         \delta_2(q_1q_2,\sigma,e) = \begin{cases} 1 & \text{ if $q_1= q_2$}\\
          0 & \text{ otherwise }\enspace,
        \end{cases}
      \]
      \item $F_o=3$ and $F_e=2$.
      \end{compactitem}
      Clearly, $\cA$ is total and bu-deterministic. 
      In Figure~\ref{fig:even-odd-two-three-det} we  represent $\cA$  as a hypergraph.

 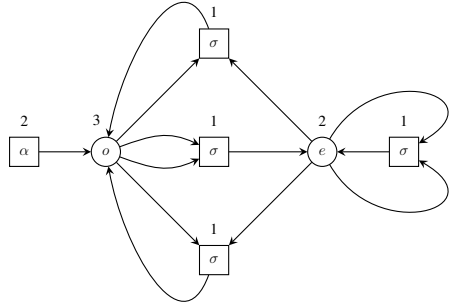
\begin{figure}[t]
   \begin{center}
\begin{tikzpicture}
\tikzset{node distance=7em, scale=0.4, transform shape}
\node[state, rectangle] (1) {\Large $\alpha$};
\node[state, right of=1] (2){\Large $o$};
\node[state, rectangle, right of=2] (3)[right=1em]{\Large $\sigma$};
\node[state, rectangle, above of=3] (4)[above=1em]{\Large $\sigma$};
\node[state, rectangle, below of=3] (5)[below=1em]{\Large $\sigma$};
\node[state, right of=3] (6)[right=1em]{\Large $e$};
\node[state, rectangle, right of=6] (7) {\Large $\sigma$};

\tikzset{node distance=2em}
\node[above of=1] (w1)[yshift=0.7em] {\Large 2};
\node[above of=2] (w2)[left=0.1em, yshift=0.7em] {\Large 3};
\node[above of=3] (w3)[yshift=0.7em] {\Large 1};
\node[above of=4] (w4)[yshift=0.7em] {\Large 1};
\node[above of=5] (w5)[yshift=0.7em] {\Large 1};
\node[above of=6] (w6)[yshift=0.7em] {\Large 2};
\node[above of=7] (w7)[yshift=0.7em] {\Large 1};

\draw[->,>=stealth] (1) edge (2);
\draw[->,>=stealth] (2) edge[out=20, in=155, looseness=1.1] (3);
\draw[->,>=stealth] (2) edge[out=-20, in=205, looseness=1.1] (3);
\draw[->,>=stealth] (2) edge (4);
\draw[->,>=stealth] (2) edge (5);
\draw[->,>=stealth] (3) edge (6);
\draw[->,>=stealth] (6) edge (4);
\draw[->,>=stealth] (4) edge[out=110, in=80, looseness=1.4] (2);
\draw[->,>=stealth] (6) edge (5);
\draw[->,>=stealth] (5) edge[out=250, in=-80, looseness=1.4] (2);
\draw[->,>=stealth] (7) edge (6);
\draw[->,>=stealth] (6) edge[out=60, in=30, looseness=2.7] (7);
\draw[->,>=stealth] (6) edge[out=-60, in=-30, looseness=2.7] (7);
\end{tikzpicture} 
\caption{\label{fig:even-odd-two-three-det} The $(\Sigma,\Ratnum)$-wta $\cA=(Q,\delta,F)$.}
   \end{center}
   \end{figure}
     
      By induction on $\T_\Sigma$, we prove that
      \begin{equation}\label{equ:number-of-alphas-det}
        \begin{aligned}
          &\text{for every $\xi \in \T_\Sigma$ and $q \in Q$, we have:}\\
          &\text{$\h_\cA(\xi)_q = \begin{cases} 2^{\#_\alpha(\xi)} & \text{ if \big($\#_\alpha(\xi)$ is even and $q=e$\big) or
                \big($\#_\alpha(\xi)$ is odd and $q=o$\big)  } \\ 0 & \text{ otherwise} \end{cases}$}
          \end{aligned}
        \end{equation}

        I.B.: Let $\xi = \alpha$. Then $\h_\cA(\xi)_o = \delta_0(\varepsilon,\alpha,o)= 2 = 2^1 = 2^{\#_\alpha(\xi)}$. Moreover, $\h_\cA(\xi)_e= \delta_0(\varepsilon,\alpha,e) = 0$.

        I.S.: Let $\xi = \sigma(\xi_1,\xi_2)$ and $q=o$. Then
        \[
\h_\cA(\xi)_o = \h_\cA(\xi_1)_e \cdot  \h_\cA(\xi_2)_o \cdot \delta_2(eo,\sigma,o)   + \h_\cA(\xi_1)_o \cdot  \h_\cA(\xi_2)_e \cdot \delta_2(oe,\sigma,o) \enspace,
\]
because $\delta_2(ee,\sigma,o)=\delta_2(oo,\sigma,o)=0$.
We proceed by case analysis.

$\#_\alpha(\xi_1)$ is even and $\#_\alpha(\xi_2)$ is even: By the I.H., we have $\h_\cA(\xi_1)_o=\h_\cA(\xi_2)_o=0$, hence $\h_\cA(\xi)_o = 0$.

$\#_\alpha(\xi_1)$ is even and $\#_\alpha(\xi_2)$ is odd:  By the I.H., we have $\h_\cA(\xi_1)_o=0$, hence also by the I.H. we obtain
\[
  \h_\cA(\xi)_o = \h_\cA(\xi_1)_e \cdot  \h_\cA(\xi_2)_o \cdot \delta_2(eo,\sigma,o) = 2^{\#_\alpha(\xi_1)} \cdot 2^{\#_\alpha(\xi_2)} = 2^{\#_\alpha(\xi)} \enspace.
  \]

  $\#_\alpha(\xi_1)$ is odd and $\#_\alpha(\xi_2)$ is even: Again by I.H. we obtain
\[
  \h_\cA(\xi)_o = \h_\cA(\xi_1)_o \cdot  \h_\cA(\xi_2)_e \cdot \delta_2(eo,\sigma,o) = 2^{\#_\alpha(\xi_1)} \cdot 2^{\#_\alpha(\xi_2)} = 2^{\#_\alpha(\xi)} \enspace.
  \]

  $\#_\alpha(\xi_1)$ is odd and $\#_\alpha(\xi_2)$ is odd: Again by I.H. we obtain $\h_\cA(\xi)_o = 0$.
  
 We can prove \eqref{equ:number-of-alphas-det} for the case that $\xi = \sigma(\xi_1,\xi_2)$ and $q=e$ similarly. Then, for each $\xi \in \T_\Sigma$ we have
  \[
\sem{\cA}(\xi) = \bigplus_{q \in Q} \h_\cA(\xi)_q \cdot F_q = \h_\cA(\xi)_e \cdot 2
+ \h_\cA(\xi)_o \cdot 3 = 
\begin{cases} 2 \cdot 2^{\#_\alpha(\xi)} & \text{ if $\#_\alpha(\xi)$ is even}\\
  3 \cdot 2^{\#_\alpha(\xi)} & \text{ otherwise}  \end{cases}
\ \ = \ \ r(\xi) \enspace.\]
Hence $r \in \budRec(\Sigma,\B)$.
\hfill $\Box$  
\end{example}

The next lemma shows that a state $q_0$ of a wta $\cA$ can be dropped without changing the semantics of~$\cA$ if for each tree $\xi\in\T_\Sigma$ we have $\h_\cA(\xi)_{q_0}= \0$. We will use this lemma when dealing with the characterization of  minimality in Subsection \ref{subsect:characterization-of-minimlity}.

\begin{lemma}\label{lm:wta-getting-rid-of-useless-state} \rm
Let $\cA= (Q,\delta,F)$ be a $(\Sigma,\B)$-wta and $q_0 \in Q$ such that, for each $\xi \in \T_\Sigma$, we have $\h_\cA(\xi)_{q_0}=\0$. Moreover, let $\cA_{-q_0}=(Q',\delta',F')$ be the $(\Sigma,\B)$-wta with $Q' = Q \setminus \{q_0\}$, $(\delta')_k = \delta_k \cap \big((Q')^k \times \Sigma^{(k)} \times Q'\big)$ for each $k \in \mathbb{N}$ and $F' = F|_{Q'}$.
Then $\sem{\cA} = \sem{\cA_{-q_0}}$.
\end{lemma}
\begin{proof} First, by induction on $\T_\Sigma$, we can easily show:
  \begingroup
\allowdisplaybreaks
\begin{equation}\label{equ:dropping-a-useless-state-leaves-hA-invariant}
\text{For every $\xi \in \T_\Sigma$ and $q \in Q'$, we have: } \h_{\cA'}(\xi)_q = \h_\cA(\xi)_q\enspace.
\end{equation}
\endgroup
Then, for each $\xi \in \T_\Sigma$, we have:
\begin{align*}
\sem{\cA}(\xi) & = \bigoplus_{q \in Q} \h_\cA(\xi)_q \otimes F_q = \bigoplus_{q \in Q'} \h_\cA(\xi)_q \otimes F_q
\tag{because $\h_\cA(\xi)_{q_0} \otimes F_{q_0} = \0$} \\
& = \bigoplus_{q \in Q'} \h_{\cA'}(\xi)_q \otimes (F')_q
                                                            \tag{by \eqref{equ:dropping-a-useless-state-leaves-hA-invariant} and definition of $F'$}\\
  &= \sem{\cA_{-q_0}}(\xi) \enspace.  
\end{align*}
\end{proof}

  \subsection{Transformation monoid of a wta}
  \label{ssec:transformatio-monoid-det}

  \emph{In this subsection, we let $\cA=(Q,\delta,F)$ be an arbitrary $(\Sigma,\B)$-wta.}

The transformation monoid of the $(\Sigma,\B)$-wta $\cA$ shows how vectors in $B^Q$ are transformed by plugging them into a context. To formalize this ``plugging in'', we define a modification of the mapping $\h_\cA: \T_\Sigma \to B^Q$, which can handle contexts.  First, we define the mapping $g: \e\C_\Sigma \to (B^Q \to B^Q)$ for every $e \in \e\C_\Sigma$ with $e = \sigma(\xi_1,\ldots,\xi_{i-1},z,\xi_{i+1},\ldots,\xi_k)$ as follows. For each  $v \in B^Q$,  we let
  \begin{equation}\label{equ:hAC-on-generator-set-det}
g(e)(v) = \delta_\cA(\sigma)\Big(\h_\cA(\xi_1),\ldots,\h_\cA(\xi_{i-1}),v,\h_\cA(\xi_{i+1}),\ldots,\h_\cA(\xi_k)\Big) \enspace.
\end{equation}
Second, since $(\C_\Sigma,\circ_z,z)$ is free in the set of all monoids with generating set $\e\C_\Sigma$ (cf. Lemma~\ref{lm:freely-generated-context-det}), there exists a unique monoid homomorphism
\(
\h_\cA^\C: \C_\Sigma \to (B^Q \to B^Q)
\)
from $(\C_\Sigma,\circ_z,z)$ to the monoid $(B^Q \to B^Q,\circ,\id_{B^Q})$ which extends $g$. Here $\circ$ denotes the usual composition of mappings.
Thus, for every $c_1,c_2 \in \C_\Sigma$, we have 
\[
  \h_\cA^\C(c_1[c_2]) = \h_\cA^\C(c_1 \circ_z c_2) = \h_\cA^\C(c_1)\circ \h_\cA^\C(c_2) \enspace.
\]
As usual, function application associates to the left, i.e., $\h_\cA^\C(c)(v)$ stands for $(\h_\cA^\C(c))(v)$.

The \emph{transformation monoid of $\cA$}, denoted by $\mathrm{Trans}(\cA)$, is the monoid
\[
\mathrm{Trans}(\cA) = (\im(\h_\cA^\C),\circ,\id_{\im(\h_\cA^\C)}) \enspace.
\]
In a natural way, we can use $\h_\cA^\C$ to express the $\h_\cA$-image of a tree $c[\xi]$ where $c \in \C_\Sigma$ and $\xi \in \T_\Sigma$.
  
\begin{lemma}\rm \label{lm:hcACchAxi=hAcxi-det} For every $c \in \C_\Sigma$ and $\xi \in \T_\Sigma$, we have  $\h_\cA(c[\xi]) = \h_\cA^\C(c)\big(\h_\cA(\xi)\big)$. 
\end{lemma}
\begin{proof} We prove the statement by induction on $(\C_\Sigma,\succ_{\C_\Sigma})$.

  I.B.: Let $c= z$ and $\xi \in \T_\Sigma$. Then $\h_\cA(c[\xi]) = \h_\cA(\xi) = \id_{B^Q}(\h_\cA(\xi)) = \h_\cA^\C(c)(\h_\cA(\xi))$.

  I.S.: Let $c=e[c']$ for some $e \in \e\C_\Sigma$ and $c' \in \C_\Sigma$. We let $e \in \e\C_\Sigma$ \\ with $e= \sigma(\xi_1,\ldots,\xi_{i-1},z,\xi_{i+1},\ldots,\xi_k)$. Then we have
  \begingroup
  \allowdisplaybreaks
  \begin{align*}
   \h_\cA\big((e[c'])[\xi]\big) 
    &= \h_\cA(e[c'[\xi]])
      \tag{by associativity of tree substitution}\\
    &=  \delta_\cA(\sigma)\Big(\h_\cA(\xi_1),\ldots,\h_\cA(\xi_{i-1}),\h_\cA(c'[\xi]),\h_\cA(\xi_{i+1}),\ldots,\h_\cA(\xi_k)\Big) \\
    &= \delta_\cA(\sigma)\Big(\h_\cA(\xi_1),\ldots,\h_\cA(\xi_{i-1}),\h_\cA^\C(c')(\h_\cA(\xi)),\h_\cA(\xi_{i+1}),\ldots,\h_\cA(\xi_k)\Big)
      \tag{by I.H.} \\
    &= \h_\cA^\C(e) \big( \h_\cA^\C(c')(\h_\cA(\xi)) \big)
      \tag{by definition of $\h_\cA^\C$}\\
    &=  \big(\h_\cA^\C(e) \circ  \h_\cA^\C(c')\big)(\h_\cA(\xi)) \big)\\
    &=  \h_\cA^\C(e[c'])(\h_\cA(\xi)) \enspace.
      \tag{because $\h_\cA^\C$ is a monoid homomorphism}
  \end{align*}
  \endgroup
  \end{proof}

  For each $c\in \C_\Sigma$, the mapping $\h_\cA^\C(c)$ is compatible with respect to scalar multiplication.

        \begin{lemma}\rm \label{lm:haC-linear-mapping-det}
          \begin{compactenum}
          \item[(1)] (\cite{boz99}, cf. also \cite[Lm.~3.6.4]{fulvog24}) For every $k \in \mathbb{N}$, $\sigma \in \Sigma^{(k)}$, $v_1,\ldots,v_k \in B^Q$, $b\in B$, $q \in Q$, and $i\in[k]$, we have 
 \[\delta_\cA(\sigma)(v_1,\ldots,v_{i-1},b \cdot v_i,v_{i+1},\ldots,v_k) =
   b \cdot  \delta_\cA(\sigma)(v_1,\ldots,v_{i-1},v_i,v_{i+1},\ldots,v_k).\]
   \item[(2)] For every $c \in \C_\Sigma$, $b\in B$, and $v \in B ^Q$, we have $\h_\cA^\C(c)(b \cdot v)  = b \cdot \h_\cA^\C(c)(v)$. In particular, $\h_\cA^\C(c)(\0^Q) = \0^Q$.
 \end{compactenum}
\end{lemma}
\begin{proof} Proof of (1): The proof follows easily by the definition of $\delta_\cA$ by using that $\B$ is commutative.

\

  Proof of (2):   We prove this by induction on $(\C_\Sigma\succ_{\C_\Sigma})$.
 
    I.B.:  Let $c=z$. Then
\(\h_\cA^\C(c)(b \cdot v)= b \cdot v =b \cdot  \h_\cA^\C(c)(v)\).

I.S.:         Let $c= e[c']$ with $e \in \e\C_\Sigma$ and  $c' \in \C_\Sigma$. Let  $e= \sigma(\xi_1,\ldots,\xi_{i-1},z,\xi_{i+1},\ldots,\xi_k)$. Then we can calculate as follows.
\begingroup
\allowdisplaybreaks
\begin{align*}
  & \h_{\cA}^\C(e[c'])(b \cdot v)\\
  &=  \h_{\cA}^\C(e) \Big( \h_\cA^\C(c')(b \cdot v) \Big)
  \tag{because $\h_\cA^\C$ is a monoid homomorphism}\\
  &=  \delta_{\cA}(\sigma)\Big(\h_{\cA}(\xi_1),\ldots,\h_{\cA}(\xi_{i-1}), \ \h_\cA^\C(c')(b \cdot v) \ ,\h_{\cA}(\xi_{i+1}),\ldots,\h_{\cA}(\xi_k)\Big)
  \tag{by  $\h_{\cA}^\C(e)=g(e)$ and \eqref{equ:hAC-on-generator-set-det}}\\
  &= \delta_{\cA}(\sigma)\Big(\h_{\cA}(\xi_1),\ldots,\h_{\cA}(\xi_{i-1}), \ b \cdot \h_{\cA}^\C(c')(v) \ ,\h_{\cA}(\xi_{i+1}),\ldots,\h_{\cA}(\xi_k)\Big) \tag{by I.H.}\\
  &=  b \cdot \delta_{\cA}(\sigma)\Big(\h_{\cA}(\xi_1),\ldots,\h_{\cA}(\xi_{i-1}), \ \h_{\cA}^\C(c')(v) \ ,\h_{\cA}(\xi_{i+1}),\ldots,\h_{\cA}(\xi_k)\Big)
    \tag{by Statement (1) of the lemma}\\
  &= b \cdot \h_{\cA}^\C(e)\Big(\h_\cA^\C(c')(v)\Big)
  \tag{by \eqref{equ:hAC-on-generator-set-det} and  $\h_{\cA}^\C(e)=g(e)$}\\
  &= b \cdot \h_{\cA}^\C(e[c'])(v) \enspace.
    \tag{because $\h_\cA^\C$ is a monoid homomorphism}
\end{align*}
\endgroup 
\end{proof}

   \subsection{Bu-deterministic wta}\label{subsect:bu-det-wta-new}

 Here we show that the semantics of a bu-deterministic $(\Sigma,\B)$-wta $\cA$ does not use the summation of the underlying weight algebra $\B$. In fact, for the computation of the value $\sem{\cA}(\xi)$ for any $\xi \in \T_\Sigma$, it only uses the values which occur as transition weights and as root weights in $\cA$ and the multiplication $\otimes$ of $\B$.  In the following we will make this statement precise.

    \emph{In this subsection, we let $\cA=(Q,\delta,F)$ be an arbitrary bu-deterministic $(\Sigma,\B)$-wta, if not specified otherwise.}

      We define $\im(\delta)=\bigcup_{k \in \mathbb{N}}\im(\delta_k)$ and we recall that the vector algebra of $\cA$ is the $\Sigma$-algebra $\V(\cA) = (B^Q,\delta_\cA)$.
We also recall that, for each $v\in B^Q$, we denote by $\supp(v)$ the set $\{q\in Q\mid v_q\ne \0\}$ and that, for the element $\0^Q$ of $B^Q$, we have $\supp(\0^Q)=\emptyset$. Then we define
\[B^Q_{=1}=\{v \in B^Q \mid |\supp(v)| = 1\} \ \text{ and } \ \ B^Q_{\le 1}= B^Q_{= 1} \cup \{\0^Q\}\enspace.\] 
For each $v \in B^Q_{= 1}$ we denote by $q_v$ the only state of $\supp(v)$. 

Next we prove that $(B^Q_{\le 1},\delta_\cA)$ is a $\Sigma$-algebra and it is a subalgebra of $\V(\cA)$.

\begin{lemma}\rm \label{lm:properties-sem-of-budet-wta-det-new} The following statements hold.
    \begin{compactenum}
      \item[(1)] For every $k \in \mathbb{N}$, $\sigma \in \Sigma^{(k)}$, $v_1,\ldots,v_k \in B^Q_{\le 1}$, and $q \in Q$, we have
        \[
          \delta_\cA(\sigma)(v_1,\ldots,v_k)_q =
          \begin{cases}
            (v_1)_{q_{v_1}} \otimes \ldots \otimes (v_k)_{q_{v_k}} \otimes  \delta_k(q_{v_1}\cdots q_{v_k},\sigma,q) 
             &\text{ if $(\forall i \in [k]): v_i \in B^Q_{=1}$}\\
            \0 & \ \ \text{ otherwise}
            \end{cases}
          \]
        \item[(2)] $B^Q_{\le 1}$ is closed under the operations in $\delta_\cA(\Sigma)$, i.e.,  $(B^Q_{\le 1},\delta_\cA)$ is a $\Sigma$-algebra and it is a subalgebra of the $\Sigma$-algebra $\V(\cA)$.
        \end{compactenum}
      \end{lemma}

    \begin{proof} 
 Proof of (1):  If $k=0$, then we are done, because the equality in Statement (1) reduces to $\delta_\cA(\sigma)()_q=\delta_0(\varepsilon,\sigma,q)$, and this holds by the definition of $\delta_\cA$.      
 Thus, let $k \in \mathbb{N}_+$, $\sigma \in \Sigma^{(k)}$, $v_1,\ldots,v_k \in B^Q_{\le 1}$, and $q \in Q$. We recall that 
 \begingroup
      \allowdisplaybreaks
      \begin{equation}\label{eq:semantics-recalled}
\delta_\cA(\sigma)(v_1,\ldots,v_k)_q = \bigoplus_{q_1 \cdots q_k \in Q^k} (v_1)_{q_1} \otimes \ldots \otimes (v_k)_{q_k} \otimes \delta_k(q_1 \cdots q_k,\sigma,q)\enspace.
          \end{equation}
        \endgroup     
We continue by case analysis. If, for each $i \in [k]$, we have $v_i \in B^Q_{= 1}$, then $q_{v_i}$ is defined. Hence, the right-hand side of  \eqref{eq:semantics-recalled} is equal to $(v_1)_{q_{v_1}} \otimes \ldots \otimes (v_k)_{q_{v_k}} \otimes \delta_k(q_{v_1} \cdots q_{v_k},\sigma,q)$  because all other members of the sum are $\0$.  Otherwise there exists $i \in [k]$ such that $v_i = \0^Q$.
 Then the right-hand side of  \eqref{eq:semantics-recalled} is equal to $\0$ because $(v_i)_{q_i}=\0$ for each $q_i \in Q$.
 
           \
           
         Proof of (2): We prove by contradiction. For this, let $k \in \mathbb{N}$, $\sigma \in \Sigma^{(k)}$, $v_1,\ldots,v_k \in B^Q_{\le 1}$, and $q_1,q_2 \in Q$ such that $q_1\ne q_2$ and
         \(
\delta_\cA(\sigma)(v_1,\ldots,v_k)_{q_1} \ne \0 \ne \delta_\cA(\sigma)(v_1,\ldots,v_k)_{q_2} 
           \).
           By Statement (1), we have $v_i \in B^Q_{=1}$ for each $i \in [k]$, and thus
           \begin{align*}
             & \delta_\cA(\sigma)(v_1,\ldots,v_k)_{q_1} = (v_1)_{q_{v_1}} \otimes \ldots \otimes (v_k)_{q_{v_k}} \otimes  \delta_k(q_{v_1}\cdots q_{v_k},\sigma,q_1) \ne \0 \text{ and}\\
             & \delta_\cA(\sigma)(v_1,\ldots,v_k)_{q_2} = (v_1)_{q_{v_1}} \otimes \ldots \otimes (v_k)_{q_{v_k}} \otimes  \delta_k(q_{v_1}\cdots q_{v_k},\sigma,q_2)  \ne \0.
             \end{align*}
This implies $\delta_k(q_{v_1}\cdots q_{v_k},\sigma,q_1) \ne \0 \ne \delta_k(q_{v_1}\cdots q_{v_k},\sigma,q_2)$.   Since $\cA$ is bu-deterministic, we obtain $q_1 = q_2$. This is a contradiction, hence $\delta_\cA(\sigma)(v_1,\ldots,v_k) \in B^Q_{\le 1}$. Thus $(B^Q_{\le 1},\delta_\cA)$ is a $\Sigma$-algebra. In fact, it is a subalgebra of the vector algebra $\V(\cA)$ of $\cA$.
\end{proof}

Now we consider the set $\im(\h_\cA)$ of images of $\h_\cA$. In order to identify, for each $\h_\cA(\xi) \in B_{=1}^Q$, the component  which is non-zero, we introduce state algebras.
  We begin with defining an auxiliary concept.
Let $k \in \mathbb{N}$, $\sigma \in \Sigma^{(k)}$, and $w \in Q^k$. We define the set of successor states by
\[
\mathrm{succ}_\cA(w,\sigma) = \{p \in Q \mid \delta_k(w,\sigma,p) \ne \0\} \enspace.
\]
Since $\cA$ is bu-deterministic, we have $|\mathrm{succ}_\cA(w,\sigma)| \le 1$. 
If $\mathrm{succ}_\cA(w,\sigma)$ is a singleton, then we denote its only element by $\mathrm{esucc}_\cA(w,\sigma)$.

Let $Q_\bot=Q \cup \{\bot\}$ where $\bot$ is an element not in $Q$.  The {\em state algebra of~$\cA$}, denoted by $\St(\cA)$,  is the $\Sigma$-algebra $\St(\cA)=(Q_\bot,\theta_\cA)$ such that,
for every  $k \in \mathbb{N}$, $\sigma \in \Sigma^{(k)}$, and $q_1,\ldots,q_k \in Q_\bot$, we define
\begin{align*}
  \theta_\cA(\sigma)(q_1,\ldots,q_k)=
  \begin{cases}
    \mathrm{esucc}_\cA(q_1 \cdots q_k,\sigma)  & \text{ if $q_1,\ldots,q_k \in Q$ and $|\mathrm{succ}_\cA(q_1 \cdots q_k,\sigma)|=1$}\\
    \bot & \text{ otherwise} \enspace.
  \end{cases}
  \end{align*}
We denote by $\state_\cA$ the unique $\Sigma$-algebra homomorphism from $\sfT_\Sigma$ to $\St(\cA)$. If $\cA$ is clear from the context, then we abbreviate $\mathrm{succ}_\cA$, $\mathrm{esucc}_\cA$, and $\state_\cA$ by $\mathrm{succ}$, $\mathrm{esucc}$, and $\state$, respectively.

\begin{example}\rm \label{ex:for-state-algebra-new} Let $\Sigma = \{\alpha^{(0)},\beta^{(0)}\}$ and $\cA=(Q,\delta,F)$ be the bu-deterministic $(\Sigma,\Ratnum)$-wta  with $Q=\{p_1,p_2\}$, and  $\delta_0(\varepsilon,\alpha,p_1)=1$,  $\delta_0(\varepsilon,\alpha,p_2)=\delta_0(\varepsilon,\beta,p_1)=\delta_0(\varepsilon,\beta,p_2)=0$, and $F_{p_1} = F_{p_2}=1$. Then the state algebra $\St(\cA) = (Q_\bot,\theta_\cA)$ has the carrier set $Q_\bot= \{p_1,p_2,\bot\}$ and the nullary operations $\theta_\cA(\alpha)$ and $\theta_\cA(\beta)$ with
  $\theta_\cA(\alpha)()=p_1$ and $\theta_\cA(\beta)()=\bot$.
  \hfill $\Box$
\end{example}

 \begin{lemma}\rm \label{lm:properties-hA-of-budet-wta-det-new} Let $\xi \in \T_\Sigma$. The following four statements hold.
    \begin{compactenum}
    \item[(1)] $\h_\cA(\xi) \in B^Q_{\le 1} \cap \big(\langle \im(\delta) \rangle_{\{\otimes\}} \big)^Q$.

      \item[(2)] For each $q \in Q$, we have $\h_\cA(\xi)_q \not= \0$ if and only if $q=\state(\xi)$.
      \item[(3)] 
    $\h_\cA(\xi) = \begin{cases}
     \h_\cA(\xi)_{\state(\xi)} \cdot \1_{\state(\xi)} & \text{ if $\state(\xi)\in Q$} \\
     \0^Q & \text{ otherwise.}
    \end{cases}
    $
       \item[(4)] 
       $\sem{\cA}(\xi)= \begin{cases}
         \h_\cA(\xi)_{\state(\xi)} \otimes F_{\state(\xi)} & \text{ if $\state(\xi)\in Q$} \\
         \0 & \text{ otherwise,}
    \end{cases}
  $ 
  
  \noindent and thus, in particular, $\sem{\cA}(\xi) \in \langle \im(\delta) \rangle_{\{\otimes\}} \otimes \im(F)$.         
      \end{compactenum}
    \end{lemma}

          \begin{proof}
   Proof of (1): We prove this by induction on $\T_\Sigma$. Let $\xi = \sigma(\xi_1,\ldots,\xi_k)$. Then
             \begin{align*}
\h_\cA(\sigma(\xi_1,\ldots,\xi_k)) &= \delta_\cA(\sigma)\big(\h_\cA(\xi_1),\ldots, \h_\cA(\xi_k)\big) \enspace.
             \end{align*}
             By I.H. we have $\h_\cA(\xi_i) \in B^Q_{\le 1} \cap \big(\langle \im(\delta) \rangle_{\{\otimes\}} \big)^Q$ for each $i \in [k]$. Thus, the statement follows from Lemma~\ref{lm:properties-sem-of-budet-wta-det-new}.       

             \

             Proof of (2): We prove this
                          by induction on $\T_\Sigma$. Let $\xi = \sigma(\xi_1,\ldots,\xi_k)$ and $q\in Q$. We use ``$\exists !$'' to state that there exists exactly one object.
  \begingroup
  \allowdisplaybreaks
  \begin{align*}
    & \h_\cA(\xi)_q \not= \0\\
    \Leftrightarrow \ \ &  \bigoplus_{q_1 \cdots q_k \in Q^k} \Big(\bigotimes_{i \in [k]} \h_\cA(\xi_i)_{q_i}\Big) \otimes \delta_k(q_1 \cdots q_k,\sigma,q) \not= \0
                          \tag{by definition}\\[3mm]
    \Leftrightarrow \ \ &  (\exists!q_1\cdots q_k\in Q^k): \Big(\bigotimes_{i \in [k]} \h_\cA(\xi_i)_{q_i}\Big) \otimes \delta_k(q_1 \cdots q_k,\sigma,q) \not= \0
    \tag{by Statement (1)}\\[3mm]
    \Leftrightarrow \ \ &  (\exists!q_1\cdots q_k\in Q^k): \big(  \h_\cA(\xi_1)_{q_1}\ne \0 \wedge \ldots  \wedge \h_\cA(\xi_k)_{q_k}\ne \0 \wedge  \delta_k(q_1\cdots q_k,\sigma,q)\ne\0 \big) \tag{for $\Leftarrow$ we use that $\B$ is zero-divisor free}\\[2mm]
  \Leftrightarrow \ \ &  (\exists!q_1\cdots q_k\in Q^k): \big(  q_1=\state(\xi_1) \wedge \ldots  \wedge q_k=\state(\xi_k) \wedge  \delta_k(q_1\cdots q_k,\sigma,q)\ne\0 \big) \tag{by I.H.}\\[2mm]  
  \Leftrightarrow \ \ &  (\exists!q_1\cdots q_k\in Q^k): \big(  q_1=\state(\xi_1) \wedge \ldots  \wedge q_k=\state(\xi_k) \wedge  \theta_\cA(\sigma)(q_1,\ldots, q_k)=q \big) \tag{by the definition of $\theta_\cA$}\\[2mm]
    \Leftrightarrow \ \ & \state( \sigma(\xi_1,\ldots,\xi_k)) = q \tag{because $\state$ is a $\Sigma$-algebra homomorphism}\enspace.
    \end{align*}
    \endgroup
    
    \ 
    
    Proof of (3): Let $\state(\xi)\in Q$. By Statement (2) we obtain that, for each $q \in Q$ with $q \not= \state(\xi)$, we have $\h_\cA(\xi)_q=\0$. Hence $\h_\cA(\xi) = \h_\cA(\xi)_{\state(\xi)} \cdot \1_{\state(\xi)}$. 
    
   Next, let  $\state(\xi)=\bot$. Again by Statement (2) we obtain $\h_\cA(\xi)=\0^Q$.

\
    
  Proof of (4): By definition, we have 
\begingroup
      \allowdisplaybreaks
      \begin{align*}
        \sem{\cA}(\xi) = \bigoplus_{q \in Q} \h_\cA(\xi)_q \otimes F_q 
                         \end{align*}
            \endgroup
            
If $\state(\xi)\in Q$, then, by Statement (2), the above sum reduces to $\h_\cA(\xi)_{\state(\xi)} \otimes F_{\state(\xi)}$.
The  membership $\sem{\cA}(\xi)\in \langle \im(\delta)\rangle_{\{\otimes\}} \otimes \im(F)$ follows straightforwardly from Statement (1). 

Otherwise,   by Statement (2),   $\h_\cA(\xi)_q=\0$ for each $q\in Q$ and thus the above sum is equal to $\0$.   Moreover, by Statement (1), we have $\0 \in \langle \im(\delta) \rangle_{\{\otimes\}}$, and the membership also follows.    
            \end{proof}

By Lemma \ref{lm:hcACchAxi=hAcxi-det}, for every $c \in \C_\Sigma$ and $\xi \in \T_\Sigma$, we have  $\h_\cA(c[\xi]) = \h_\cA^\C(c)\big(\h_\cA(\xi)\big)$. Since $\cA$ is bu-deterministic, the value $\h_\cA(c[\xi])_q$ (for any state $q \in Q$) can be calculated as the product of $\h_\cA(\xi)_{\state(\xi)}$ and $\h_\cA^\C(c)(\1_{\state(\xi)})_q$ if $\state(\xi) \in Q$.

  \begin{lemma}\rm \label{obs:total-bu-det-wta-calc(new)-det-new} For every  $\xi \in \T_\Sigma$, $c \in \C_\Sigma$, and $q \in Q$, we have 
    \[
      \h_\cA(c[\xi])_q = \begin{cases}
        \h_\cA(\xi)_{\state(\xi)} \otimes \h_\cA^\C(c)(\1_{\state(\xi)})_q& \text{ if $\state(\xi) \in Q$}\\
                        \0 & \text{ otherwise} \enspace.
                      \end{cases}
                    \]
                    \end{lemma}
 
 \begin{proof} Let $\xi \in \T_\Sigma$, $c \in \C_\Sigma$, and $q \in Q$. Then
   \begingroup
   \allowdisplaybreaks
   \begin{align*}
     \h_\cA(c[\xi])_q &=  \,\h_\cA^\C(c)(\h_\cA(\xi))_q
                        \tag{by Lemma \ref{lm:hcACchAxi=hAcxi-det}} \\
                     &=  \, \begin{cases}
                       \h_\cA^\C(c)(\h_\cA(\xi)_{\state(\xi)}\cdot\,\1_{\state(\xi)})_q & \text{ if $\state(\xi) \in Q$}\\
                        \h_\cA^\C(c)(\0^Q)_q & \text{ otherwise}
                       \end{cases}
                       \tag{by Lemma \ref{lm:properties-hA-of-budet-wta-det-new}(3)}\\
                      &=  \, \begin{cases}
                        \big(\h_\cA(\xi)_{\state(\xi)} \cdot \h_\cA^\C(c)(\1_{\state(\xi)})\big)_q& \text{ if $\state(\xi) \in Q$}\\
                        \0 & \text{ otherwise}
                       \end{cases}
                        \tag{by Lemma \ref{lm:haC-linear-mapping-det}(2)}\\
                      &=  \, \begin{cases}
                        \h_\cA(\xi)_{\state(\xi)} \otimes \h_\cA^\C(c)(\1_{\state(\xi)})_q& \text{ if $\state(\xi) \in Q$}\\
                        \0 & \text{ otherwise}
                      \end{cases}
                     \tag{by definition of the scalar multiplication}   \enspace.
   \end{align*}
   \endgroup
   \end{proof}

Statements (1) and (4) of Lemma \ref{lm:properties-hA-of-budet-wta-det-new} show that the addition $\oplus$ of $\B$ does not play a role in the semantics of any bu-deterministic wta. This is also expressed in the following corollary.

 \begin{corollary}\label{cor:budetwta-addition-irrelevant-det-new} \rm
  Let $\B_1 = (B,\oplus,\otimes,\0,\1)$ and $\B_2 = (B,+,\otimes,\0,\1)$ be commutative semifields and $\cA=(Q,\delta,F)$ be a bu-deterministic $(\Sigma,\B_1)$-wta. Then, for the bu-deterministic $(\Sigma,\B_2)$-wta $\cB=(Q,\delta,F)$, we have $\sem{\cB}=\sem{\cA}$.
\end{corollary}

 In \cite[Cor.~4.1.4(2)]{fulvog24}, the above statement was proved even for strong bimonoids; roughly speaking,  such algebras are semirings in which distributivity is not required; each semiring (and hence, each semifield) and each bounded lattice  is a strong bimonoid but not vice versa (cf., e.g., \cite[Sec.~2.7.5 and 2.7.6]{fulvog24}).

Finally we show two examples of pairs $\B_1 = (B,\oplus,\otimes,\0,\1)$ and $\B_2 = (B,+,\otimes,\0,\1)$ of commutative semifields:
  \begin{compactitem}
  \item the Boolean semiring $(\{0,1\},\vee,\wedge,0,1)$ with disjunction $\vee$ (e.g., $1 \vee 1 = 1$) and conjunction $\wedge$; and the field $\sfFtwo =(\{0,1\},+,\cdot,0,1)$ in which $1 + 1 = 0$ holds;
    \item $(\mathbb{Q}_{\ge 0},\max,\cdot,0,1)$ and $(\mathbb{Q}_{\ge 0},+,\cdot,0,1)$ with the usual maximum $\max$, addition +, and multiplication $\cdot$ of non-negative rational numbers.
    \end{compactitem}


\section{B-A's theorem}
\label{sect:B-A-BLs-result-det}

In order to ease the comparison of our first main result with B-A's theorem from \cite{bozale89}, we recall some notion and notation and B-A's theorem in more detail.
We let $\mathrm{Pol}(\Sigma,\B)$ denote the set of all polynomial $(\Sigma,\B)$-weighted tree languages for some field $\B$.
We represent each polynomial $s\in \mathrm{Pol}(\Sigma,\B)$ with $\supp(s) = \{\xi_1,\ldots,\xi_n\}$ as $s=b_{1}.\xi_{1} \oplus\cdots\oplus b_{n}.\xi_{n}$ where $b_i=s(\xi_i)$ for each $i\in [n]$.
The scalar multiplication for polynomials is defined such that, for every $b\in B$ and $s\in \mathrm{Pol}(\Sigma,\B)$,  we let $(b\cdot s)(\xi)=b\otimes s(\xi)$ for each $\xi \in \T_\Sigma$. Moreover, the sum $s_1\oplus s_2$ of two polynomials $s_1,s_2\in \mathrm{Pol}(\Sigma,\B)$ is defined by $(s_1\oplus s_2)(\xi)=s_1(\xi)\oplus s_2(\xi)$ for each $\xi\in \T_\Sigma$. It is easy to see that $\mathsf{Pol}(\Sigma,\B)=(\Pol(\Sigma,\B),\oplus,\widetilde{\0})$ is a $\B$-vector space via $\cdot$.

For every $k\in \mathbb{N}$ and $\sigma \in \Sigma^{(k)}$, we define the \emph{$\sigma$-top-concatenation (of polynomial $(\Sigma,\B)$-weighted tree languages)}, denoted by  $\ttop(\sigma)$, to be the following $k$-ary operation on $\Pol(\Sigma,\B)$. Let $s_1,\ldots,s_k\in \Pol(\Sigma,\B)$. For each $i\in [k]$, let $s_i=b_{i1}.\xi_{i1} \oplus\cdots\oplus b_{in_i}.\xi_{in_i}$ for some  $n_i\in\mathbb{N}$, $b_{i1},\ldots, b_{in_i}\in B$, and $\xi_{i1},\ldots, \xi_{in_i}\in \T_\Sigma$. Then we define
\begin{equation} \label{eq:definition-of-top-concatenation} \ttop(\sigma)(s_1,\ldots,s_k)=\bigoplus_{j_1\in[n_1],\ldots,j_k\in[n_k]} (b_{1j_1}\otimes\ldots\otimes b_{kj_k}).\sigma(\xi_{1j_1},\ldots,\xi_{kj_k})\enspace.
\end{equation}
In particular,  $\ttop(\alpha)(\,)=\1.\alpha$ for each $\alpha\in \Sigma^{(0)}$. Thus $(\Pol(\Sigma,\B),\ttop)$ is a $\Sigma$-algebra. The algebra  $(\mathrm{Pol}(\Sigma,\B),\oplus,\widetilde{\0},\ttop)$ is called the $(\Sigma,\B)$-vector space of polynomial $(\Sigma,\B)$-weighted tree languages (cf. \cite{bozale89} and \cite{fulste11}).

For a weighted tree language $r: \T_\Sigma \to B$, the congruence $\approx_r$  on $(\mathrm{Pol}(\Sigma,\B),\oplus,\widetilde{\0},\ttop)$, called the syntactic congruence of $r$,  is defined as follows: for every $s_1 = b_1.\xi_1 \oplus \ldots \oplus b_m.\xi_m$ and $s_2=a_1.\zeta_1\oplus\ldots\oplus a_n.\zeta_n$ in $\mathrm{Pol}(\Sigma,\B)$, we let  $s_1 \approx_r s_2$ if for every $c \in \C_\Sigma$, we have
      \[
        b_1 \otimes r(c[\xi_1]) \oplus \ldots \oplus b_m \otimes r(c[\xi_m]) = a_1 \otimes r(c[\zeta_1]) \oplus \ldots \oplus a_n \otimes r(c[\zeta_n])  \enspace.
      \]
    We call the quotient algebras $(\mathrm{Pol}(\Sigma,\B),\oplus,\widetilde{\0},\ttop)/_{\approx_r}$ and 
  $\mathsf{Pol}(\Sigma,\B)/_{\approx_r}$ the \emph{syntactic $(\Sigma,\B)$-vector space} and the \emph{syntactic $\B$-vector space (of $r$)}, respectively. For more details we refer to \cite[Sec.~18.4]{fulvog24} (in \cite[Thm.~18.4.1]{fulvog24} $\approx_r$ is denoted by $\ker(\Phi_r)$).
  
  Let $\approx$ be a congruence on  $(\mathrm{Pol}(\Sigma,\B),\oplus,\widetilde{\0},\ttop)$ and $r: \T_\Sigma \to B$. We say that $\approx$ \emph{saturates $r$} if there exists a linear form $\gamma: \Pol(\Sigma,\B)/_{\approx} \to B$ such that, for each $\xi \in \T_\Sigma$, we have $r(\xi) = \gamma([\1.\xi]_{\approx})$.

    \begin{theorem} \label{th:MN-fields-det} {\rm \cite[Prop.~2]{bozale89}  (cf. also \cite[Thm.~18.4.1]{fulvog24})} Let $\B$ a field. Moreover, let  $r: \T_\Sigma \to B$. Then the following three statements are equivalent.
  \begin{compactenum}
  \item[(A)] $r \in \Rec(\Sigma,\B)$.
  \item[(B)] There exists a congruence $\approx$ on the $(\Sigma,\B)$-vector space $(\mathrm{Pol}(\Sigma,\B),\oplus,\widetilde{\0},\ttop)$ such that $\approx$ respects~$r$, and the $\B$-vector space $\mathsf{Pol}(\Sigma,\B)/_\approx$ is finite-dimensional.
  \item[(C)] The $\B$-vector space $\mathsf{Pol}(\Sigma,\B)/_{\approx_r}$ is finite-dimensional.
    \end{compactenum}
  \end{theorem}

  \begin{example}\rm \label{ex:wtl-number-of-occurrences-of-gamma-finite-basis} We give an example of a ranked alphabet $\Sigma$ and a $(\Sigma,\Ratnum)$-weighted tree language $r$ for which the $\Ratnum$-vector space $\sfPol(\Sigma,\Ratnum)/_{\approx_r}$ is finite-dimensional; hence, by Theorem~\ref{th:MN-fields-det}(C)$\Rightarrow$(A), $r$ is recognizable.

Let  $\Sigma=\{\sigma^{(2)},\gamma^{(1)},\alpha^{(0)}\}$ and $r=\#_\gamma$ where $\#_\gamma: \T_\Sigma \to \mathbb{Q}$ is the mapping defined such that $\#_\gamma(\xi)$ is the number of  the occurrences of the symbol $\gamma$ in $\xi$ for each $\xi \in\T_\Sigma$ (cf. \cite[Ex.~18.1.5]{fulvog24}).
In the following we abbreviate $\approx_{\#_\gamma}$ by  $\approx_{\gamma}$.
 
First, for each $(\Sigma,\Ratnum)$-polynomial weighted tree language
  \(s = b_1 . \xi_1 + \ldots + b_n . \xi_n\), we define $\mathrm{sumc}(s)\in \mathbb{Q}$ and $\mathrm{sum}_\gamma(s)\in \mathbb{Q}$ by
  \[
    \mathrm{sumc}(s) = \bigplus_{i \in [n]} b_i  \ \ \text{ and } \ \ \mathrm{sum}_\gamma(s) = \bigplus_{i \in [n]} b_{i}\cdot\#_\gamma(\xi_{i})  \enspace,
  \]
  respectively. In particular, $\mathrm{sumc}(\widetilde{0})=\mathrm{sum}_\gamma(\widetilde{0})=0$. Then we prove the following characterization of $\approx_{\gamma}$:
\begin{equation}\label{equ:form-of-ker-Varphi-for-num-gamma}
    \begin{aligned}
    &\text{for every $s_1,s_2 \in \Pol(\Sigma,\Ratnum)$, we have $s_1 \approx_{\gamma} s_2$ if and only if}\\
    &\mathrm{sumc}(s_1) = \mathrm{sumc}(s_2)\ \text{ and } \ \mathrm{sum}_\gamma(s_1) = \mathrm{sum}_\gamma(s_2)
      \end{aligned}
\end{equation}
For this, let
\begin{align*}
s_1 = b_{11}.\xi_{11} + \ldots + b_{1n_1}.\xi_{1n_1} \ \ \text{ and } \ \
s_2 = b_{21}.\xi_{21} + \ldots + b_{2n_2}.\xi_{2n_2}
\end{align*}
be two $(\Sigma,\Ratnum)$-polynomial weighted tree languages. Then we can calculate as follows (by using the obvious extension of $\#_\gamma$ to contexts). 
\begingroup
\allowdisplaybreaks
\begin{align*}
  &s_1 \approx_{\gamma} s_2\\[2mm]
  \text{iff } \ & (\forall c \in \C_\Sigma): b_{11} \cdot \#_\gamma(c[\xi_{11}]) + \ldots + b_{1n_1} \cdot \#_\gamma(c[\xi_{1n_1}])\\
 & \hspace*{15mm} = b_{21} \cdot \#_\gamma(c[\xi_{21}]) + \ldots + b_{2n_2} \cdot \#_\gamma(c[\xi_{2n_2}]) \\[2mm]
  \text{iff } \ & (\forall c \in \C_\Sigma): b_{11} \cdot (\#_\gamma(c) + \#_\gamma(\xi_{11})) + \ldots + b_{1n_1} \cdot (\#_\gamma(c) + \#_\gamma(\xi_{1n_1}))\\
  & \hspace*{15mm} = b_{21} \cdot (\#_\gamma(c) + \#_\gamma(\xi_{21})) + \ldots + b_{2n_2} \cdot (\#_\gamma(c) + \#_\gamma(\xi_{2n_2})) \\[2mm]
  \text{iff } \ & (\forall c \in \C_\Sigma): \mathrm{sumc}(s_1) \cdot \#_\gamma(c) + \mathrm{sum}_\gamma(s_1) 
                  = \mathrm{sumc}(s_2) \cdot \#_\gamma(c) + \mathrm{sum}_\gamma(s_2)\\[2mm]
  \text{iff } \ & (\forall n \in \mathbb{N}): f(n) = g(n)
                  \tag{where $f(n)= \mathrm{sumc}(s_1) \cdot n + \mathrm{sum}_\gamma(s_1)$ and
                  $g(n) = \mathrm{sumc}(s_2) \cdot n + \mathrm{sum}_\gamma(s_2)$}\\[2mm]
  \text{iff } \ & \mathrm{sumc}(s_1)  = \mathrm{sumc}(s_2) \ \text{ and } \
                  \mathrm{sum}_\gamma(s_1) =  \mathrm{sum}_\gamma(s_2)\enspace,
  \end{align*}
  \endgroup
  where the last equivalence holds, because the two linear mappings $f$ and $g$ are equal if and only if their two parameters coincide pairwise. This proves \eqref{equ:form-of-ker-Varphi-for-num-gamma}.
  
  Now we let $H=\{[1.\alpha]_{\approx_{\gamma}},[1.\gamma(\alpha)]_{\approx_{\gamma}}\}$. We show that $H$ generates   $\sfPol(\Sigma,\Ratnum)/_{\approx_{\gamma}}$. 
     More precisely, for each $s\in \Pol(\Sigma,\Ratnum)$, we prove that the following equivalence holds:
  \begin{equation}\label{eq:decomposition-into-basis}
  [s]_{\approx_{\gamma}}=a_1\cdot [1.\alpha]_{\approx_{\gamma}} + a_2\cdot [1.\gamma(\alpha)]_{\approx_{\gamma}} \ \  \text{ iff } \ \
  a_1= \mathrm{sumc}(s) - \mathrm{sum}_\gamma(s) \text{ and } a_2= \mathrm{sum}_\gamma(s)
  \end{equation}
  We calculate as follows:
  \begin{align*}
  & [s]_{\approx_{\gamma}}=a_1\cdot [1.\alpha]_{\approx_{\gamma}} + a_2\cdot [1.\gamma(\alpha)]_{\approx_{\gamma}}\\
  \text{iff } \ & s \approx_{\gamma} a_1.\alpha + a_2.\gamma(\alpha) \tag{because $\approx_{\gamma}$ is a congruence} \\
   \text{iff } \ & \mathrm{sumc}(s) = a_1 + a_2 \text{ and } \mathrm{sum}_\gamma(s) = a_2 \tag{by \eqref{equ:form-of-ker-Varphi-for-num-gamma}}\\
   \text{iff } \ & a_1= \mathrm{sumc}(s) - \mathrm{sum}_\gamma(s) \text{ and } a_2= \mathrm{sum}_\gamma(s).
  \end{align*}
  This proves \eqref{eq:decomposition-into-basis}.
 Hence  the $\Ratnum$-vector space $\sfPol(\Sigma,\Ratnum)/_{\approx_{\gamma}}$ is finitely generated.  Finally, we show that $[1.\alpha]_{\approx_{\gamma}}$ and $[1.\gamma(\alpha)]_{\approx_{\gamma}}$ are linearly independent,  i.e., that $H$ is a basis.
 For this, assume that 
 \(b_1\cdot [1.\alpha]_{\approx_{\gamma}} + b_2\cdot [1.\gamma(\alpha)]_{\approx_{\gamma}}=[\widetilde{0}]_{\approx_{\gamma}}\) for some $b_1,b_2\in \mathbb{Q}$.
 Using \eqref{equ:form-of-ker-Varphi-for-num-gamma}, we obtain that $b_1+b_2=0$ and $b_2=0$, i.e., that $b_1=b_2=0$.  
 \hfill $\Box$
\end{example}


\section{Scalar algebras}
\label{subsect:mono-spaces-det}

\sloppy B-A's theorem characterizes a recognizable
$(\Sigma,\B)$-weighted tree language $r$ in terms of finite-dimensionality of a particular quotient of the $\B$-vector space of polynomials, i.e., the quotient $\mathsf{Pol}(\Sigma,\B)/_{\approx_r}$. When analysing the proof of Theorem \ref{th:MN-fields-det} under the assumption that $r$ is recognized by some bu-deterministic wta, then one realizes that the addition  of the involved $\B$-vector space is not needed. 
This corresponds to Corollary \ref{cor:budetwta-addition-irrelevant-det-new}, which tells us that the summation $\oplus$ of the semifield $\B$ is irrelevant when we deal with bu-deterministic $(\Sigma,\B)$-wta. As a consequence, it is likely to obtain a B-A-like theorem for bu-deterministic $(\Sigma,\B)$-wta by considering a set of algebras which are similar to $(\Sigma,\B)$-vector spaces but from which the addition $+$ is dropped. For this, we propose the concept of $(\Sigma,\B)$-scalar algebra.

\subsection{General concept}
\label{subsect:scalar-algebras-det}

For the formal definition of $(\Sigma,\B)$-scalar algebras, first we define $\B$-scalar algebras.

Let $\V=(V,0)$, where $V$ is a set and $0\in V$. Moreover, let $\cdot: B\times V \rightarrow V$ be a mapping, called \emph{scalar multiplication} such that the following laws hold for every $b,b' \in B$ and $v \in V$:
    \begin{eqnarray}
&(b \otimes b') \cdot v = b \cdot (b' \cdot v)  \label{SM1-det}\\
&\1 \cdot v = v  \label{SM4-det} \\
&      b \cdot 0 = \0 \cdot v = 0 \label{SM5-det} \enspace.
\end{eqnarray}
    Then we call $\V$ a \emph{$\B$-scalar algebra (via the scalar multiplication $\cdot$)}.
    
We note that $0$ is the unique element of $V$ for which (\ref{SM5-det}) holds. To see this, let $0' \in V$ for which (\ref{SM5-det}) holds. Then, for every $v\in V$, we have $0'=\0 \cdot v= 0$. 
In the sequel we will always assume that the scalar multiplication is denoted by $\cdot$ if not specified otherwise. In particular, $(B,\0)$ is a $\B$-scalar algebra with scalar multiplication $b\cdot v=b\otimes v$ for every $b,v\in B$.

The concept of $\B$-scalar algebra is similar to that of group action, which is well known in algebra  (cf., e.g., \cite{ker99}, \cite{hum04}, and \cite{lan12}).

\begin{observation}\rm \label{obs:product=0-implies-vector=0-det} Let $\V=(V,0)$ be a $\B$-scalar algebra.  For each $v \in V\setminus\{0\}$ and $b \in B$, if $b \cdot v = 0$, then $b=\0$.
\end{observation}
\begin{proof} We prove by contradiction.  Assume that there exist  $v \in V\setminus\{0\}$ and $b \in B$ such that $b \cdot v = 0$ and $b\ne\0$. Then $b^{-1} \cdot (b \cdot v)=b^{-1}\cdot 0 =0$. On the other hand, $b^{-1} \cdot (b \cdot v)=(b^{-1}\otimes b) \cdot v= \1\cdot v =v \ne 0$, a contradiction.
\end{proof}

A $\B$-scalar algebra  $\V=(V,0)$ is \emph{cancellative} if, for every $b_1,b_2\in B$ and $v \in V\setminus \{0\}$ we have that $b_1 \cdot v = b_2\cdot v$ implies $b_1=b_2$.
Trivially, if $V=\{0\}$, then $\V$ is cancellative. Moreover, the $\B$-scalar algebra $(B,\0)$ is cancellative, because $\B$ is a semifield.

We note that, for each $\B$-vector space $(V,+,0)$ with field $\B$ and scalar multiplication $\cdot$ \cite{axl24,mooyaq98}, the reduct $(V,0)$ is a cancellative $\B$-scalar algebra with the same scalar multiplication. The cancellativity follows from the fact that, in this case, $b_1 \cdot v = b_2\cdot v$ implies that $(b_1+(-b_2))\cdot v =0$, where $-b_2$ is the additive inverse of $b_2$. By Observation \ref{obs:product=0-implies-vector=0-det}, the latter implies that $b_1=b_2$.

A \emph{$(\Sigma,\B)$-scalar algebra} is a triple $(V,0,\mu)$ where
\begin{compactitem}
\item $(V,0)$ is a $\B$-scalar algebra and 
\item  $(V,\mu)$ is a $\Sigma$-algebra such that, for every $k \in \mathbb{N}_+$, $\sigma \in \Sigma^{(k)}$, $i \in [k]$, $b \in B$, and $v,v_1,\ldots,v_k \in V$, we have
\begin{eqnarray}
  \begin{aligned}
  &\mu(\sigma)\big(v_1,\ldots,v_{i-1},\ b \cdot v \ ,v_{i+1},\ldots,v_k\big) = 
  b \cdot \mu(\sigma)\big(v_1,\ldots,v_{i-1},v,v_{i+1},\ldots,v_k\big) \enspace.
    \end{aligned}\label{ml-Omega-det}
\end{eqnarray}
\end{compactitem}
In particular, (\ref{ml-Omega-det})  (with $b=\0$) and (\ref{SM5-det}) imply that, for every $k \in \mathbb{N}_+$, $\sigma \in \Sigma^{(k)}$, $i \in [k]$,  and $v_1,\ldots,v_k \in V$, we have
\begin{eqnarray}
  \begin{aligned}
  &\mu(\sigma)(v_1,\ldots,v_{i-1},0 ,v_{i+1},\ldots,v_k) = 0\enspace.
    \end{aligned}\label{ml-absorbtive-det}
\end{eqnarray}

We can represent each $(\Sigma,\B)$-scalar algebra $(V,0,\mu)$ as a universal algebra $(V,\eta)$ by viewing the scalar multiplications as unary operations on $V$.
For this, we define the index set  $I= \{0\} \cup \{(b\cdot)\mid b \in B\} \cup \Sigma$ and the mapping $\eta: I \to \mathrm{Ops}(V)$ such that
\begin{compactitem}
\item $\eta(0)()=0$,
  \item for every $b \in B$ and $v \in V$, we let $\eta((b \cdot))(v) = b \cdot v$, and
    \item for every $k \in \mathbb{N}$ and $\sigma \in \Sigma^{(k)}$, we let $\eta(\sigma) = \mu(\sigma)$.
\end{compactitem}
And similarly, each $\B$-scalar algebra can be represented as universal algebra.
Thus, all the concepts like subalgebra,
        congruence, homomorphism, and finitely generated, and all the  results of universal algebra
         are available for the concepts of $\B$-scalar algebra and  $(\Sigma,\B)$-scalar algebra.
   In particular, if $\V=(V,0)$ is a $\B$-scalar algebra and $H \subseteq V$ with $H\ne \emptyset$, then 
\begin{equation}\label{eq:equivalent-to-generates}
  \text{$H$ generates $\V$  \  (i.e, \ 
    $\langle H \rangle_{\{0\} \cup \{\eta(b\cdot) \mid b \in B\}}=V$) \ \ \  if and only if \ \ \
    $B\cdot H =V$} \enspace.
\end{equation}

We call a  $\B$-scalar algebra homomorphism a \emph{scalar-linear mapping} (in analogy to the concept of linear mappings between vector spaces). In the rest of this subsection, we show some useful properties of $\B$-scalar algebras which we will use to prove our main results.

 
The next lemma is crucial for our approach. It is an analogue to the fact that, 
in each $\B$-vector space $\V$ with basis $H$, each vector can be represented in a unique way as linear combination over the base vectors. Here the $\B$-vector space and the basis $H$ are replaced by a cancellative $\B$-scalar algebra and a pair-independent generating set $H$, respectively.

Formally, two elements $u,v\in V$ are \emph{dependent (in $(V,0)$)} if there exists a $b\in B$ with $u=b\cdot v$ or $v=b \cdot u$. If $u$ and $v$ are not dependent, then we say that they are \emph{independent}. This definition of dependency is motivated by the concept of linear dependency in vector spaces. In fact, if $(V,+,0)$ is a $\B$-vector space for some field $\B$, then
  \begin{align}
    &\text{for every $u,v \in V$ we have: $u,v$ are dependent in the $\B$-scalar algebra $(V,0)$}\label{equ:dependent-linearly-dependent}\\[-1mm]                                                                            
    &\text{if and only if $u,v$ are linearly dependent in the $\B$-vector space $(V,+,0)$.\notag}
      \end{align}
     We say that \emph{dependency in $(V,0)$ is decidable} if, for every $u,v \in V$, it is decidable whether $u$ and $v$ are dependent or not.

A subset $H\subseteq V$ is called \emph{pair-independent} if any two different elements of $H$ are independent. In particular,
the empty set and each singleton subset of $V$ is pair-independent. Moreover, if $|H| \ge 2$ and $H$ is pair-independent,  then $0 \not\in H$. This is because if $0\in H$ and there exists $a\in H$ with $a\ne 0$, then we have $0= \0\cdot a$, i.e.,  $H$ is not pair-independent.

  \begin{lemma}\rm \label{lm:crucial-canc-pair-ind-imply-uniqueness-det}
     Let $\V=(V,0)$ be a $\B$-scalar algebra, and $H \subseteq V$ be a generating set of $\V$. Then the following two statements hold.
     \begin{compactenum}
     \item[(1)] If $H$ is finite, then there exists a pair-independent subset $H' \subseteq H$ which generates $\V$. Moreover, if additionally dependency in $\V$ is decidable, then we can construct a pair-independent subset $H' \subseteq H$ which generates $\V$.
     \item[(2)] If $\V$ is cancellative and $H$ is pair-independent, then for each $v \in V\setminus \{0\}$ there exist unique $d \in B^{-\0}$ and $u \in H\setminus\{0\}$ such that $v=d\cdot u$.
       \end{compactenum}
\end{lemma}
    
  \begin{proof} Proof of (1): By induction, we define a sequence $H_0,H_1,H_2,\ldots$ of subsets of $V$ such that (a) for each $i \in \mathbb{N}$, the set $H_i$ generates $\V$ and (b) there exists $i\in \mathbb{N}$ such that $H_i$ is pair-independent.

Let $H_0=H$. Assume that $H_i$ is already defined and it generates $\V$. If $H_i$ is pair-independent, then let $H_{i+1}=H_i$; clearly, $H_{i+1}$ generates $\V$.
Otherwise, there exist two different $u,v \in H_i$ and some $b \in B$ such that $u =b \cdot v$ or $v =b \cdot u$. 
Assume that $u =b \cdot v$. We define the set $H_{i+1} = H_i \setminus\{u\}$; since $u\ne v$ we have  $v \in H_{i+1}$. We show that $H_{i+1}$ generates $\V$.
    For this, let $v'\in V$. Since $H_i$ generates $\V$, there exist $u'\in H_i$ and $a\in B$ with $v'=a\cdot u'$. If $u'\ne u$, then $u'\in H_{i+1}$ and we are done. Otherwise, we have $v'=a\cdot u' =a\cdot u = a\cdot (b \cdot v) = (a\otimes b)\cdot v$. The proof of the case that $v =b \cdot u$ is similar by symmetry.

The following facts are clear from the above definition: for each $i \in \mathbb{N}$, (1) we have $H_i=H_{i+1}$ or $H_i\supset H_{i+1}$, and  (2) if $H_i=H_{i+1}$, then $H_i$ is pair-independent (and $H_i=H_{i+1}=H_{i+2}=\ldots$). Since $H_0$ is finite and each singleton is pair-independent, it follows that there exists $i \in \mathbb{N}$ with 
$H_i=H_{i+1}$. Let $H'=H_i$ for the smallest $i$ for which $H_i=H_{i+1}$.

To prove the second part of Statement (1), assume that dependency in $\V$ is decidable. Then, for each $i \in \mathbb{N}$, we can decide if $H_i$ is pair-independent and thus we can construct $H_{i+1}$. Hence we can find the smallest $i \in \mathbb{N}$ with  $H_i=H_{i+1}$.

Proof of (2): Let $v \in V\setminus\{0\}$. Since $H$ generates $\V$, there exist $d\in B$ and $u \in H$ such that $v=  d\cdot u$. Since $v \ne 0$, we have $d \ne \0$ and $u \ne 0$. We prove the uniqueness by contradiction. For this, we assume that there exist $d,d' \in B^{-\0}$ and $u,u' \in H\setminus \{0\}$ such that
$v=d\cdot u=d'\cdot u'$ and $d\ne  d'$ or $u\ne u'$. 
If $u \ne u'$, then $u =(d^{-1}\otimes d')\cdot u'$ contradicts that $H$ is pair-independent.
Thus, we have $u=u'$ and  $d\cdot u=d'\cdot u$ where $d\not=d'$. This contradicts the fact that $\V$ is cancellative. Thus $d=d'$.
\end{proof}

Given $\V$ and $H$ as in Lemma \ref{lm:crucial-canc-pair-ind-imply-uniqueness-det}(2), then we define the mapping
\begin{equation}\label{equ:scalar-vector-det-new}
\mathrm{dec}: V\setminus \! \{0\} \to B^{-\0} \times \big(H \setminus\! \{0\}\big) \enspace,
\end{equation}
 called \emph{decomposition mapping for $V\setminus \! \{0\}$},
 such that for each $v \in V\setminus \! \{0\}$ we let $\mathrm{dec}(v) = (d,u)$ if $v = d \cdot u$.
By Lemma \ref{lm:crucial-canc-pair-ind-imply-uniqueness-det}(2), the mapping $\mathrm{dec}$ is well defined. We denote the first component and second component of $\mathrm{dec}(v)$ by $\mathrm{scal}(v)$ and $\mathrm{gen}(v)$, respectively. Thus, for each $v \in V\setminus \{0\}$, we have 
\(
  v = \mathrm{scal}(v) \cdot \mathrm{gen}(v)
\).

Next we show that the cardinality of a pair-independent generating set of a finitely generated $\B$-scalar algebra is unique.

\begin{lemma}\rm\label{dimension-of-B-scalar-algebra} Let $\V=(V,0)$ be a finitely generated $\B$-scalar algebra. Moreover, let $H_1$ and $H_2$ be finite and pair-independent generating sets of $\V$. Then $|H_1|=|H_2|$.
\end{lemma}
\begin{proof} We give a proof by contradiction. For this, assume that $|H_1|<|H_2|$. Since $H_1$ generates $\V$, there exist $u\in H_1$, $v_1,v_2\in H_2$ with $v_1\ne v_2$, and
  $b_1,b_2\in B^{-\0}$ such that $v_1=b_1\cdot u$ and $v_2=b_2\cdot u$.  Hence
\[b_2\cdot v_1=b_2\cdot (b_1 \cdot u)=(b_2\otimes b_1)\cdot u= (b_1\otimes b_2)\cdot u=b_1\cdot( b_2\cdot u)=b_1\cdot v_2\enspace\]
Then we have $v_1=(b_2^{-1}\otimes b_1)\cdot v_2$, which contradicts the assumption that $H_2$ is pair-independent.
\end{proof}

For each  finitely generated $\B$-scalar algebra $\V$, we define the \emph{degree of $\V$}, denoted by $\deg(\V)$, to be the cardinality of a pair-independent generating set of $\V$. By Lemma \ref{dimension-of-B-scalar-algebra}, $\deg(\V)$ is well defined.

Due to the analogy to the concept of basis of $\B$-vector spaces, we call a pair-independent generating set of a cancellative $\B$-scalar algebra of $(V,0)$ a \emph{scalar-basis} (of $(V,0)$).

\begin{samepage}
\begin{lemma}\rm \label{obs:scalar-algebra-fin-gen-subalg-hom-Z}  Let $(V,0)$ be a $\B$-scalar algebra and $(V',0)$ a sub-$\B$-scalar algebra of it.  If $(V,0)$ is finitely generated, then $(V',0)$ is finitely generated
and $\deg((V',0))\le \deg((V,0))$.
\end{lemma}
\end{samepage}
\begin{proof}  The statement is obvious if $V'=\{0\}$. Therefore, we assume that $V'\not =\{0\}$. Let  $H$ be a finite and pair-independent generating set of $(V,0)$ (cf. Lemma \ref{lm:crucial-canc-pair-ind-imply-uniqueness-det}(1)).
  Let $H'=\{u\in H \mid (\exists b\in B^{-\0}): b\cdot u\in V'\}$. 
  Then $H'\ne \emptyset$ because $H$ generates $(V,0)$ and $\{0\}\ne V'\subseteq V$. We prove that $H'\subseteq V'$ and that $B\cdot H'=V'$ (cf. \eqref{eq:equivalent-to-generates}).
  
  For the proof of the inclusion, let $u\in H'$. By definition, there exist $u\in H$ and  $b\in B^{-\0}$ such that $b\cdot u\in V'$. Since $(V',0)$  is a $\B$-scalar algebra,
  we have that $b^{-1}\cdot(b\cdot u)\in V'$. But $b^{-1}\cdot(b\cdot u)=(b^{-1}\otimes b)\cdot u=\1\cdot u =u$, hence  $u\in V'$. 
  
  Now we prove the equality. First we note that $B\cdot H'\subseteq V'$ because $H'\subseteq V'$ and  $(V',0)$ is a $\B$-scalar algebra. 
  For the proof of $ V' \subseteq B\cdot H'$, let $v\in V'$. If $v=0$, then $v\in B\cdot V'$ obviously, so assume that $v\ne 0$. 
  Since $H$ generates $(V,0)$, there exists $u\in H$ and $b\in B^{-\0}$ such that $v=b\cdot u$. Then  by definition, $u\in H'$, i.e., $v\in B\cdot H'$.
  
  Hence $H'$ generates $(V',0)$. By Lemma \ref{lm:crucial-canc-pair-ind-imply-uniqueness-det}(1), there exists a pair-independent subset $H''\subseteq H'$ which generates 
  $(V',0)$. Then we obtain $\deg((V',0))=|H''| \le |H'| \le |H| = \deg((V,0))$.
 \end{proof}

\begin{observation}\rm \label{obs:isomorphism-preserves-degree} Let $(V,0)$ and $(V',0')$ be $\B$-scalar algebras, $(V,0)$ be finitely generated, and $f: V \to V'$ be a $\B$-scalar algebra isomorphism. Then $(V',0')$ is finitely generated and  $\deg((V,0)) = \deg((V',0'))$.
\end{observation}

Let $(V,0,\mu)$ be a $(\Sigma,\B)$-scalar algebra and let $\sim$ be a congruence on $(V,0,\mu)$. We recall that the quotient algebra of $(V,0,\mu)$ with respect to $\sim$ is the $(\Sigma,\B)$-scalar algebra
\((V,0,\mu)/_\sim =(V/_\sim,[0]_\sim,\mu/_\sim) \). In this quotient algebra the scalar multiplication is the mapping $\cdot/_\sim : B\times V/_\sim \to V/_\sim$ defined for every $b\in B$ and $[v]_\sim \in V/_\sim$ by $b \cdot\!/_\sim [v]_\sim = [b\cdot v]_\sim$. In the sequel we write $\cdot$ for  $\cdot/_\sim$. Moreover, $\mu/_\sim$ is defined, for every $k\in\mathbb{N}$, $\sigma\in\Sigma^{(k)}$, and $[v_1]_\sim,\ldots,[v_k]_\sim \in V/_\sim$, by
\(\mu/_\sim(\sigma)([v_1]_\sim,\ldots,[v_k]_\sim)=[\mu(\sigma)(v_1,\ldots,v_k)]_\sim.\).
The canonical mapping $\pi_\sim: V\to V/_\sim$, defined by $\pi_\sim(v)=[v]_\sim$ for each $v\in V$,  is a homomorphism from $\V$ to  $(\V,\mu)/_\sim$ (cf. Lemma~\ref{thm:canonical-map-of-congr-is-hom}).

\begin{lemma}\rm\label{dimension-of-quotient-B-scalar-algebra} Let $\V=(V,0)$ be a finitely generated $\B$-scalar algebra and let $\sim$ be a congruence on $\V$.
  Then the quotient $\B$-scalar algebra $\V/_\sim=(V/_\sim,[0]_\sim)$ is  finitely generated and $\deg(\V/_\sim)\le \deg(\V)$.
  \end{lemma}
\begin{proof} Let $H$ be a pair-independent generating set of $\V$. Then the set $H/_\sim=\{[u]_\sim\mid u\in H\}$ generates $\V/_\sim$. To see this, let 
$[v]_\sim \in V/_\sim$. Since $H$ generates $\V$, there are $u\in H$ and $b\in B$ such that $v=b\cdot u$. Then $ [v]_\sim=b\cdot [u]_\sim$ because $\sim$ is a congruence.
Moreover 
\(\deg(\V/_\sim) \le |H/_\sim| \le |H|= \deg(\V)\).
\end{proof}

Let $(V,0)$ be a $\B$-scalar algebra. A \emph{scalar-linear} form is a scalar-linear mapping from $(V,0)$ to the $\B$-scalar algebra $(B,\0)$, i.e., it is a mapping $\gamma: V \to B$ such that $\gamma(b\cdot v) = b \otimes \gamma(v)$ for every $b\in B$ and $v \in V$.

\subsection{Scalar algebras  of monomial weighted tree languages}
\label{subsec:sc-alg-general-and-monomials}

Here we  define a particular $\B$-scalar algebra and a particular $(\Sigma,\B)$-scalar algebra. Since monomial weighted tree languages play an important role,  we recall that a monomial $(\Sigma,\B)$-weighted tree language is a mapping $r:\T_\Sigma \to B$ such that there exist a $b \in B$ and a tree $\xi\in\T_\Sigma$ with $r(\xi)=b$ and $r(\zeta)=\0$ for each $\zeta\in\T_\Sigma$ with $\zeta \ne \xi$. We denote this monomial by $b.\xi$.

Now let
 \(
\sfMon(\Sigma,\B) = (\rmMon(\Sigma,\B),\widetilde{\0})
\) where $\rmMon(\Sigma,\B) = \{b.\xi \mid b \in B, \xi \in \T_\Sigma\}$ and $\widetilde{\0}$ is the constant zero $(\Sigma,\B)$-weighted tree language, and let
 \(\cdot: B\times \rmMon(\Sigma,\B) \rightarrow \rmMon(\Sigma,\B)\) 
be the scalar multiplication of $(\Sigma,\B)$-weighted tree languages, i.e.,  for each $a \in B$ and $b.\xi \in \rmMon(\Sigma,\B)$, we have $a \cdot (b.\xi) = (a \otimes b).\xi$.

 It is easy to check that the laws \eqref{SM1-det}, \eqref{SM4-det}, and \eqref{SM5-det} hold, hence $\sfMon(\Sigma,\B)$ is a 
 $\B$-scalar algebra. We call it the \emph{$\B$-scalar algebra  of monomial $(\Sigma,\B)$-weighted tree languages}.
  It is also easy to see that $\sfMon(\Sigma,\B)$ is cancellative and that it is generated by the set $\{\1.\xi\mid \xi\in \T_\Sigma\}$.

Given that $\B$ is a field, the relationship between the $\B$-scalar algebra $(\mathrm{Mon}(\Sigma,\B),\widetilde{\0})$ and the $\B$-vector space $(\mathrm{Pol}(\Sigma,\B),\oplus,\widetilde{\0})$ is the following: the $\B$-scalar algebra $(\rmMon(\Sigma,\B),\widetilde{\0})$ is a subalgebra of the  $\B$-scalar algebra $(\Pol(\Sigma,\B),\widetilde{\0})$; in its turn, the $\B$-scalar algebra $(\Pol(\Sigma,\B),\widetilde{\0})$ is a reduct of the $\B$-vector space $(\Pol(\Sigma,\B),\oplus,\widetilde{\0})$.

We recall that, for every $\sigma \in \Sigma$, the operation  $\ttop(\sigma)$ on $\Pol(\Sigma,\B)$,
called top-concatenation with $\sigma$, is defined in \eqref{eq:definition-of-top-concatenation} of Section~\ref{sect:B-A-BLs-result-det}. It is easy to see that $\rmMon(\Sigma,\B)$ is closed under $\ttop(\sigma)$, hence $\ttop(\sigma)$ is an operation also on $\rmMon(\Sigma,\B)$.

  \begin{lemma}\rm \label{lm:monomials-are-smallest-closed-under-scalar-m-top-det} (cf. \cite[Lm.~15.4.1]{fulvog24}) The set $\rmMon(\Sigma,\B)$ is the  smallest set of $(\Sigma,\B)$-weighted tree languages which is closed under scalar multiplications and top-concatenations. 
\end{lemma}
\begin{proof}
  For convenience, we denote by $\cC$ the  smallest set of $(\Sigma,\B)$-weighted tree languages which is closed under scalar multiplications and top-concatenations.
  
  Since $\rmMon(\Sigma,\B)$ is closed under scalar multiplications and top-concatenations, we have $\cC \subseteq \rmMon(\Sigma,\B)$.  
  
  For  the proof of the other inclusion, first we show that the monomial $\1.\xi \in \cC$ for each $\xi \in \T_\Sigma$. We use induction on $\T_\Sigma$. Let $\xi = \sigma(\xi_1,\ldots,\xi_k)$ and assume that $\1.\xi_i$ is in $\cC$ for each $i \in [k]$. Since $\cC$ is closed under  top-concatenations, we obtain that 
  \(\1.\xi =
    \ttop(\sigma)(\1.\xi_1,\ldots,\1.\xi_k)\) is also in $\cC$. (In case $k=0$, we have $\1.\xi =\1.\sigma$.)
  
  Now let $b.\xi \in \rmMon(\Sigma,\B)$ for some $b\in B$ and $\xi\in\T_\Sigma$. Since  $\1.\xi \in \cC$ and $\cC$ is closed under scalar multiplications,
  we obtain that  $b.\xi=b\cdot (\1.\xi)$ is also in $\cC$.
    \end{proof}

It is also easy to see that, for each $\sigma \in \Sigma$, the operation  $\ttop(\sigma)$ on $\rmMon(\Sigma,\B)$ has property (\ref{ml-Omega-det}). Therefore
\(
    (\rmMon(\Sigma,\B),\widetilde{\0},\ttop) 
  \)
is a $(\Sigma,\B)$-scalar algebra. We call it the \emph{$(\Sigma,\B)$-scalar algebra  of monomial $(\Sigma,\B)$-weighted tree languages}.

The importance of $(\mathrm{Mon}(\Sigma,\B),\widetilde{\0},\ttop)$ is shown in the next lemma: it is initial in the set of all $(\Sigma,\B)$-scalar algebras. This corresponds to the fact that $(\mathrm{Pol}(\Sigma,\B),\oplus,\widetilde{\0},\ttop)$ is initial in the set of all $(\Sigma,\B)$-vector spaces, cf. \cite[p.~451]{bozale89} and \cite[p.~352]{boz91} (also see \cite[Lm.~18.2.4]{fulvog24}).

\begin{theorem} \label{lm:Mon-SigmaB-initial-det} Let $\B$ be a commutative semifield. The $(\Sigma,\B)$-scalar algebra $(\mathrm{Mon}(\Sigma,\B),\widetilde{\0},\ttop)$ is initial in the set of all $(\Sigma,\B)$-scalar algebras.
\end{theorem}
\begin{proof} We use the representation of  $(\mathrm{Mon}(\Sigma,\B),\widetilde{\0},\ttop)$  as universal algebra $(\rmMon(\Sigma,\B),\eta)$ as defined in Subsection~\ref{subsect:scalar-algebras-det}, and we prove the three properties (a), (b), and (c) of the definition of initial  universal algebra.
    
(a) As mentioned, $(\mathrm{Mon}(\Sigma,\B),\widetilde{\0},\ttop)$ is a member of the set of all $(\Sigma,\B)$-scalar algebras.

(b) The universal algebra $(\mathrm{Mon}(\Sigma,\B),\eta)$ (cf. Subsection \ref{subsect:scalar-algebras-det}) is generated by $\emptyset$, i.e., $\langle \emptyset\rangle_{\im(\eta)} = \rmMon(\Sigma,\B)$. This follows by Lemma \ref{lm:monomials-are-smallest-closed-under-scalar-m-top-det}.
    
(c) Let $(V,0,\mu)$ be an arbitrary $(\Sigma,\B)$-scalar algebra. We  prove that there exists a scalar-linear mapping from $(\rmMon(\Sigma,\B),\widetilde{\0},\ttop)$ to $(V,0,\mu)$.     

We define the mapping $h:\rmMon(\Sigma,\B) \to V$ such that, for every $b \in B$ and $\xi \in \T_\Sigma$, we let
\begin{equation}\label{eq:from-h-to-hV-mon-det}
  h(b.\xi) = b \cdot \h_\mu(\xi)\enspace,
  \end{equation}
where $\h_\mu$ denotes the unique  $\Sigma$-algebra homomorphism from the $\Sigma$-term algebra  $\mathsf{T}_\Sigma$ to $(V,\mu)$. 
Now we prove that $h$ is a scalar-linear mapping from $(\rmMon(\Sigma,\B),\widetilde{\0},\ttop)$ to $(V,0,\mu)$.  

It is obvious that, for every $s\in \rmMon(\Sigma,\B)$ and $b\in B$, we have $h(b\cdot s)=b\cdot h(s)$. 

Next we show that $h$ is a  $\Sigma$-algebra homomorphism from $(\rmMon(\Sigma,\B),\ttop)$ to $(V,\mu)$, i.e., for every $k\in \mathbb{N}$, $\sigma\in \Sigma^{(k)}$, and $s_1,\ldots,s_k \in \rmMon(\Sigma,\B)$, we have $h\big(\ttop(\sigma)(s_1,\ldots,s_k)\big)=\mu(\sigma)\big(h(s_1),\ldots,h(s_k)\big)$.
For this, for each $i\in [k]$, let $s_i=b_i.\xi_i$ for some $b_i \in B$ and $\xi_i \in \T_\Sigma$.
Then we can calculate as follows.
\begingroup
\allowdisplaybreaks
\begin{align*}
  & h\big(\ttop(\sigma)(s_1,\ldots,s_k)\big)\\[2mm]
  &= h\big((b_1\otimes\ldots\otimes b_k).\sigma(\xi_1,\ldots,\xi_k)\big)
    \tag{by definition of $\ttop(\sigma)$}\\[2mm]
  &= (b_1\otimes\ldots\otimes b_k)\cdot \h_\mu(\sigma(\xi_1,\ldots,\xi_k))
  \tag{by \eqref{eq:from-h-to-hV-mon-det}}\\[2mm]
  &= (b_1\otimes\ldots\otimes b_k)\cdot \mu(\sigma)(\h_\mu(\xi_1),\ldots,\h_\mu(\xi_k))
  \tag{because $\h_\mu$ is a $\Sigma$-algebra homomorphism}\\[2mm]
  &= \mu(\sigma)\big(b_1\cdot \h_\mu(\xi_1),\ldots, b_k\cdot \h_\mu(\xi_k)  \big)
  \tag{by \eqref{ml-Omega-det}}\\[2mm]
  &= \mu(\sigma)\big(h(s_1),\ldots,h(s_k)\big)
    \tag{by \eqref{eq:from-h-to-hV-mon-det}}
    \enspace.
\end{align*}
\endgroup
This proves that $h$ is a $\Sigma$-algebra homomorphism. Hence $h$ is a scalar-linear mapping. 
\end{proof}


\section{The m-syntactic congruence of a weighted tree language}
\label{sec:syntactic-congruence-det}

   In B-A's theorem (cf. Theorem \ref{th:MN-fields-det}), the syntactic congruence $\approx_r$
of a $(\Sigma,\B)$-weighted tree language $r: \T_\Sigma \to B$  over the field $\B$ is defined as a congruence on the $(\Sigma,\B)$-vector space $(\Pol(\Sigma,\B),\oplus,\widetilde{\0},\ttop)$ of polynomial weighted tree languages. Here we consider the restriction of  $\approx_r$ to monomials, i.e., we let
\(\sim_r \, = \,\approx_r \cap \,(\rmMon(\Sigma,\B) \times \rmMon(\Sigma,\B))\).
In more detail, for every $b_1.\xi_1, b_2.\xi_2 \in \rmMon(\Sigma,\B)$, we let
\[
b_1.\xi_1 \sim_r b_2.\xi_2 \ \text{ iff } \ (\forall c \in \C_\Sigma): b_1 \otimes r(c[\xi_1]) = b_2 \otimes r(c[\xi_2]) \enspace.
  \]
  
  Since $\approx_r$ is a congruence on the $(\Sigma,\B)$-vector space $(\Pol(\Sigma,\B),\oplus,\widetilde{\0},\ttop)$,  it is also a congruence on the $(\Sigma,\B)$-scalar algebra $(\Pol(\Sigma,\B),\widetilde{\0},\ttop)$.
  Moreover, since $(\mathrm{Mon}(\Sigma,\B),\widetilde{\0},\ttop)$ is a subalgebra of $(\Pol(\Sigma,\B),\widetilde{\0},\ttop)$, by Lemma \ref{thm:congruence-restricted-to-subalgebra},
  the relation $\sim_r$ is a congruence on $(\mathrm{Mon}(\Sigma,\B),\widetilde{\0},\ttop)$.
   
  Since $\sim_r$ is the restriction of $\approx_r$ to \underline{m}onomials, we call it the \emph{m-syntactic congruence of $r$}. Moreover, we call the quotient algebras 
  $(\mathrm{Mon}(\Sigma,\B),\widetilde{\0},\ttop)/_{\sim_r}$ and  $\mathsf{Mon}(\Sigma,\B)/_{\sim_r}$ the \emph{m-syntactic $(\Sigma,\B)$-scalar algebra} and the \emph{m-syntactic $\B$-scalar algebra (of $r$)}, respectively.

  \begin{example}\rm \label{ex-syntactic-congruence} We consider the $(\Sigma,\Ratnum)$-weighted tree language $r$ in Example \ref{ex:bu-det+total-wta-det} and compute the congruence $\sim_r$. For this, we recall that, for each $\xi \in \T_\Sigma$, we have
    \[
      r(\xi) = \begin{cases} 2 \cdot 2^{\#_\alpha(\xi)} & \text{ if $\#_\alpha(\xi)$ is even}\\
       3 \cdot 2^{\#_\alpha(\xi)} & \text{ otherwise,} \end{cases}
   \]
   where $\#_\alpha(\xi)$ denotes the number of occurrences of $\alpha$ in $\xi$.

   Then, for every $b_1.\xi_1,b_2.\xi_2 \in \rmMon(\Sigma,\Ratnum)$, we have
   \begingroup
   \allowdisplaybreaks
\begin{align*} 
& \hspace*{5mm} b_1.\xi_1 \sim _r b_2.\xi_2   \\
&  \text{iff } (\forall c \in \C_\Sigma): b_1 \otimes r(c[\xi_1]) = b_2 \otimes r(c[\xi_2])\\
&  \text{iff }   \Big[(\forall c\in \C_\Sigma) : \Big( \big(b_1\cdot 2 \cdot 2^{\#_\alpha(c[\xi_1])} = b_2\cdot 2 \cdot 2^{\#_\alpha(c[\xi_2])} \text{ and } \#_\alpha(c[\xi_1]) \text{ and } \#_\alpha(c[\xi_2]) \text{ are even} \big) \text{ or }\\
& \hspace*{25mm} \big( b_1\cdot 3 \cdot 2^{\#_\alpha(c[\xi_1])} = b_2\cdot 3 \cdot 2^{\#_\alpha(c[\xi_2])} \text{ and } \#_\alpha(c[\xi_1]) \text{ and } \#_\alpha(c[\xi_2]) \text{ are odd}\big)\Big)\Big]\\
&  \hspace*{25mm} \text{or \ } b_1=b_2=0\\[2mm]
&  \text{iff } \Big(\big( b_1\cdot 2^{\#_\alpha(\xi_1)} = b_2\cdot  2^{\#_\alpha(\xi_2)} \text{ and } \#_\alpha(\xi_1) \text{ and } \#_\alpha(\xi_2) \text{ are even }\big) \text{ or }\\
 & \hspace*{5mm} \big(  b_1\cdot 2^{\#_\alpha(\xi_1)} = b_2\cdot  2^{\#_\alpha(\xi_2)} \text{ and } \#_\alpha(\xi_1) \text{ and } \#_\alpha(\xi_2) \text{ are odd } \big)\Big)\ \text{or \ } b_1=b_2=0\\[2mm]
 &   \text{iff } \Big(b_1 \cdot 2^{\#_\alpha(\xi_1)} =  b_2 \cdot 2^{\#_\alpha(\xi_2)} \text{\  and\  } \big(\#_\alpha(\xi_1) \! \mod(2)\big) = \big(\#_\alpha(\xi_2) \! \mod (2)\big)\Big)\text{ or \ } b_1=b_2=0\enspace,
\end{align*}
\endgroup
where at the third equivalence we extend $\#_\alpha$ to contexts in the obvious way and use the fact that $\#_\alpha(c[\xi])=\#_\alpha(c)+\#_\alpha(\xi)$ for each $\xi\in \T_\Sigma$. It is easy to see that in this particular case $\sim_r=\ker(\sfh_\cA)$, where $\cA$ is the bu-deterministic $(\Sigma,\Ratnum)$ which recognizes $r$ (cf. Example~\ref{ex:bu-det+total-wta-det}).
Moreover, dependency in $\sfMon(\Sigma,\B)/_{\sim_r}$ is decidable.
\hfill$\Box$
\end{example}

In the rest of this subsection we show several properties of the m-syntactic congruence $\sim_r$.
It is obvious that if $r=\widetilde{\0}$, then $\sim_r=\rmMon(\Sigma,\B) \times \rmMon(\Sigma,\B)$, i.e., the index of $\sim_r$ is 1.
If $r\ne \widetilde{\0}$, then for each $\xi\in \T_\Sigma$ with $r(\xi)\ne \0$ and $b_1,b_2\in B$, the relation $b_1.\xi\sim_r b_2.\xi$ implies that $b_1=b_2$.
This is because if $b_1.\xi\sim_r b_2.\xi$, then, with context $c=z$, we have  $b_1\otimes r(\xi) = b_2\otimes r(\xi)$ and
by multiplying with $r(\xi)^{-1}$ we obtain $b_1=b_2$. Hence $\sim_r=\rmMon(\Sigma,\B) \times \rmMon(\Sigma,\B)$ implies that $r=\widetilde{\0}$.
It also follows that if $r\ne \widetilde{\0}$, then $\sim_r$ has at least as many equivalence classes as the cardinality of $B$. In particular, if $r\ne \widetilde{\0}$
and $B$ is infinite, then $\sim_r$ does not have finite index.

    Intuitively, we can retrieve from $\sim_r$ the weighted tree language $r$ by means of a scalar-linear form. To formalize this, we consider an arbitrary $r: \T_\Sigma \to B$ and an arbitrary congruence $\sim$ on $(\rmMon(\Sigma,\B),\widetilde{\0},\ttop)$. We say that $\sim$ \emph{saturates $r$} if there exists a scalar-linear form $\gamma: \rmMon(\Sigma,\B)/_\sim \to B$ such that, for every $\xi \in \T_\Sigma$, we have that $r(\xi) = \gamma([\1.\xi]_\sim)$. If this is the case, then we say that  \emph{$\sim$ saturates $r$ via $\gamma$}.
    
    \begin{lemma}\rm \label{lm:uniqueness-of-scalar-linear-form} Let $r: \T_\Sigma \to B$ and $\sim$ be a congruence on $(\rmMon(\Sigma,\B),\widetilde{\0},\ttop)$ such that $\sim$ saturates $r$. Then there exists exactly one scalar-linear form $\gamma: \rmMon(\Sigma,\B)/_\sim \to B$ such that $\sim$ saturates $r$ via $\gamma$.
    \end{lemma}
    \begin{proof} By definition, there exists a  scalar-linear form $\gamma: \rmMon(\Sigma,\B)/_\sim \to B$ such that  $\sim$ saturates $r$ via $\gamma$. Now let  $\gamma': \rmMon(\Sigma,\B)/_\sim \to B$ be a scalar-linear form such that  $\sim$ saturates $r$ via $\gamma'$. Then, for each $b.\xi \in \rmMon(\Sigma,\B)$, we have
      \begingroup
      \allowdisplaybreaks
      \begin{align*}
        \gamma([b.\xi]_\sim) = \gamma(b \cdot [\1.\xi]_\sim) = b \cdot \gamma([\1.\xi]_\sim) = b \cdot r(\xi)
                             = b \cdot \gamma'([\1.\xi]_\sim) = \gamma'(b \cdot [\1.\xi]_\sim)
         = \gamma'([b.\xi]_\sim).
     \end{align*}
      \endgroup 
      \end{proof}

  \begin{lemma}\rm \label{lm:sim-r-is-coarsest} Let $r:\T_\Sigma \to \B$. Then $\sim_r$ is the coarsest congruence among the congruences on $(\rmMon(\Sigma,\B),\widetilde{\0},\ttop)$ which saturate $r$.
\end{lemma}
\begin{proof} First, we prove that $\sim_r$ saturates $r$. We define the mapping $\gamma: \rmMon(\Sigma,\B)/_{\sim_r}  \to B$ for every $b \in B$ and $\xi \in \T_\Sigma$ by $\gamma([b.\xi]_{\sim_r}) = b\otimes r(\xi)$. We show that  $\gamma$ is well defined. Let $a \in B$ and $\zeta \in \T_\Sigma$ such that $[b.\xi]_{\sim_r}=[a.\zeta]_{\sim_r}$. Then, for each $c\in \C_\Sigma$, we have $b \otimes r(c[\xi])= a \otimes r(c[\zeta])$. For $c=z$, this yields $b\otimes r(\xi) = a \otimes r(\zeta)$, i.e., $\gamma([b.\xi]_{\sim_r}) = \gamma([a.\zeta]_{\sim_r})$. Thus $\gamma$ is well defined. Obviously, $\gamma$ is a scalar-linear form. Moreover, for each $\xi \in \T_\Sigma$, we have $\gamma([\1.\xi]_{\sim_r})=r(\xi)$, i.e., $\sim_r$ saturates $r$.

Now let $\sim$ be a congruence on $(\rmMon(\Sigma,\B),\widetilde{\0},\ttop)$ which saturates $r$  via the scalar-linear form 
$\gamma: \rmMon(\Sigma,\B)/_{\sim}  \to B$. Moreover, let $b_1.\xi_1, b_2.\xi_2 \in \rmMon(\Sigma,\B)$ such that
$b_1.\xi_1 \sim b_2.\xi_2$. Then
\begingroup
\allowdisplaybreaks
\begin{align*}
& b_1.\xi_1 \sim b_2.\xi_2 \ \Rightarrow \  (\forall c\in \C_\Sigma): b_1.c[\xi_1] \sim b_2.c[\xi_2] \tag{because $\sim$ is a congruence} \\
\Leftrightarrow \ \ & (\forall c\in \C_\Sigma): [b_1.c[\xi_1]]_\sim = [b_2.c[\xi_2]]_\sim \
\Rightarrow \  (\forall c\in \C_\Sigma): \gamma\big([b_1.c[\xi_1]]_\sim\big) = \gamma\big([b_2.c[\xi_2]]_\sim\big) \\
\Rightarrow \ \ & (\forall c\in \C_\Sigma): b_1\cdot\gamma\big([\1.c[\xi_1]]_\sim\big) = b_2\cdot\gamma\big([\1.c[\xi_2]]_\sim\big) \tag{because $\gamma$ is a scalar-linear form}\\
\Rightarrow \ \ & (\forall c\in \C_\Sigma): b_1\cdot r(c[\xi_1]) =  b_2\cdot r(c[\xi_2]) \tag{because $\sim$ saturates $r$ via $\gamma$}\\
  \Rightarrow \ \ &  b_1.\xi_1 \sim_r b_2.\xi_2 \enspace.
\end{align*}
\endgroup
\end{proof}

\begin{lemma} \rm \label{lm:MonSigmaB-modulo-simr-cancellative-det} Let $r:\T_\Sigma \to \B$. Then the $\B$-scalar algebra $\sfMon(\Sigma,\B)/_{\sim_r}$ is cancellative.
\end{lemma} 
\begin{proof}  Let $[b.\xi]_{\sim_r} \in  \rmMon(\Sigma,\B)/_{\sim_r}\setminus \{[\widetilde{\0}]_{\sim_r}\}$ and $b_1,b_2 \in B$ such that $b_1 \cdot [b.\xi]_{\sim_r}= b_2 \cdot [b.\xi]_{\sim_r}$. We have to prove that $b_1=b_2$.

  Since $[b.\xi]_{\sim_r} \ne [\widetilde{\0}]_{\sim_r}$, the definition of $\sim_r$ implies that there exists $c_\xi \in \C_\Sigma$ such that $b \otimes r(c_\xi[\xi])\ne \0$.
Since $b_1 \cdot [b.\xi]_{\sim_r}= b_2 \cdot [b.\xi]_{\sim_r}$, we have $(b_1 \otimes b).\xi \sim_r (b_2 \otimes b).\xi$.  Thus, for each $c \in \C_\Sigma$, we have $b_1 \otimes b \otimes r(c[\xi]) = b_2 \otimes b \otimes r(c[\xi])$. In particular, we have $b_1 \otimes b \otimes r(c_\xi[\xi]) = b_2 \otimes b \otimes r(c_\xi[\xi])$. By dividing both sides  by  $b \otimes r(c_\xi[\xi])$ (which is not $\0$), we obtain that $b_1=b_2$.
  \end{proof}


\section{B-A's results for bu-deterministic wta}
\label{sec:result-det}

In  Subsection \ref{subsect:det-version-BA-theorem}
  we prove the deterministic version of B-A's theorem (cf. Theorem~\ref{thm:budRec-mono-det}). The two main proof tasks are the following: (1) For a given weighted tree language $r \in \budRec(\Sigma,\B)$, prove that the $\B$-scalar algebra $\sfMon(\Sigma,\B)/_{\sim_{r}}$ is finitely generated and (2) from the objects given in Theorem~\ref{thm:budRec-mono-det}(B), construct a bu-deterministic $(\Sigma,\B)$-wta $\cA$ and prove that $\cA$ recognizes $r$. In Subsections~\ref{subsect:B-scalar-algebra-finitely-generated} and \ref{sec:constr-bu-det-wta-from-congurence+pair-independent-gen-set}, we present these two proofs, respectively. Then in Subsection~\ref{subs:minimization} and \ref{subsect:characterization-of-minimlity}, we construct a minimal bu-deterministic wta for each given bu-determistic wta and prove a characterization of minimality, respectively.

\subsection{The $\B$-scalar algebra $\sfMon(\Sigma,\B)/_{\sim_{\sem{\cA}}}$ is finitely generated}
\label{subsect:B-scalar-algebra-finitely-generated}

Let $\cA=(Q,\delta,F)$ be a bu-deterministic $(\Sigma,\B)$-wta.
In a natural way, $\cA$  induces a $(\Sigma,\B)$-scalar algebra as follows. We consider the $\B$-scalar algebra $(B^Q_{\le 1},\0^Q)$ with the scalar multiplication $\cdot: B \times B^Q_{\le 1} \to B^Q_{\le 1}$ defined in the obvious way and also the $\Sigma$-algebra $(B^Q_{\le 1},\delta_\cA)$ as defined in Lemma \ref{lm:properties-sem-of-budet-wta-det-new}(2). 
Then the triple 
\(
  \sfM(\cA)=(B^Q_{\le 1},\0^Q,\delta_\cA)
\)
is a $(\Sigma,\B)$-scalar algebra via scalar multiplication $\cdot$ because,  by Lemma~\ref{lm:properties-sem-of-budet-wta-det-new}(1), the law \eqref{ml-Omega-det} holds for $\delta_\cA$.

Since, by Theorem~\ref{lm:Mon-SigmaB-initial-det}, $(\mathrm{Mon}(\Sigma,\B),\widetilde{\0},\ttop)$ is initial in the set of all $(\Sigma,\B)$-scalar algebras, there exists a unique scalar-linear mapping from $(\mathrm{Mon}(\Sigma,\B),\widetilde{\0},\ttop)$ to $\sfM(\cA)$. We denote it by~$\sfh_\cA$. By~\eqref{eq:from-h-to-hV-mon-det} (for $(V,0,\mu) = \sfM(\cA)$), for each $b.\xi \in \rmMon(\Sigma,\B)$, we have
\begin{equation}\label{equ:sfh-related-to-h}
\sfh_\cA(b.\xi) = b \cdot \h_\cA(\xi) \enspace,
\end{equation}
because, by Lemma \ref{lm:properties-hA-of-budet-wta-det-new}(1), $\h_\cA$ is the unique $\Sigma$-algebra homomorphism from $\sfT_\Sigma$ also to $(B^Q_{\le 1},\delta_\cA)$.
By Theorem~\ref{thm:kernel-is-congruence}, the kernel $\ker(\sfh_\cA)$ of $\sfh_\cA$ is a congruence on $(\rmMon(\Sigma,\B),\widetilde{\0},\ttop)$.

By Lemma \ref{lm:hom-image=subalgebra}, $(\im(\sfh_\cA),\0^Q,\delta_\cA)$ is a sub-$(\Sigma,\B)$-scalar algebra of $\sfM(\cA)$.
We denote this subalgebra by $\sfM_{\mathrm{im}}(\cA)$. By Corollary~\ref{cor:image-of-hom-isomorphic-to-quotient-of-kernel}, we have
\begin{equation}\label{eq:isomorphism}
(\rmMon(\Sigma,\B),\widetilde{\0},\ttop)/_{\ker(\sfh_\cA)} \cong  \sfM_{\mathrm{im}}(\cA) \enspace.
\end{equation}
In a part of Figure \ref{fig:overview-congruence-relations-isos} we visualize \eqref{eq:isomorphism}.

    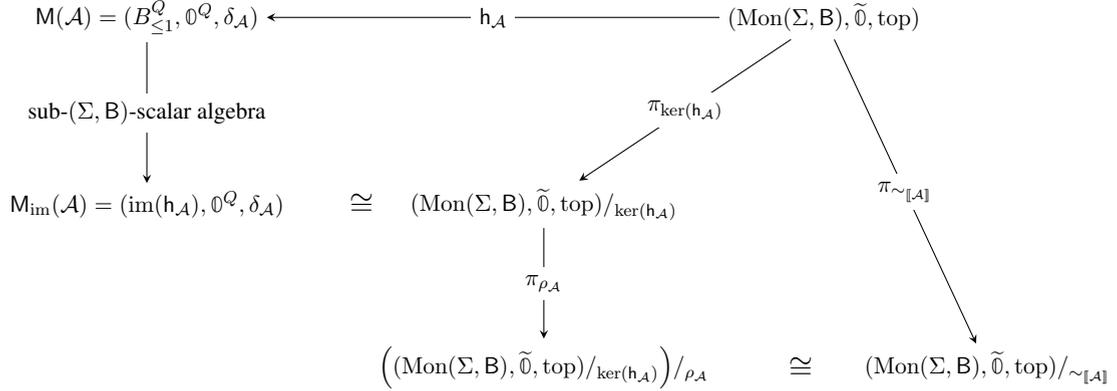
\begin{figure}[t]
 \begin{center}
   \begin{tikzpicture}
    \tikzset{node distance=4em, scale=0.8, transform shape}
     \node (1){$\sfM(\cA) =(B_{\le 1}^Q,\0^Q,\delta_\cA)$};
     \node[right of=1, xshift=25em] (2) {$(\rmMon(\Sigma,\B),\widetilde{\0},\ttop)$};
     \node[below of=1, yshift=-4em] (3) {$\sfM_{\mathrm{im}}(\cA) =(\im(\sfh_\cA),\0^Q,\delta_\cA)$};
     \node[right of=3, xshift=13em] (4) {$(\rmMon(\Sigma,\B),\widetilde{\0},\ttop)/_{\ker(\sfh_\cA)}$};
     \node[below of=4, yshift=-3em] (5) {$\Big((\rmMon(\Sigma,\B),\widetilde{\0},\ttop)/_{\ker(\sfh_\cA)}\Big)/_{\rho_\cA}$};
     \node[right of=5, xshift=15em] (6) {$(\rmMon(\Sigma,\B),\widetilde{\0},\ttop)/_{\sim_{\sem{\cA}}}$};

     \draw (2) edge[->,>=stealth] node[fill=white] {$\sfh_\cA$} (1);
        \draw (1) edge[->,>=stealth] node[fill=white] {sub-$(\Sigma,\B)$-scalar algebra} (3);
     \draw (2) edge[->,>=stealth] node[fill=white] {$\pi_{\ker(\sfh_\cA)}$} (4);
     \draw (2) edge[->,>=stealth] node[fill=white] {$\pi_{\sim_{\sem{\cA}}}$} (6);
     \draw (4) edge[->,>=stealth] node[fill=white] {$\pi_{\rho_\cA}$} (5);

          \node[right of=3, xshift=5.2em] (7) {\Large $\cong$};
          \node[right of=5, xshift=7em] (8) {\Large $\cong$};

   \end{tikzpicture}
 \end{center}
 \caption{\label{fig:overview-congruence-relations-isos} Overview of the relation among several $(\Sigma,\B)$-scalar algebras which occur in Theorem~\ref{lm:congruences-modulo-congruences-det}. }
      \end{figure}

For each bu-deterministic $(\Sigma,\B)$-wta $\cA$, we relate $\sim_{\sem{\cA}}$ and $\ker(\sfh_\cA)$ as follows, where  $\sfh_\cA$ is defined in \eqref{equ:sfh-related-to-h}.

\begin{lemma}\rm \label{lm:kerPsiA-subseteeq-sim-semA-det} Let $\cA$ be a bu-deterministic $(\Sigma,\B)$-wta. Then $\ker(\sfh_\cA) \subseteq \sim_{\sem{\cA}}$.
\end{lemma} 

\begin{proof}  By Lemma \ref{lm:sim-r-is-coarsest}, $\sim_{\sem{\cA}}$ is the coarsest congruence on $(\sfMon(\Sigma,\B),\oplus,\widetilde{\0},\ttop_\Sigma)$ which saturates $\sem{\cA}$. Thus it suffices to prove that $\ker(\sfh_\cA)$ saturates $\sem{\cA}$. For this we define the mapping $\gamma: \rmMon(\Sigma,\B)/_{\ker(\sfh_\cA)} \to B$ for each $b.\xi \in \rmMon(\Sigma,\B)$ by
  \(
\gamma\big([b.\xi]_{\ker(\sem{\cA})}  \big) = b \otimes \sem{\cA}(\xi) 
\).
We prove that $\gamma$ is well defined. For this let $b.\xi, a.\zeta \in \rmMon(\Sigma,\B)$. Then we can calculate as follows.
\begingroup
\allowdisplaybreaks
\begin{align*}
  b.\xi \ \ker(\sfh_\cA) \ a.\zeta
 \ \ \Leftrightarrow \ \ & \ \ b \otimes \h_\cA(\xi) = a \otimes \h_\cA(\zeta)
  \tag{by \eqref{equ:sfh-related-to-h}}\\
  \ \ \Rightarrow \ \ & \ \ \big(b \otimes \h_\cA(\xi) \big)\cdot F = \big(a \otimes \h_\cA(\zeta)\big)\cdot F
  \tag{where $\cdot$ is the scalar product of two vectors}\\
\ \  \Leftrightarrow \ \ & \ \ b \otimes \big(\h_\cA(\xi)\cdot F\big) = a \otimes \big(\h_\cA(\zeta)\cdot F\big)\\
\ \   \Leftrightarrow \ \  & \ \ b \otimes \sem{\cA}(\xi) = a \otimes \sem{\cA}(\zeta) \\
 \ \ \Leftrightarrow \ \  & \ \ \gamma([b.\xi]_{\ker(\sem{\cA})})=\gamma([a.\zeta]_{\ker(\sem{\cA})}) \enspace.
\end{align*}
\endgroup
Hence $\gamma$ is well defined. It is easy to see that $\gamma$ is a linear form. Since, for each $\xi \in \T_\Sigma$, we have $\gamma([\1.\xi]_{\ker(\sfh_\cA)})= \sem{\cA}(\xi)$, the congruence $\ker(\sfh_\cA)$ saturates $\sem{\cA}$ via $\gamma$.
    \end{proof}

    \begin{example}\rm \label{ex:synt-congr-of-running-example-new-det} We show a total and bu-deterministic $(\Sigma,\Ratnum)$-wta $\cA=(Q,\delta,F)$ such that $\ker(\sfh_\cA)\! \subset \ \sim_{\sem{\cA}}$, i.e., $\ker(\sfh_\cA)$ is a proper subset of  $\sim_{\sem{\cA}}$.
      Let $\Sigma = \{\alpha^{(0)},\beta^{(0)}\}$, $Q=\{p,q\}$,
    $\delta_0(\varepsilon,\alpha,p)=\delta_0(\varepsilon,\beta,q)=1$ and $\delta_0(\varepsilon,\alpha,q)=\delta_0(\varepsilon,\beta,p)=0$
    and $F_p=F_q=1$. Using the order $p$ above $q$, we have:
    \[
      \h_\cA(\alpha) = \left(\begin{matrix}1\\0\end{matrix}\right) \ \ \text{ and } \ \
      \h_\cA(\beta) = \left(\begin{matrix}0\\1\end{matrix}\right) \enspace.
    \]
 Hence, for every $b_1,b_2\in B$ and $\xi_1,\xi_2 \in \{\alpha,\beta\}$, we have $b_1\cdot \h_\cA(\xi_1)=b_2\cdot \h_\cA(\xi_2)$ if and only if
   $b_1=b_2$ and $\xi_1=\xi_2$. This means $\ker(\sfh_\cA)=\{(b.\xi,b.\xi)\mid b\in B, \xi \in \{\alpha,\beta\}\}$, i.e., $\ker(\sfh_\cA)$ is the identity on $\rmMon(\Sigma,\B)$. 
    
    On the other hand, since $\C_\Sigma = \{z\}$, for every $b \in \mathbb{Q}$ and $c\in \C_\Sigma$, we have
    \[
      b \otimes \sem{\cA}(c[\alpha]) = b \otimes \h_\cA(\alpha)_p \otimes F_p = b \otimes \h_\cA(\beta)_q \otimes F_q = b \otimes \sem{\cA}(c[\beta]) \enspace.
    \]
    Thus $b.\alpha \sim_{\sem{\cA}} b.\beta$. In fact,
    $\sim_{\sem{\cA}} = \{ (b.\xi_1,b.\xi_2) \mid b \in B, \xi_1,\xi_2 \in \{\alpha,\beta\}\}$.
    \hfill $\Box$
    \end{example} 
 
    For each bu-deterministic $(\Sigma,\B)$-wta $\cA$, we define the binary relation $\rho_\cA$ on $\rmMon(\Sigma,\B)/_{\ker(\sfh_\cA)}$ such that, for every $[b_1.\xi_1]_{\ker(\sfh_\cA)}, [b_2.\xi_2]_{\ker(\sfh_\cA)} \in \rmMon(\Sigma,\B)/_{\ker(\sfh_\cA)}$, we let
       \begin{equation}\label{equ:kerPsiA-versus-simA}
      [b_1.\xi_1]_{\ker(\sfh_\cA)} \ \rho_\cA \  [b_2.\xi_2]_{\ker(\sfh_\cA)} \ \text{ if  } \
    b_1.\xi_1 \ \sim_{\sem{\cA}} \  b_2.\xi_2 \enspace.
  \end{equation}

  The relation $\rho_\cA$ is well defined because $\ker(\sfh_\cA) \subseteq \sim_{\sem{\cA}}$ (cf. Lemma \ref{lm:kerPsiA-subseteeq-sim-semA-det}).

    \begin{theorem}\label{lm:congruences-modulo-congruences-det} (cf. Figure \ref{fig:overview-congruence-relations-isos}.) Let $\B$ be a commutative semifield. Moreover, let $\cA$ be a bu-deterministic $(\Sigma,\B)$-wta. Then the following statements hold.
      \begin{compactenum}
    \item[(1)]  The relation $\rho_\cA$ is a congruence on the $(\Sigma,\B)$-scalar algebra $(\rmMon(\Sigma,\B),\widetilde{\0},\ttop)/_{\ker(\sfh_\cA)}$.
    \item[(2)] $\Big((\rmMon(\Sigma,\B),\widetilde{\0},\ttop)/_{\ker(\sfh_\cA)}\Big)/_{\rho_\cA} \cong (\rmMon(\Sigma,\B),\widetilde{\0},\ttop)/_{\sim_{\sem{\cA}}}$.
    \end{compactenum}
  \end{theorem}
  
  \begin{proof} By Lemma \ref{lm:kerPsiA-subseteeq-sim-semA-det}, we have $\ker(\sfh_\cA) \subseteq \sim_{\sem{\cA}}$.  Then  Statement (1) follows from Theorem~\ref{thm:snd-isom-theorem}(1) and Statement (2) follows Theorem~\ref{thm:snd-isom-theorem}(2).
    \end{proof}

  For instance, for the bu-deterministic $(\Sigma,\B)$-wta $\cA$ of Example \ref{ex:synt-congr-of-running-example-new-det}, for every $b_1.\xi_1, b_2.\xi_2 \in \rmMon(\Sigma,\B)$ we have: $[b_1.\xi_1]_{\ker(\sfh_\cA)} \rho_\cA [b_2.\xi_2]_{\ker(\sfh_\cA)}$ iff $b_1=b_2$.

   The next theorem shows that, for each bu-deterministically recognizable $(\Sigma,\B)$-weighted tree language~$r$, the number of states of any $(\Sigma,\B)$-wta which recognizes $r$, is at least the degree of $\sfMon(\Sigma,\B)/_{\sim_r}$. The steps of its proof can be tracked on Figure~\ref{fig:overview-congruence-relations-isos}.  This theorem can be compared to \cite[Cor.~3]{bor03}.
    
  \begin{theorem}\rm\label{lm:dimension-synt-factor-algebra-bounded} Let $\cA=(Q,\delta,F)$ be a  bu-deterministic $(\Sigma,\B)$-wta. Then the $\B$-scalar algebra $\sfMon(\Sigma,\B)/_{\sim_{\sem{\cA}}}$ is finitely generated and $\deg\big( \sfMon(\Sigma,\B)/_{\sim_{\sem{\cA}}}\big)\le |Q|$.
  \end{theorem} 
\begin{proof} Clearly, the $\B$-scalar algebra $(B^Q_{\le 1},\0^Q)$ is finitely generated and its degree is $|Q|$, i.e.,
  \(
|Q| = \deg\big( (B^Q_{\le 1},\0^Q) \big) 
  \).
  By Lemma~\ref{lm:hom-image=subalgebra} $(\im(\sfh_\cA),\0^Q)$ is a sub-$\B$-scalar algebra of $(B^Q_{\le 1},\0^Q)$, and by Lemma \ref{obs:scalar-algebra-fin-gen-subalg-hom-Z},  it is also finitely generated and
  \(
\deg\big( (B^Q_{\le 1},\0^Q) \big) \ge  \deg\big( (\im(\sfh_\cA),\0^Q)\big) \enspace.
    \)
 
    By \eqref{eq:isomorphism}, $(\im(\sfh_\cA),\0^Q) \cong \sfMon(\Sigma,\B)/_{\ker(\sfh_\cA)}$ and thus, by Observation~\ref{obs:isomorphism-preserves-degree}, we have
  \(
\deg\big( (\im(\sfh_\cA),\0^Q)\big) = \deg\big( \sfMon(\Sigma,\B)/_{\ker(\sfh_\cA)}\big) \).
 By Theorem \ref{lm:congruences-modulo-congruences-det}(1), $\rho_\cA$ is a congruence on $\sfMon(\Sigma,\B)/_{\ker(\sfh_\cA)}$ and, by Lemma \ref{dimension-of-quotient-B-scalar-algebra}\,  we have
  \(
\deg\big( \sfMon(\Sigma,\B)/_{\ker(\sfh_\cA)}\big) \ge \deg\Big(\big(\sfMon(\Sigma,\B)/_{\ker(\sfh_\cA)}\big)/_{\rho_\cA}\Big) \enspace.
\)

    Lastly, by Theorem \ref{lm:congruences-modulo-congruences-det}(2), we obtain
$\big(\sfMon(\Sigma,\B)/_{\ker(\sfh_\cA)}\big)/_{\rho_\cA} \cong \sfMon(\Sigma,\B)/_{\sim_{\sem{\cA}}}$ and thus, by Observation~\ref{obs:isomorphism-preserves-degree}, we have 
 \(
\deg\Big(\big(\sfMon(\Sigma,\B)/_{\ker(\sfh_\cA)}\big)/_{\rho_\cA}\Big) = \deg\big( \sfMon(\Sigma,\B)/_{\sim_{\sem{\cA}}} \big) \enspace.
\)
     \end{proof}

\subsection{Construction of a bu-deterministic wta from a weighted tree language, a congruence, and a finite pair-independent generating set}
\label{sec:constr-bu-det-wta-from-congurence+pair-independent-gen-set}

As a preparation for the proof of Theorem~\ref{thm:budRec-mono-det}(B)$\Rightarrow$(A), here, for every
  \begin{compactitem}
  \item weighted tree language $r:\T_\Sigma \to B$,
  \item congruence $\sim$ on the $(\Sigma,\B)$-scalar algebra $(\rmMon(\Sigma,\B),\widetilde{\0},\ttop)$ such that $\sim$ saturates $r$ and $\sfMon(\Sigma,\B)/_\sim$ is finitely generated and cancellative, and
  \item finite scalar basis $H$ of $\sfMon(\Sigma,\B)/_\sim$ such that $H \subseteq  \{[\1.\zeta]_\sim \mid \zeta \in \T_\Sigma\}$,
       \end{compactitem}
we define a bu-deterministic $(\Sigma,\B)$-wta  $\budwta(r,\sim,H)=(H,\delta,F)$ and show that it recognizes $r$ (cf. Definition~\ref{def:wta(r,sim,H)} and Lemma \ref{lm:wta(r,sim,H)}(1) and (2)). Moreover, we show that, if  the decomposition mapping for $\big(\rmMon(\Sigma,\B)/_\sim\big) \setminus \{[\widetilde{\0}]_\sim\}$ is computable,
then we can even construct $\budwta(r,\sim,H)$ (cf. Lemma \ref{lm:wta(r,sim,H)}(3)).  
    
\begin{quote} {\em In the rest of this subsection,  we abbreviate $[\1.\xi]_\sim$ by $[\1.\xi]$ for each $\xi\in \T_\Sigma$.}
\end{quote}

Using the decomposition mapping 
\[
\mathrm{dec}: \big(\rmMon(\Sigma,\B)/_\sim\big) \setminus \{[\widetilde{\0}]\} \to B^{-\0} \times \big(H \setminus \{[\widetilde{\0}]\}\big)
\]
for $V=\rmMon(\Sigma,\B)/_\sim$ as it is defined in \eqref{equ:scalar-vector-det-new},
we have for each $[\1.\xi]\ne[\widetilde{\0}]$:
\begin{equation}\label{equ:decomposition-of-class}
  [\1.\xi] = \mathrm{scal}([\1.\xi]) \cdot \mathrm{gen}([\1.\xi]) \enspace.
  \end{equation}
Then, for every $\xi \in \T_\Sigma$ and $[\1.\zeta] \in H$, we introduce the notation  $[\1.\xi]_{[\1.\zeta]}$ for the value in $B$ defined by
\begin{equation}\label{equ:definition-of-coefficient-det-new-new}
  [\1.\xi]_{[\1.\zeta]}=\begin{cases} \mathrm{scal}([\1.\xi]) & \text{ if $[\widetilde{\0}]\ne[\1.\xi]$ and $\mathrm{gen}([\1.\xi]) = [\1.\zeta]$}\\
\0 & \text{ otherwise.}
\end{cases}
\end{equation}
We note that the well-definedness of the mapping $\mathrm{dec}$, and hence of the notation  $[\1.\xi]_{[\1.\zeta]}$, is due to Lemma~\ref{lm:crucial-canc-pair-ind-imply-uniqueness-det}(2).

\begin{definition}\rm \label{def:wta(r,sim,H)} Let $r$, $\sim$, and $H$ be given as in the above list of objects.  Moreover, let $\gamma: \rmMon(\Sigma,\B)/_\sim \to B$ be the scalar-linear form which is uniquely determined by $\sim$ and $r$ (cf. Lemma~\ref{lm:uniqueness-of-scalar-linear-form}).  We define the $(\Sigma,\B)$-wta $\budwta(r,\sim,H)=(H,\delta,F)$ where
  \begin{compactitem}
\item $\delta = (\delta_k:H^k\times \Sigma^{(k)}\times H \to B\mid k \in \mathbb{N})$ and for every $k \in \mathbb{N}$, $\sigma \in \Sigma^{(k)}$, $[\1.\zeta_{1}],\ldots,[\1.\zeta_{k}],[\1.\zeta]\in H$,  we let
  \begin{equation}\label{eq:delta-k-definition-det-new}
   \delta_k\big([\1.\zeta_{1}]\cdots[\1.\zeta_{k}],\sigma,[\1.\zeta]\big) = [\1.\sigma(\zeta_{1},\ldots,\zeta_{k})]_{[\1.\zeta]}
          \end{equation}
        \item $F: H \to B$ such that, for each $[\1.\zeta] \in H$, we define $F_{[\1.\zeta]} = \gamma([\1.\zeta])$. 
        \end{compactitem}
                        \hfill $\Box$
  \end{definition}

  \begin{lemma}\rm \label{lm:wta(r,sim,H)}  Let $r$, $\sim$, and $H$ be given as in Definition \ref{def:wta(r,sim,H)}. Then the following statements hold.
    \begin{compactenum}
    \item[(1)] The $(\Sigma,\B)$-wta $\budwta(r,\sim,H)$ is bu-deterministic.
      \item[(2)] $\sem{\budwta(r,\sim,H)}=r$.
      \item[(3)] If  the decomposition mapping for $\big(\rmMon(\Sigma,\B)/_\sim\big) \setminus \{[\widetilde{\0}]\}$ is computable, then we can  construct $\budwta(r,\sim,H)$.
              \end{compactenum}
    \end{lemma}
  \begin{proof} Proof of (1):  It follows from the fact that, due to \eqref{equ:definition-of-coefficient-det-new-new}, in the definition~\eqref{eq:delta-k-definition-det-new} there exists at most one $[\1.\zeta] \in H$ with $[\1.\sigma(\zeta_{1},\ldots,\zeta_{k})]_{[\1.\zeta]}\ne\0$.

      Proof of (2): Let $\budwta(r,\sim,H) = (H,\delta,F)$. We abbreviate $\budwta(r,\sim,H)$ by $\cA$.  By induction on $\T_\Sigma$, we can prove the following statement.
    \begin{equation}\label{equ:hom-is-canonical-det-new}
\text{For each $\xi \in \T_\Sigma$ and $[\1.\zeta] \in H$, we have $\h_{\cA}(\xi)_{[\1.\zeta]} = [\1.\xi]_{[\1.\zeta]}$.}
\end{equation}

Let $\xi = \sigma(\xi_1,\ldots,\xi_k)$ and $[\1.\zeta] \in H$.  Then we can calculate as follows.
\begingroup
\allowdisplaybreaks
\begin{align*}
  & \ \ \h_{\cA}(\sigma(\xi_1,\ldots,\xi_k))_{[\1.\zeta]}\\
  &= \bigoplus_{[\1.\zeta_1]\cdots [\1.\zeta_k] \in H^k}  \h_{\cA}(\xi_1)_{[\1.\zeta_1]} \otimes \ldots \otimes \h_{\cA}(\xi_k)_{[\1.\zeta_k]} \otimes \delta_k([\1.\zeta_1] \cdots [\1.\zeta_k], \sigma, [\1.\zeta])\\
  &=\bigoplus_{[\1.\zeta_1]\cdots [\1.\zeta_k] \in H^k}  [\1.\xi_1]_{[\1.\zeta_1]} \otimes \ldots \otimes [\1.\xi_k]_{[\1.\zeta_k]} \otimes
    \delta_k([\1.\zeta_1] \cdots [\1.\zeta_k], \sigma, [\1.\zeta])
    \tag{by I.H.}\\
   &=\bigoplus_{[\1.\zeta_1]\cdots [\1.\zeta_k] \in H^k}  [\1.\xi_1]_{[\1.\zeta_1]} \otimes \ldots \otimes [\1.\xi_k]_{[\1.\zeta_k]} \otimes
    [\1.\sigma(\zeta_1, \ldots,\zeta_k)]_{[\1.\zeta]} \enspace.
    \tag{by the definition of $\delta_k$}
\end{align*}
\endgroup
If $k=0$, then we are ready. For $k\ge 1$ we continue by case analysis.

\underline{case (a):} There exists $i \in [k]$ such that $[\1.\xi_i]= [\widetilde{\0}]$, i.e.,  for each $[\1.\zeta] \in H$, we have $[\1.\xi_i]_{[\1.\zeta]} = \0$. Then  we can continue as follows:
 \begingroup
\allowdisplaybreaks
\begin{align*}
  &\bigoplus_{[\1.\zeta_1]\cdots [\1.\zeta_k] \in H^k}  [\1.\xi_1]_{[\1.\zeta_1]} \otimes \ldots \otimes [\1.\xi_k]_{[\1.\zeta_k]} \otimes [\1.\sigma(\zeta_1, \ldots,\zeta_k)]_{[\1.\zeta]} = \0\\
  &= \Big(\ttop/\!_\sim(\sigma)([\1.\xi_1],\ldots,[\1.\xi_{i-1}],[\widetilde{\0}], [\1.\xi_{i+1}],\ldots,[\1.\xi_k])\Big)_{[\1.\zeta]}\tag{by (\ref{ml-absorbtive-det}) for $\ttop/\!_\sim(\sigma)$}\\
      &= \Big(\ttop/\!_\sim(\sigma)([\1.\xi_1],\ldots,[\1.\xi_k])\Big)_{[\1.\zeta]} = [\1.\sigma(\xi_1,\ldots,\xi_k)]_{[\1.\zeta]} \enspace.
\end{align*}
\endgroup
  
\underline{case (b):}  For each $i \in [k]$ we have $[\1.\xi_i]\ne [\widetilde{\0}]$; thus, by \eqref{equ:decomposition-of-class}, we have $[\1.\xi_i] = \mathrm{scal}([\1.\xi_i]) \cdot \mathrm{gen}([\1.\xi_i])$. Since $\mathrm{gen}([\1.\xi_i]) \in H$, there exists $\widehat{\zeta_i} \in \T_\Sigma$ such that $\mathrm{gen}([\1.\xi_i]) = [\1.\widehat{\zeta_i}]$. Then, by \eqref{equ:definition-of-coefficient-det-new-new},  $[\1.\xi_i] = [\1.\xi_i]_{[\1.\widehat{\zeta_i}]}\cdot [\1.\widehat{\zeta_i}]$ for each $i \in [k]$. Now we can continue as follows:
   \begingroup
\allowdisplaybreaks
\begin{align*}
  &\bigoplus_{[\1.\zeta_1]\cdots [\1.\zeta_k] \in H^k}  [\1.\xi_1]_{[\1.\zeta_1]} \otimes \ldots \otimes [\1.\xi_k]_{[\1.\zeta_k]} \otimes [\1.\sigma(\zeta_1, \ldots,\zeta_k)]_{[\1.\zeta]}\\[2mm]
  &=   [\1.\xi_1]_{[\1.\widehat{\zeta_1}]} \otimes \ldots \otimes [\1.\xi_k]_{[\1.\widehat{\zeta_k}]} \otimes
    [\1.\sigma(\widehat{\zeta}_1,\ldots,\widehat{\zeta}_k)]_{[\1.\zeta]}
    \tag{because $ [\1.\xi_i]_{[\1.\zeta']}=\0 $ if $[\1.\zeta']\ne [\1.\widehat{\zeta}_i]$ for $i\in [k]$}\\
   &= [\1.\xi_1]_{[\1.\widehat{\zeta_1}]} \otimes \ldots \otimes [\1.\xi_k]_{[\1.\widehat{\zeta_k}]} \otimes
     \big(\ttop/\!_\sim(\sigma)([\1.\widehat{\zeta}_1],\ldots,[\1.\widehat{\zeta}_k])\big)_{[\1.\zeta]}
  \tag{by definition of $\ttop/\!_\sim(\sigma)$}\\
   &= \Big( \big([\1.\xi_1]_{[\1.\widehat{\zeta_1}]} \otimes \ldots \otimes [\1.\xi_k]_{[\1.\widehat{\zeta_k}]}\big) \cdot
     \ttop/\!_\sim(\sigma)([\1.\widehat{\zeta}_1],\ldots,[\1.\widehat{\zeta}_k]) \Big)_{[\1.\zeta]}
  \tag{by axioms for scalar multiplication of weighted tree languages}\\
  &= \Big( \ttop/\!_\sim(\sigma)([\1.\xi_1]_{[\1.\widehat{\zeta_1}]}\cdot [\1.\widehat{\zeta_1}],\ldots,
    [\1.\xi_k]_{[\1.\widehat{\zeta_k}]}\cdot [\1.\widehat{\zeta_k}]\Big)_{[\1.\zeta]}
  \tag{by \eqref{ml-Omega-det} for $\ttop/\!_\sim(\sigma)$}\\
  &= \Big(\ttop/\!_\sim(\sigma)([\1.\xi_1],\ldots,[\1.\xi_k])\Big)_{[\1.\zeta]}
    \tag{because  $[\1.\xi_i] = [\1.\xi_i]_{[\1.\widehat{\zeta_i}]}\cdot [\1.\widehat{\zeta_i}]$ for each $i \in [k]$}\\
  &= [\1.\sigma(\xi_1,\ldots,\xi_k)]_{[\1.\zeta]} \enspace.
     \tag{by definition of $\ttop/_\sim(\sigma)$}
  \end{align*}
  \endgroup
  This finishes the proof of \eqref{equ:hom-is-canonical-det-new}.
Now let $\xi \in \T_\Sigma$.   By Lemma \ref{lm:properties-hA-of-budet-wta-det-new}(4), we have
       \begingroup
    \allowdisplaybreaks
    \begin{align*}
      \sem{\cA}(\xi) &= \begin{cases}
                \h_{\cA}(\xi)_{\state(\xi)} \otimes F_{\state(\xi)} & \text{ if $\state(\xi) \in Q$}\\
               \0        & \text{ otherwise}
             \end{cases}                           
                           \end{align*}
                           \endgroup
                           We continue by case analysis.
                           
    \underline{$[\1.\xi]=[\widetilde{\0}]$:}
Then by \eqref{equ:hom-is-canonical-det-new} we have $\h_\cA(\xi)=\0^Q$ and thus
  $\sem{\cA}(\xi) = \0  = \gamma([\widetilde{\0}]) = \gamma([\1.\xi]) = r(\xi)$
    where the last equality holds because $\sim$ saturates $r$ via the scalar-linear form $\gamma$.

 \underline{$[\1.\xi] \ne [\widetilde{\0}]$:} Similarly as in case (b) above, there exists $\widehat{\zeta} \in \T_\Sigma$ such that $[\1.\widehat{\zeta}]\in H$ and $[\1.\xi] = [\1.\xi]_{[\1.\widehat{\zeta}]}\cdot [\1.\widehat{\zeta}]$. Since $[\1.\xi]_{[\1.\widehat{\zeta}]}\ne \0$, by  \eqref{equ:hom-is-canonical-det-new} and Lemma \ref{lm:properties-hA-of-budet-wta-det-new}(2), we have $\state(\xi)=[\1.\widehat{\zeta}]$.
 Then  we can continue as follows:
     \begingroup
    \allowdisplaybreaks
    \begin{align*}
      \sem{\cA}(\xi) &= \h_{\cA}(\xi)_{\state(\xi)} \otimes F_{\state(\xi)}= \h_{\cA}(\xi)_{[\1.\widehat{\zeta}]} \otimes F_{[\1.\widehat{\zeta}]} \\
                        &= [\1.\xi]_{[\1.\widehat{\zeta}]} \otimes \gamma([\1.\widehat{\zeta}]) 
      \tag{by \eqref{equ:hom-is-canonical-det-new} and definition of $F$}\\
                     &= \gamma\Big([\1.\xi]_{[\1.\widehat{\zeta}]} \cdot [\1.\widehat{\zeta}]\Big) 
                           \tag{because $\gamma$ is a scalar-linear form}\\
      &= \gamma([\1.\xi]) 
      \tag{by the above}\\
                          &= r(\xi)
                            \tag{because $\sim$ saturates $r$ via the scalar-linear form $\gamma$}\enspace. 
    \end{align*}
    \endgroup
        Proof of (3): Let $k \in \mathbb{N}$, $\sigma \in \Sigma^{(k)}$, and $[\1.\zeta_{1}],\ldots,[\1.\zeta_{k}],[\1.\zeta]\in H$. We show that we can compute the right-hand side of \eqref{eq:delta-k-definition-det-new}. For this, we decide if  $[\1.\sigma(\zeta_{1},\ldots,\zeta_{k})]=[\widetilde{\0}]$, i.e., if $\1.\sigma(\zeta_{1},\ldots,\zeta_{k})\sim \widetilde{\0}$. (This is possible because, by our general convention, $\sim$ is given effectively.)  If the answer is yes, then $[\1.\sigma(\zeta_{1},\ldots,\zeta_{k})]_{[\1.\zeta]}=\0$. Otherwise,
  we compute $\mathrm{scal}([\1.\sigma(\zeta_{1},\ldots,\zeta_{k})])$ and $\mathrm{gen}([\1.\sigma(\zeta_{1},\ldots,\zeta_{k})])$.
  Now if $\mathrm{gen}([\1.\sigma(\zeta_{1},\ldots,\zeta_{k})])=[\1.\zeta]$, then 
  $[\1.\sigma(\zeta_{1},\ldots,\zeta_{k})]_{[\1.\zeta]}=\mathrm{scal}([\1.\sigma(\zeta_{1},\ldots,\zeta_{k})])$,
  otherwise again  $[\1.\sigma(\zeta_{1},\ldots,\zeta_{k})]_{[\1.\zeta]}=\0$.
  
  Moreover, let $[\1.\zeta] \in H$. Then we can compute $F_{[\1.\zeta]}$ because $\gamma([\1.\zeta])=r(\zeta)$ and by our convention 
  the mapping $r$ is computable.
          \end{proof}

\begin{example}\rm \label{ex:bu-det+total-wta3-det} We illustrate the construction of $\budwta(r,\sim,H)$.
For this, we consider the $(\Sigma,\Ratnum)$-weighted tree language
$r: \T_\Sigma \to \mathbb{Q}$ of Example \ref{ex:bu-det+total-wta-det} and the m-syntactic congruence $\sim_r$ (cf. Example~\ref{ex-syntactic-congruence}). By Lemma~\ref{lm:sim-r-is-coarsest}, $\sim_r$ saturates $r$ via the scalar-linear form $\gamma: \rmMon(\Sigma,\Ratnum)/_{\sim_r} \to \mathbb{Q}$
  defined by
  \(
\gamma([b.\xi]_{\sim_r}) = b \cdot r(\xi) \ \text{ for each $b.\xi \in \rmMon(\Sigma,\Ratnum)$}\).
In the sequel, we abbreviate $[b.\xi]_{\sim_r}$ by $[b.\xi]$.

    By Lemma~\ref{lm:MonSigmaB-modulo-simr-cancellative-det}, the $\B$-scalar algebra $\sfMon(\Sigma,\Ratnum)/_{\sim_r}$ is cancellative.
In order to show that $\sfMon(\Sigma,\Ratnum)/_{\sim_r}$ is finitely generated, we consider the set
  \(
H = \{[1.\alpha], \ [1.\sigma(\alpha,\alpha)] \} 
\)
and show that it is a generating set of $\sfMon(\Sigma,\Ratnum)/_{\sim_r}$. By using the characterization of $\sim_r$ given in  Example  \ref{ex-syntactic-congruence}, 
for every $\xi \in \T_\Sigma$ and $b\in \mathbb{Q}$, we have
\begin{align}\label{eq:congruence-classes}
[b.\xi] = 
\begin{cases}
[(b\cdot 2^{\#_\alpha(\xi)-1}).\alpha]=(b\cdot 2^{\#_\alpha(\xi)-1})\cdot[1.\alpha] & \text{ if \ $\big(\#_\alpha(\xi)\!\mod(2)\big)=1$}\\
[(b\cdot 2^{\#_\alpha(\xi)-2}).\sigma(\alpha,\alpha)]=(b\cdot 2^{\#_\alpha(\xi)-2})\cdot[1.\sigma(\alpha,\alpha)] & \text{ if \ $\big(\#_\alpha(\xi)\!\mod(2)\big)=0$}.
\end{cases}
\end{align}
This shows that $H$ generates $\sfMon(\Sigma,\Ratnum)/_{\sim_r}$. By \eqref{eq:congruence-classes}, we also obtain that there does not exist $b \in \mathbb{Q}$ such that $[1.\alpha] = b \cdot [1.\sigma(\alpha,\alpha)]$ because $\big(\#_\alpha(\alpha) \! \mod(2)\big) \not = \big(\#_\alpha(\sigma(\alpha,\alpha)) \! \mod (2)\big)$. Thus  $H$ is pair-independent, and hence $H$ is a finite scalar-basis.
Hence $r$, $\sim_r$, and $H$ satisfy all the conditions of  Definition \ref{def:wta(r,sim,H)}.

Moreover, \eqref{eq:congruence-classes} also  shows that  the decomposition mapping $\mathrm{dec}$ for $\big(\rmMon(\Sigma,\Ratnum)/_{\sim_r}\big) \setminus \{[\widetilde{\0}]\}$ is computable. In particular, for each $\xi\in\T_\Sigma$, we have
\begin{align}\label{eq:dec-mapping}
\mathrm{dec}([1.\xi])=
\begin{cases} (2^{\#_\alpha(\xi)-1}, [1.\alpha]) & \text{ if \ $\big(\#_\alpha(\xi)\!\mod(2)\big)=1$}\\
(2^{\#_\alpha(\xi)-2}, [1.\sigma(\alpha,\alpha)]) & \text{ if \ $\big(\#_\alpha(\xi)\!\mod(2)\big)=0$}.
\end{cases}
\end{align}

Hence, by  Lemma \ref{lm:wta(r,sim,H)}(3), we can construct the $(\Sigma,\Ratnum)$-wta $\budwta(r,\sim_r,H) =(H,\delta,F)$ as follows:
\begin{compactitem}
\item $\delta_0(\varepsilon,\alpha,[1.\alpha]) = 1$, \
  $\delta_0(\varepsilon,\alpha,[1.\sigma(\alpha,\alpha)]) = 0$, and for every $q_1q_2 \in H$ we have
  \[
    \delta_2(q_1q_2,\sigma,[1.\alpha]) = 
    \begin{cases} 4  & \text{ if $q_1\not= q_2$}\\
           0 & \text{ otherwise }
           \end{cases}
    \]
        and
        \[
       \delta_2(q_1q_2,\sigma,[1.\sigma(\alpha,\alpha)]) = 
        \begin{cases} 4 & \text{ if $q_1= q_2= [1.\sigma(\alpha,\alpha)]$}\\
         1 & \text{ if $q_1= q_2= [1.\alpha]$}\\
          0 & \text{ otherwise }\enspace.
        \end{cases}
        \]
       For instance $\delta_2([1.\sigma(\alpha,\alpha)][1.\alpha],\sigma,[1.\alpha])=4$ because
       $[1.\sigma(\sigma(\alpha,\alpha),\alpha)]_{[1.\alpha]}=\mathrm{scal}([1.\sigma(\sigma(\alpha,\alpha),\alpha)])=2^{3-1}=4$.
       
        \item $F_{[1.\alpha]} = \gamma([1.\alpha]) = r(\alpha) = 6$ and
  $F_{[1.\sigma(\alpha,\alpha)]} = \gamma([1.\sigma(\alpha,\alpha)]) = r(\sigma(\alpha,\alpha)) = 8$.
\end{compactitem}
By Lemma \ref{lm:wta(r,sim,H)}(1), $\budwta(r,\sim_r,H)$ is a bu-deterministic $(\Sigma,\B)$-wta; it is even total.  By Lemma~\ref{lm:wta(r,sim,H)}(2), we have $\sem{\budwta(r,\sim_r,H)}=r$.
    \hfill$\Box$
\end{example}

 \begin{figure}[t]
   \begin{center}
\begin{tikzpicture}
\tikzset{node distance=7em, scale=0.4, transform shape}
\node[state, rectangle] (1) {\Large $\alpha$};
\node[state, right of=1] (2){\Large $[1.\alpha]$};
\node[state, rectangle, right of=2] (3)[right=1em]{\Large $\sigma$};
\node[state, rectangle, above of=3] (4)[above=1em]{\Large $\sigma$};
\node[state, rectangle, below of=3] (5)[below=1em]{\Large $\sigma$};
\node[state, right of=3] (6)[right=1em]{\Large $[1.\sigma(\alpha,\alpha)]$};
\node[state, rectangle, right of=6] (7) {\Large $\sigma$};

\tikzset{node distance=2em}
\node[above of=1] (w1)[yshift=0.7em] {\Large 1};
\node[above of=2] (w2)[left=0.1em,yshift=1em] {\Large 6};
\node[above of=3] (w3)[yshift=0.7em] {\Large 1};
\node[above of=4] (w4)[yshift=0.7em] {\Large 4};
\node[above of=5] (w5)[yshift=0.7em] {\Large 4};
\node[above of=6] (w6)[above=1.8em] {\Large 8};
\node[above of=7] (w7)[yshift=0.7em] {\Large 4};

\draw[->,>=stealth] (1) edge (2);
\draw[->,>=stealth] (2) edge[out=20, in=155, looseness=1.1] (3);
\draw[->,>=stealth] (2) edge[out=-20, in=205, looseness=1.1] (3);
\draw[->,>=stealth] (2) edge (4);
\draw[->,>=stealth] (2) edge (5);
\draw[->,>=stealth] (3) edge (6);
\draw[->,>=stealth] (6) edge (4);
\draw[->,>=stealth] (4) edge[out=110, in=80, looseness=1.4] (2);
\draw[->,>=stealth] (6) edge (5);
\draw[->,>=stealth] (5) edge[out=250, in=-80, looseness=1.4] (2);
\draw[->,>=stealth] (7) edge (6);
\draw[->,>=stealth] (6) edge[out=60, in=30, looseness=2.7] (7);
\draw[->,>=stealth] (6) edge[out=-60, in=-30, looseness=2.7] (7);
\end{tikzpicture} 
\caption{\label{fig:even-odd-two-three-3-det} The $(\Sigma,\Ratnum)$-wta $\budwta(r,\sim_r,H) = (H,\delta,F)$.}
   \end{center}
   \end{figure}
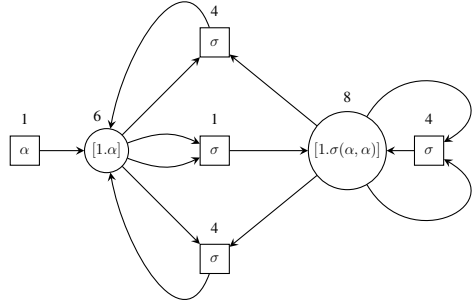

\subsection{The deterministic version of B-A's theorem}
\label{subsect:det-version-BA-theorem}

Here we present our first main result, which might be called the deterministic version of B-A's theorem.

\begin{samepage}
  \begin{theorem}\label{thm:budRec-mono-det}  Let $\B$  a commutative semifield. Moreover, let $r: \T_\Sigma \to B$. The following three statements are equivalent.
  \begin{compactenum}
  \item[(A)] $r \in \budRec(\Sigma,\B)$.
  \item[(B)] There exists a congruence $\sim$ on the $(\Sigma,\B)$-scalar algebra $(\rmMon(\Sigma,\B),\widetilde{\0},\ttop)$ such that  $\sim$ respects~$r$ and the $\B$-scalar algebra $\sfMon(\Sigma,\B)/_\sim$ is finitely generated and cancellative.
  \item[(C)]  The $\B$-scalar algebra $\sfMon(\Sigma,\B)/_{\sim_r}$ is finitely generated.
  \end{compactenum}
  \end{theorem}
\end{samepage}
\begin{proof}  Proof of  (A)$\Rightarrow$(C): Let $\cA$ be a bu-deterministic $(\Sigma,\B)$-wta with $\sem{\cA} = r$.  Then the implication follows from Theorem \ref{lm:dimension-synt-factor-algebra-bounded}.
    
  Proof of (C)$\Rightarrow$(B): We choose $\sim\,=\,\sim_r$.
  Then  $\sim$ is a congruence on $(\rmMon(\Sigma,\B),\widetilde{\0},\ttop)$  and, by Lemma~\ref{lm:sim-r-is-coarsest}, it saturates $r$. By Lemma~\ref{lm:MonSigmaB-modulo-simr-cancellative-det}, the $\B$-scalar algebra $\sfMon(\Sigma,\B)/_{\sim}$ is cancellative, and by our choice it is finitely generated.
    
Proof of (B)$\Rightarrow$(A): Let $\B=(B,\oplus,\otimes,\0,\1)$. Moreover, let $\sim$ be a congruence on  $(\rmMon(\Sigma,\B),\widetilde{\0},\ttop)$ such that $\sim$ saturates $r$ via the scalar-linear form $\gamma: \rmMon(\Sigma,\B)/_{\sim}\to B$, i.e., for each $\xi \in \T_\Sigma$, we have $r(\xi) = \gamma([\1.\xi]_\sim)$. Moreover, let  $\sfMon(\Sigma,\B)/_{\sim}$ be cancellative and generated by the finite set $H$.  Without loss of generality, we can assume that $H \subseteq \{[\1.\zeta]_\sim \mid \zeta \in \T_\Sigma\}$.

By Lemma \ref{lm:crucial-canc-pair-ind-imply-uniqueness-det}(1) there exists a pair-independent subset $H' \subseteq H$ which generates $\sfMon(\Sigma,\B)/_{\sim}$ and it is a finite scalar-basis because $\sfMon(\Sigma,\B)/_\sim$ is cancellative.
Then, by Lemma \ref{lm:wta(r,sim,H)}, the $(\Sigma,\B)$-wta $\budwta(r,\sim,H')$ (defined in Definition \ref{def:wta(r,sim,H)}) is bu-deterministic  and  $\sem{\budwta(r,\sim,H')}=r$.  
    \end{proof}

\begin{example}\rm \label{ex:number-of-gammas}  We consider the ranked alphabet $\Sigma$ and the $(\Sigma,\Ratnum)$-weighted tree language $r$ of Example~\ref{ex:wtl-number-of-occurrences-of-gamma-finite-basis}, i.e., $r=\#_\gamma$. We prove that the $\Ratnum$-scalar algebra $\sfMon(\Sigma,\Ratnum)/_{\sim_r}$ is not finitely generated. This means  $r\not\in \budRec(\Sigma,\Ratnum)$  (cf. Theorem  \ref{thm:budRec-mono-det} (A)$\Leftrightarrow$(C)).
In the following we abbreviate $\sim_{\#_\gamma}$ by  $\sim_{\gamma}$.

First we give the following characterization of $\sim_{\gamma}$:
\begin{equation}\label{equ:char-of-sim-for-num-gamma}
    \begin{aligned}
    &\text{for every $b_1.\xi_1, b_2.\xi_2\in \mathrm{Mon}(\Sigma,\Ratnum)$, we have $b_1.\xi_1 \sim_{\gamma} b_2.\xi_2$ }\\
    &\text{if and only if }(b_1=b_2\ \text{ and } \ \#_\gamma(\xi_1)=\#_\gamma(\xi_2)) \text{ or } b_1=b_2=0\enspace.
      \end{aligned}
\end{equation} 
Let $b_1.\xi_1, b_2.\xi_2\in \mathrm{Mon}(\Sigma,\Ratnum)$
such that $b_1.\xi_1 \sim_{\gamma} b_2.\xi_2$. The proof of the if-part is obvious. 

We prove the only-if-part by case analysis.   

\underline{$b_1=0$ or $b_2=0$:} It follows from the definition of $ \sim_{\gamma}$ that
$b_1=b_2=0$, and hence the right-hand side of \eqref{equ:char-of-sim-for-num-gamma} holds. 

\underline{ $b_1\ne 0\ne b_2$:} Then
\begingroup
\allowdisplaybreaks
\begin{align*}
  &b_1.\xi_1 \sim_{\gamma} b_2.\xi_2\\[2mm]
  \text{iff } \ & (\forall c \in \C_\Sigma): b_1\cdot \#_\gamma(c[\xi_1]) =  b_2\cdot \#_\gamma(c[\xi_2])\\[2mm]
  \text{iff } \ & (\forall c \in \C_\Sigma): b_1\cdot \#_\gamma(c) + b_1\cdot \#_\gamma(\xi_1)=b_2\cdot \#_\gamma(c) + b_2\cdot \#_\gamma(\xi_2)\\[2mm]
  \text{iff } \ & (\forall n \in \mathbb{N}): b_1\cdot n + a_1 = b_2\cdot n + a_2 \tag{where $a_i=b_i\cdot \#_\gamma(\xi_i)$ for $i\in\{1,2\}$}\\[2mm]
  \text{iff } \ &  b_1=b_2 \text{ and }  a_1=a_2\\[2mm]
  \text{iff } \ &  b_1=b_2 \text{ and } \#_\gamma(\xi_1)=\#_\gamma(\xi_2)\enspace.
 \end{align*}
  \endgroup 
  This finishes the proof of \eqref{equ:char-of-sim-for-num-gamma}. It follows from \eqref{equ:char-of-sim-for-num-gamma} that, for every $b\in \mathbb{Q}\setminus\{0\}$ and $\xi\in\T_\Sigma$, we have 
  \begin{equation}\label{eq:eq-class-of-sim-gamma}
  [b.\xi]_{\sim_{\#_\gamma}}=\{b.\zeta\mid \xi\in\T_\Sigma, \#_\gamma(\xi)=\#_\gamma(\zeta)\}.
  \end{equation}
  
  Now we show by contradiction that $\sfMon(\Sigma,\Ratnum)/_{\sim_{\gamma}}$ is not finitely generated. For this, assume that
  there exists an $n\in\mathbb{N}$ and a finite subset $H=\{[b_1.\xi_1]_{\sim_{\gamma}},\ldots,[b_n.\xi_n]_{\sim_{\gamma}}\}$ of $\mathrm{Mon}(\Sigma,\Ratnum)/_{\sim_{\gamma}}$ such that $H$ generates $\sfMon(\Sigma,\Ratnum)/_{\sim_{\gamma}}$. We may assume without loss of generality that $b_i\ne 0$ for each $i\in[n]$. Now let $b\in \mathbb{Q}\setminus\{0\}$ and $\zeta\in \T_\Sigma$ such that
  $\#_\gamma(\zeta)\not \in \{\#_\gamma(\xi_1),\ldots,\#_\gamma(\xi_n)\}$. Then it follows from \eqref{eq:eq-class-of-sim-gamma}
  that there does not exist $a\in \mathbb{Q}$ and $i\in[n]$ such that $[b.\zeta]_{\sim_{\gamma}}=a\cdot [b_i.\xi_i]_{\sim_{\gamma}}$. This contradicts the assumption that $H$ generates $\sfMon(\Sigma,\Ratnum)/_{\sim_{\gamma}}$. Hence the $\Ratnum$-scalar algebra 
$\sfMon(\Sigma,\Ratnum)/_{\sim_{\gamma}}$ is not finitely generated.
 \hfill $\Box$
\end{example}

 
\subsection{Minimal bu-deterministic wta and its construction}
\label{subs:minimization}

Here we show how to construct, for a given bu-deterministic $(\Sigma,\B)$-wta and equivalent one which is minimal.

  {\em In this subsection $\cA=(Q,\delta,F)$ denotes an arbitrary bu-deterministic $(\Sigma,\B)$-wta, if not specified otherwise.}

We say that $\cA$ is \emph{minimal} if, for each bu-deterministic $(\Sigma,\B)$-wta  $\cB=(Q',\delta',F')$ with $\sem{\cA}=\sem{\cB}$, we have $|Q|\le |Q'|$.
The next lemma can be compared to \cite[Thm.~3(a)]{bor03}.

 \begin{lemma}\label{cor:min-du-det-wta-exists}\rm For each finite scalar-basis  $H$ of $\sfMon(\Sigma,\B)/_{\sim_{\sem{\cA}}}$, the $(\Sigma,\B)$-wta $\budwta(\sem{\cA},\sim_{\sem{\cA}},H)$ is defined and it is minimal.
 \end{lemma}
 \begin{proof}  Let us abbreviate $\sem{\cA}$ by $r$. By Lemma \ref{lm:sim-r-is-coarsest} the congruence $\sim_r$ saturates $r$ and, by Lemma \ref{lm:MonSigmaB-modulo-simr-cancellative-det},  the $\B$-scalar algebra $\sfMon(\Sigma,\B)/_{\sim_r}$ is  cancellative. Now let $H$ be a finite scalar-basis of  $\sfMon(\Sigma,\B)/_{\sim_{r}}$.
 Then $\budwta(r,\sim_r,H)$ is defined (cf. Definition \ref{def:wta(r,sim,H)}) and, by Lemma \ref{lm:wta(r,sim,H)}, it is a bu-deterministic $(\Sigma,\B)$-wta and it is equivalent to $\cA$.

Next we show that $\budwta(r,\sim_r,H)$ is minimal.  We recall that $H$ is the set of states of $\budwta(r,\sim_r,H)$ and $|H|=\deg\big( \mathsf{Mon}(\Sigma,\B)/_{\sim_r}\big)$.  Let $\cB=(Q',\delta',F')$ be an arbitrary bu-deterministic $(\Sigma,\B)$-wta with  $r=\sem{\cB}$. By Theorem~\ref{lm:dimension-synt-factor-algebra-bounded}, we have $|H|\le |Q'|$.  Thus $\budwta(r,\sim_r,H)$ is minimal.
\end{proof}

In the following we show that we can construct a finite generating set $\mathrm{H}_{\cA}$ of $\sfMon(\Sigma,\B)/_{\sim_{\sem{\cA}}}$ (cf. Lemma \ref{lm:tq-has-minimal-height}). Hence, if dependency in  $\sfMon(\Sigma,\B)/_{\sim_{\sem{\cA}}}$ is decidable, then, by Lemma \ref{lm:crucial-canc-pair-ind-imply-uniqueness-det}(1) and because $\sfMon(\Sigma,\B)/_{\sim_{\sem{\cA}}}$ is cancellative, we can construct a finite scalar-basis $H$ of $\sfMon(\Sigma,\B)/_{\sim_{\sem{\cA}}}$.
Moreover, if the decomposition mapping 
\[\mathrm{dec}: \big(\rmMon(\Sigma,\B)/_{\sim_{\sem{\cA}}}\big) \setminus \{[\widetilde{\0}]_{\sim_{\sem{\cA}}}\} \to B^{-\0} \times \big(H \setminus \{[\widetilde{\0}]_{\sim_{\sem{\cA}}}\}\big)\]
for $\big(\rmMon(\Sigma,\B)/_{\sim_{\sem{\cA}}}\big) \setminus \{[\widetilde{\0}]_{\sim_{\sem{\cA}}}\}$ is computable, then by Lemma~\ref{lm:wta(r,sim,H)}(3) we can construct the minimal $(\Sigma,\B)$-wta $\budwta(\sem{\cA},\sim_{\sem{\cA}},H)$.
The construction  of the minimal $(\Sigma,\B)$-wta  proceeds in three steps and summarized in Theorem \ref{thm:minimization-theorem-new-2}.

\underline{Step 1:} We define the concept of a slim $(\Sigma,\B)$-wta and we construct a slim $(\Sigma,\B)$-wta which is equivalent to $\cA$ (cf. Theorem~\ref{lm:construction-of-simple-det}).

We say that $\cA$ is \emph{slim} if for each $q\in Q$ there exists a tree $\xi\in \T_\Sigma$ with $q=\state_\cA(\xi)$, i.e., $Q \subseteq \im(\state_\cA)$.

\begin{theorem}\label{lm:construction-of-simple-det} Let $\B$ be a commutative semifield. Moreover, let  $\cA$ be a bu-deterministic $(\Sigma,\B)$-wta. We can construct a slim $(\Sigma,\B)$-wta $\cA'$ such that $\sem{\cA}=\sem{\cA'}$.
\end{theorem}
\begin{proof}  
First we compute the set $\im(\state_\cA)$. For this, for each $i\in \mathbb{N}$, we define the sequence 
$P_0\subseteq P_1 \subseteq \ldots \subseteq (Q\cup\{\bot\})$ of sets by induction as follows:
$P_0=\{\theta_\cA(\sigma)()\mid \sigma \in \Sigma^{(0)}\}$ and 
\begin{align*}
  P_{i+i}=P_i\cup\{\theta_\cA(\sigma)(p_1,\ldots,p_k)\mid k\in \mathbb{N}_+, \sigma \in \Sigma^{(k)}, \text{ and } p_1,\ldots,p_k\in P_i\}.
\end{align*}
It is easy to show that, for each $i\in \mathbb{N}$, we have $P_i\subseteq \im(\state_\cA)$ and that, if $P_i=P_{i+1}$, then 
$P_{i+1}=P_{i+2}$. Clearly, there exists $i\in \mathbb{N}$ with $P_i=P_{i+1}$ and it is also easy to see that, if this is the case, then $P_i=\im(\state_\cA)$.  Since, for each $i\in \mathbb{N}$, we can construct $P_i$, we can find the smallest integer $i_0\in [n]$ which satisfies $P_{i_0}=P_{i_0+1}$. 

For the construction of $\cA'$, we distinguish the following two cases.

\underline{$\im(\state_\cA)=\{\bot\}$}: By Lemma \ref{lm:properties-hA-of-budet-wta-det-new}(4), we have
$\sem{\cA}(\xi)=\0$ for each $\xi\in\T_\Sigma$, i.e., $\sem{\cA}=\widetilde{\0}$. We construct the bu-deterministic $(\Sigma,\B)$-wta $\cA'=(\{p\},\delta',F')$ where,
 for every $k\in\mathbb{N}$ and $\sigma\in\Sigma^{(k)}$, we let $\delta'_k(p^k,\sigma,p)=\1$, and $F'_p=\0$.
Then $\cA'$ is slim because $\im(\state_{\cA'} )=\{p\}$, and it is obvious that $\sem{\cA'}=\widetilde{\0}$.

\underline{$\im(\state_\cA)\not=\{\bot\}$}: We construct the $(\Sigma,\B)$-wta $\cA'=(Q',\delta',F')$ where
\begin{compactitem}
\item $Q'=\im(\state_\cA)\cap Q$,
\item $\delta'=(\delta'_k \mid k\in \mathbb{N})$ and, for each $k\in \mathbb{N}$, the mapping $\delta'_k$ is the restriction of $\delta_k$ to $(Q')^k \times \Sigma^{(k)}\times Q'$, 
\item $F'=F|_{Q'}$.
\end{compactitem}

It is obvious that $\cA'$ is bu-deterministic. We show that $\cA'$ is slim. As a first step, we show that
\begin{equation}\label{equ:succ-equ}
\text{for every $k\in \mathbb{N}$, $w\in (Q')^k$, and $\sigma\in\Sigma^{(k)}$, we have $\mysucc_{\cA'}(w,\sigma)=\mysucc_{\cA}(w,\sigma)$.}
\end{equation}

Let $w=q_1\cdots q_k$ with $q_1,\ldots,q_k\in Q'$.
Obviously, $\mysucc_{\cA'}(w,\sigma)\subseteq\mysucc_{\cA}(w,\sigma)$ by the definition of $\delta'$. 
We show that the inclusion $\mysucc_{\cA}(w,\sigma)\subseteq \mysucc_{\cA'}(w,\sigma)$ also holds. 
This is trivial if $\mysucc_{\cA}(w,\sigma)=\emptyset$, therefore assume that $\mysucc_{\cA}(w,\sigma)=\{p\}$ for some $p\in Q$.
Then $\delta_k(w,\sigma,p)\ne \0$, hence $\theta_\cA(\sigma)(q_1,\ldots,q_k)=p$.
For each  $i\in[k]$, there exists $\xi_i\in \T_\Sigma$ such that
$q_i=\state_\cA(\xi_i)$. Then, for $\xi=\sigma(\xi_1,\ldots,\xi_k)$ we have $\state_\cA(\xi)=p$, i.e., $p\in Q'$.
By the definition of $\delta'$, we have $\delta'_k(w,\sigma,p)=\delta_k(w,\sigma,p)\ne \0$, hence 
$\mysucc_{\cA'}(w,\sigma)=\{p\}$.
This proves \eqref{equ:succ-equ}.

Using \eqref{equ:succ-equ}, we can easily prove that,
for every $k\in \mathbb{N}$, $q_1,\ldots,q_k \in (Q')_{\bot}$, and $\sigma\in\Sigma^{(k)}$, we have $\theta_{\cA'}(\sigma)(q_1,\ldots,q_k)=\theta_{\cA}(\sigma)(q_1,\ldots,q_k)$.
Then it follows that $\mathsf{S}(\cA')=\big(\im(\state_\cA)\cup\{\bot\},\theta_\cA\big)$. Thus
\begin{equation}\label{equ:state-equ}
\text{for every $\xi \in \T_\Sigma$, we have
     $\state_{\cA}(\xi) = \state_{\cA'}(\xi)$.}
\end{equation}
By \eqref{equ:state-equ} and the definition of $Q'$ we obtain that $\cA'$ is slim.  Moreover, 
       \begin{equation}\label{equ:in-Q-iff-in-Q-prime}
   \text{for every $\xi \in \T_\Sigma$, we have
     $\state_{\cA}(\xi)\in Q$ if and only if $\state_{\cA'}(\xi)\in Q'$},
            \end{equation}
because $\state_{\cA}(\xi)\in Q$ iff $\state_{\cA}(\xi)\ne\bot$  iff
  $\state_{\cA'}(\xi)\ne\bot$ iff $\state_{\cA'}(\xi)\in Q'$ and the second equivalence is justified by \eqref{equ:state-equ}.
  
Next we show that $\sem{\cA}=\sem{\cA'}$.  For this, by induction on $\T_\Sigma$, we show the following:
\begin{equation}\label{equ:hAxia=hAprimexiq-Z}
  \text{for every $\xi \in \T_\Sigma$ and $q\in Q'$, we have
   $\h_\cA(\xi)_{q}=\h_{\cA'}(\xi)_{q}$.}
\end{equation}
Let $\xi=\sigma(\xi_1,\ldots,\xi_k)$. We proceed by case analysis.

\underline{$(\forall i\in[k]): \state_{\cA'}(\xi_i)\in Q'$:}
\begingroup
\allowdisplaybreaks
\begin{align*}
\h_{\cA'}(\xi)_{q}
  &=  \bigoplus_{q_1 \cdots q_k \in (Q')^k} \Big( \bigotimes_{i\in [k]} \h_{\cA'}(\xi_i)_{q_i}\Big) \otimes \delta'_k(q_1\cdots q_k,\sigma,q)\\[2mm]
  &= \Big(\bigotimes_{i\in[k]}  \h_{\cA'}(\xi_i)_{\state_{\cA'}(\xi_i)}\Big) \otimes \delta'_k(\state_{\cA'}(\xi_1)\cdots \state_{\cA'}(\xi_k),\sigma,q)
    \tag{because by Lemma  \ref{lm:properties-hA-of-budet-wta-det-new}(2), for each $i\in[k]$, we have $\h_{\cA'}(\xi_i)_{q_i}\ne \0$ iff $q_i=\state_{\cA'}(\xi_i)$}  \\[2mm]
  &= \Big(\bigotimes_{i\in[k]}  \h_{\cA'}(\xi_i)_{\state_{\cA}(\xi_i)}\Big) \otimes \delta'_k(\state_{\cA}(\xi_1)\cdots \state_{\cA}(\xi_k),\sigma,q)
    \tag{by \eqref{equ:state-equ}}\\[2mm]
&= \Big(\bigotimes_{i\in[k]}  \h_{\cA}(\xi_i)_{\state_{\cA}(\xi_i)}\Big) \otimes \delta'_k(\state_{\cA}(\xi_1)\cdots \state_{\cA}(\xi_k),\sigma,q) 
      \tag{by I.H.}\\[2mm]
&= \Big(\bigotimes_{i\in[k]}  \h_{\cA}(\xi_i)_{\state_{\cA}(\xi_i)}\Big) \otimes \delta_k(\state_{\cA}(\xi_1)\cdots \state_{\cA}(\xi_k),\sigma,q) 
          \tag{by the definition of $\delta'_k$}\\[2mm]
&=  \bigoplus_{q_1 \cdots q_k \in Q^k} \Big( \bigotimes_{i\in [k]} \h_{\cA}(\xi_i)_{q_i}\Big) \otimes \delta_k(q_1\cdots q_k,\sigma,q) \tag{by Lemma  \ref{lm:properties-hA-of-budet-wta-det-new}(2)}\\[2mm]
  &= \h_{\cA}(\xi)_{q} \enspace.
\end{align*}
\endgroup

\underline{$(\exists j \in[k]): \state_{\cA'}(\xi_j)\not\in Q'$:} By Lemma  \ref{lm:properties-hA-of-budet-wta-det-new}(3), we have $\h_{\cA'}(\xi_j) = \0^Q$. 
Moreover, by \eqref{equ:in-Q-iff-in-Q-prime}, we have $\state_{\cA}(\xi_j)\not\in Q$ and thus, by Lemma~\ref{lm:properties-hA-of-budet-wta-det-new}(3), we have $\h_{\cA}(\xi_j) = \0^Q$. 
Then  $\h_{\cA'}(\xi) = \0^Q = \h_{\cA'}(\xi)$. 
This proves \eqref{equ:hAxia=hAprimexiq-Z}.

Now we are ready to show that $\sem{\cA}=\sem{\cA'}$. For this, let $\xi \in \T_\Sigma$.  We proceed by case analysis.

   \underline{$\state_{\cA'}(\xi)\in Q'$:} By \eqref{equ:in-Q-iff-in-Q-prime}, we have $\state_{\cA}(\xi)\in Q$. 
   Then:
   \begingroup
\allowdisplaybreaks
\begin{align*}
  \sem{\cA'}(\xi) &= \h_{\cA'}(\xi)_{\state_{\cA'}(\xi)} \otimes F'_{\state_{\cA'}(\xi)}
   \tag{by Lemma \ref{lm:properties-hA-of-budet-wta-det-new}(4)}  \\
  &= \h_{\cA'}(\xi)_{\state_{\cA}(\xi)} \otimes F'_{\state_{\cA}(\xi)}
        \tag{by \eqref{equ:state-equ}}\\
                &=\h_{\cA}(\xi)_{\state_{\cA}(\xi)} \otimes F'_{\state_{\cA}(\xi)}
                  \tag{by \eqref{equ:hAxia=hAprimexiq-Z}}\\    
                &= \h_{\cA}(\xi)_{\state_{\cA}(\xi)} \otimes F_{\state_{\cA}(\xi)} \tag{by the definition of $F'$}\\
                  &=  \sem{\cA}(\xi) \enspace.  \tag{by Lemma \ref{lm:properties-hA-of-budet-wta-det-new}(4)}           
    \end{align*}
    \endgroup
    
      \underline{$\state_{\cA'}(\xi)\not\in Q'$:} By \eqref{equ:in-Q-iff-in-Q-prime}, we have $\state_{\cA}(\xi)\not\in Q$.
Then $\sem{\cA'}(\xi) = \0 = \sem{\cA}(\xi)$ by Lemma~\ref{lm:properties-hA-of-budet-wta-det-new}(4).
  \end{proof}

\begin{example}\rm \label{ex:example-of-slim-wta-for-a-given-bu.det-wta} We consider the bu-deterministic $(\Sigma,\Ratnum)$-wta $\cA$ of Example \ref{ex:for-state-algebra-new}. By using the construction described in the proof of Theorem \ref{lm:construction-of-simple-det},
    we obtain the bu-deterministic $(\Sigma,\Ratnum)$-wta $\cA'=(Q',\delta',F')$ which is slim and equivalent to $\cA$, where $Q'=\{p_1\}$ and $\delta'_0(\varepsilon,\alpha,p_1) =1$, $\delta'_0(\varepsilon,\beta,p_1)=0$, and $F'_{p_1}=1$.
    \hfill $\Box$
    \end{example}

In fact, for each bu-deterministic wta being slim is a necessary condition for being minimal.

\begin{lemma}\label{lm:minimal-implies-slim}\rm  If $\cA$ is minimal, then $\cA$ is slim.
\end{lemma}
\begin{proof} We prove the contraposition of the statement. Let $\cA$ be not slim. Then there exists a $q_0 \in Q$ such that, for each $\xi \in \T_\Sigma$, we have $q_0\ne \state_\cA(\xi)$, which is equivalent to $\h_\cA(\xi)_{q_0}=\0$ by Lemma~\ref{lm:properties-hA-of-budet-wta-det-new}(2). 

Now we consider the $(\Sigma,\B)$-wta $\cA_{-q_0}=(Q',\delta',F')$ as it is defined in Lemma~\ref{lm:wta-getting-rid-of-useless-state}. Clearly, $\cA_{-q_0}$ is bu-deterministic and $|Q'| = |Q|-1$. Since $\sem{\cA'} = \sem{\cA}$ by Lemma~\ref{lm:wta-getting-rid-of-useless-state}, we obtain that $\cA$ is not minimal.
\end{proof}

\underline{Step 2:} For each slim $(\Sigma,\B)$-wta $\cA$, we define the concept of a candidate set induced by $\cA$ for $\sfMon(\Sigma,\B)/_{\sim_{\sem{\cA}}}$. We prove that each such candidate set generates $\sfMon(\Sigma,\B)/_{\sim_{\sem{\cA}}}$.
Essentially, this is due to the fact that $\cA$ is slim. Then we show how to construct such a candidate set.

Let $\cA$
be slim with $Q= \{q_1,\ldots,q_n\}$ and let us abbreviate $\sem{\cA}$ by $r$.  Since $\cA$ is slim, none of the sets $\state^{-1}(q_1),\ldots,\state^{-1}(q_n)$ is empty. 
For every $\zeta_1 \in \state^{-1}(q_1), \ldots, \zeta_n \in \state^{-1}(q_n)$, we call the set
\(
\{[\1.\zeta_1]_{\sim_r}, \ldots,  [\1.\zeta_n]_{\sim_r}\},
\)
a \emph{candidate set induced by $\cA$ for $\sfMon(\Sigma,\B)/_{\sim_r}$}.
(Hence for each slim  $(\Sigma,\B)$-wta  there exists such a candidate set.)

\begin{lemma} \label{thm:bu-det-r-construction-ofVLQ(r)-nnew-det} \rm \sloppy Let  $\cA$
  be slim  with $|Q|= n$ and let us abbreviate $\sem{\cA}$ by $r$. Moreover, let $H = \{[\1.\zeta_1]_{\sim_r}, \ldots,  [\1.\zeta_n]_{\sim_r}\}$ be a candidate set induced by $\cA$ for $\sfMon(\Sigma,\B)/_{\sim_r}$. The following two statements hold.
\begin{compactitem}
\item[(1)] For every $b \in B$ and  $\xi \in \T_\Sigma$, we have:
  if $\state(\xi) \in Q$, then there exists $j \in [n]$ such that $[b.\xi]_{\sim_r} = d \cdot [\1.\zeta_j]_{\sim_r}$ where
  $d = b \otimes \h_\cA(\xi)_{\state(\xi)} \otimes \Big(\h_\cA(\zeta_j)_{\state(\zeta_j)}\Big)^{-1}$,
  otherwise $[b.\xi]_{\sim_r} = \0\cdot [\1.\zeta_j]_{\sim_r}$ for each $j\in [n]$.
 
\item[(2)]  $H$ generates the $\B$-scalar algebra $\sfMon(\Sigma,\B)/_{\sim_r}$.
  \end{compactitem}
\end{lemma}

\begin{proof} Proof of (1):  Let $b \in B$ and $\xi \in \T_\Sigma$. 
We distinguish the following two cases.

\underline{$\state(\xi) \in Q$:}  Let $j \in [n]$ be the unique number such that $\state(\xi) = \state(\zeta_j)$.
  We prove that $b.\xi \sim_r d.\zeta_j$. Let $c \in \C_\Sigma$. We distinguish two subcases.
  
\underline{$\state(c[\xi]) \in Q$:}   
\begingroup
  \allowdisplaybreaks
\begin{align*}
  b \otimes r(c[\xi])  = & \ b \otimes \h_\cA(c[\xi])_{\state(c[\xi])} \otimes F_{\state(c[\xi])} 
   \tag{by Lemma~\ref{lm:properties-hA-of-budet-wta-det-new}(4)}\\[2mm]
= & \  b \otimes  \h_\cA(\xi)_{\state(\xi)} \otimes \h_\cA^\C(c)(\1_{\state(\xi)})_{\state(c[\xi])} \otimes F_{\state(c[\xi])} 
 \tag{by Lemma~\ref{obs:total-bu-det-wta-calc(new)-det-new}}\\[2mm]
= & \  b \otimes  \h_\cA(\xi)_{\state(\xi)} \otimes \h_\cA^\C(c)(\1_{\state(\zeta_j)})_{\state(c[\zeta_j])} \otimes F_{\state(c[\zeta_j])}
\tag{because  $\state(\xi) = \state(\zeta_j)$ and thus $\state(c[\xi]) = \state(c[\zeta_j])$} \\[2mm]
= & \ d \otimes \h_\cA(\zeta_j)_{\state(\zeta_j)} \otimes \h_\cA^\C(c)(\1_{\state(\zeta_j)})_{\state(c[\zeta_j])} \otimes F_{\state(c[\zeta_j])}
\tag{where $d$ is defined in Statement (1) of the theorem}\\[2mm]
= & \  d \otimes \h_\cA(c[\zeta_j])_{\state(c[\zeta_j])} \otimes F_{\state(c[\zeta_j])}
\tag{by Lemma~\ref{obs:total-bu-det-wta-calc(new)-det-new}}\\[2mm]
= & \  d \otimes r(c[\zeta_j]) \tag{by Lemma~\ref{lm:properties-hA-of-budet-wta-det-new}(4)}
      \enspace.
\end{align*}
\endgroup

\underline{$\state(c[\xi]) \not\in Q$:}  Then $\state(c[\xi]) = \bot$ and, by $\state(c[\xi]) = \state(c[\zeta_j])$, also $\state(c[\zeta_j]) = \bot$.
Thus, by  Lemma~\ref{lm:properties-hA-of-budet-wta-det-new}(4), we obtain $ b \otimes r(c[\xi])  = \0 = d \otimes r(c[\zeta_j])$.

Hence in both cases  $[b.\xi]_{\sim_r} =  [d.\zeta_j]_{\sim_r}$, i.e., $[b.\xi]_{\sim_r} =  d \cdot [\1.\zeta_j]_{\sim_r}$.

\underline{$\state(\xi) \not\in Q$:} Then $\state(\xi)=\bot$.
  We prove that $b.\xi \sim_r \widetilde{\0}$. Let $c \in \C_\Sigma$.
  \begin{align*}
    b \otimes r(c[\xi]) &= b \otimes \bigoplus_{q \in Q} \h_\cA(c[\xi])_q \otimes F_q = b \otimes \bigoplus_{q \in Q} \0 \otimes F_q
    \text{\ (by Lemma~\ref{obs:total-bu-det-wta-calc(new)-det-new})} = \0 \enspace.
  \end{align*}
Thus $[b.\xi]_{\sim_r} =  [\widetilde{\0}]_{\sim_r}$ and hence $[b.\xi]_{\sim_r} = \0\cdot [\1.\zeta_j]_{\sim_r}$ for each $j\in [n]$.

Proof of (2): By Statement (1), we have  $\rmMon(\Sigma,\B)/_{\sim_r} \subseteq B \cdot H$. Since $B \cdot H \subseteq \rmMon(\Sigma,\B)/_{\sim_r}$ and $H \ne \emptyset$, Statement(2) follows from \eqref{eq:equivalent-to-generates}.
\end{proof}

\begin{example}\rm \label{ex:candidate-set-det} The $(\Sigma,\Ratnum)$-wta $\cA$ of Example \ref{ex:bu-det+total-wta-det} is slim; it has the state set $Q=\{o,e\}$. An example of a candidate set induced by $\cA$ for $\sfMon(\Sigma,\Ratnum)/_{\sim_r}$ (with $r= \sem{\cA}$) is
  \(
  H = \{[1.\alpha]_{\sim_r}, \ [1.\sigma(\alpha,\alpha)]_{\sim_r} \}
\)
because $\state_\cA(\alpha) = o$ and $\state_\cA(\sigma(\alpha,\alpha))=e$.
(The set $\{[1.\sigma(\sigma(\alpha,\alpha),\alpha)]_{\sim_r}, \ [1.\sigma(\alpha,\alpha)]_{\sim_r} \}$ is another example of a candidate set induced by $\cA$ for $\sfMon(\Sigma,\Ratnum)/_{\sim_r}$ because also $\state_\cA\big(\sigma(\sigma(\alpha,\alpha),\alpha)\big)=o$.) We demonstrate Lemma \ref{thm:bu-det-r-construction-ofVLQ(r)-nnew-det}(1) as follows:
for each $b \in \mathbb{Q}$, we have
\begingroup
\allowdisplaybreaks
\begin{align*}
  & [b.\sigma(\sigma(\alpha,\alpha),\alpha)]_{\sim_r} =  \ d \cdot [1.\alpha]_{\sim_r} \\
  &\ \ \ \  \text{ with } d = b \cdot \h_\cA(\sigma(\sigma(\alpha,\alpha),\alpha))_o \cdot (\h_\cA(\alpha)_o)^{-1}
  = b \cdot 8 \cdot 2^{-1} = b \cdot 4 \ \text{ and }\\[2mm]
   &[b.\sigma(\sigma(\alpha,\alpha),\sigma(\alpha,\alpha))]_{\sim_r}
     = \ d' \cdot [1.\sigma(\alpha,\alpha)]_{\sim_r} \\
  &\ \ \ \ \text{ with } d' = \ b \cdot \h_\cA(\sigma(\sigma(\alpha,\alpha),\sigma(\alpha,\alpha)))_o \cdot
  (\h_\cA(\sigma(\alpha,\alpha))_o)^{-1} = b \cdot 16 \cdot 4^{-1} = b \cdot 4 \enspace.
\end{align*}
\endgroup
    By Lemma \ref{thm:bu-det-r-construction-ofVLQ(r)-nnew-det}(2), the set $H$ generates $\sfMon(\Sigma,\Ratnum)/_{\sim_r}$.
    \hfill $\Box$
  \end{example}
  
  The question arises whether, for each slim  $(\Sigma,\B)$-wta, each candidate set induced by $\cA$ is pair-independent. 
  By the following example we demonstrate that the answer is negative.

  \begin{example}\rm \label{ex:candidate-set-not-independent} Let $\Sigma=\{\alpha^{(0)},\beta^{(0)}\}$.  Then $\T_\Sigma=\{\alpha,\beta\}$ and $\C_\Sigma=\{z\}$.
    Let $\cA=(Q,\delta,F)$ be the $(\Sigma,\Ratnum)$-wta where
  \begin{compactitem}
  \item $Q=\{q_0,q_1\}$,
  \item  $\delta_0(\varepsilon,\alpha,q_0)=1$ and $\delta_0(\varepsilon,\beta,q_1)=1$,
  \item $F_{q_0}=2$ and $F_{q_1}=1$.
 \end{compactitem}
 Then $\sem{\cA}(\alpha)=2$ and $\sem{\cA}(\beta)=1$.

 The only candidate set is $\{[1.\alpha]_{\sem{\cA}},[1.\beta]_{\sem{\cA}}\}$.
We have $[1.\alpha]_{\sem{\cA}}\ne[1.\beta]_{\sem{\cA}}$ because (with the only context $c=z$) \[1\cdot \sem{\cA}(c[\alpha]) =1\cdot \sem{\cA}(\alpha) =2 \ne 1 = 1\cdot \sem{\cA}(\beta)= 1\cdot \sem{\cA}(c[\beta]).\]
 Moreover, $\{[1.\alpha]_{\sem{\cA}},[1.\beta]_{\sem{\cA}}\}$ is not pair-independent because $[1.\alpha]_{\sem{\cA}}=2\cdot[1.\beta]_{\sem{\cA}}$.
 This is because $2\cdot[1.\beta]_{\sem{\cA}}=[2.\beta]_{\sem{\cA}}$ and $1.\alpha \sim_{\sem{\cA}} 2.\beta$ because for every context (we have only $c=z$) we have
 \[1\cdot \sem{\cA}(c[\alpha])=1\cdot \sem{\cA}(\alpha)=1\cdot 2 = 2 = 2\cdot 1 = 2\cdot \sem{\cA}(\beta)=2\cdot \sem{\cA}(c[\beta]).\]
   \hfill $\Box$
  \end{example}

  \begin{lemma}\rm \label{lm:tq-has-minimal-height} Let  $\cA$ be slim. We can construct a finite generating set $\mathrm{H}_{\cA}$ of  $\sfMon(\Sigma,\B)/_{\sim_{\sem{\cA}}}$.
  \end{lemma}
  \begin{proof} Let us abbreviate $\sem{\cA}$ by $r$. First we construct a mapping $\mathrm{t}_\cA: Q \to \T_\Sigma$ such that $\state(\mathrm{t}_\cA(q))=q$ for each $q\in Q$. The construction is as follows.

We assume that a linear order $\xi_1,\xi_2,\ldots$ on $\T_\Sigma$ is given.
  We construct a finite prefix of a sequence $t_0,t_1,t_2,\ldots$ of partial mappings from $Q$ to $\T_\Sigma$ as follows. We let $t_0=\emptyset$. Let $t_i$ be already constructed. If $t_i$ is defined for each $q\in Q$,  then we stop the construction and we let  $\mathrm{t}_\cA=t_i$. Otherwise we proceed by case analysis. If $\state(\xi_{i+1})=\bot$ or  $\state(\xi_{i+1})\in Q$ and $t_i(\state(\xi_{i+1}))$ is already defined, 
  then let $t_{i+1}=t_i$. Otherwise we let $t_{i+1}=t_i \cup \{(\state(\xi_{i+1}),\xi_{i+1})\}$.  Since $\cA$ is slim, and hence $Q\subseteq \im(\state)$, there exists an $i \in \mathbb{N}$ such that $t_i$ is a mapping, and hence the construction eventually stops.
  
Then the set $\mathrm{H}_{\cA} = \{[\1.\mathrm{t}_\cA(q)]_{\sim_{r}} \mid q \in Q\}$
is a candidate set induced by $\cA$ for $\sfMon(\Sigma,\B)/_{\sim_{r}}$.
By Lemma \ref{thm:bu-det-r-construction-ofVLQ(r)-nnew-det}(2), $\mathrm{H}_\cA$ generates $\sfMon(\Sigma,\B)/_{\sim_{r}}$.
\end{proof}

 \underline{Step 3:} 
Now we show how to construct, for a given  bu-deterministic $(\Sigma,\B)$-wta, an equivalent minimal bu-deterministic $(\Sigma,\B)$-wta.
  
 \begin{theorem}\label{thm:minimization-theorem-new-2} Let $\B$ be a commutative semifield. Moreover, let $\cA$ be a bu-deterministic $(\Sigma,\B)$-wta and let us abbreviate $\sem{\cA}$ by $r$. We assume that dependency in $\sfMon(\Sigma,\B)/_{\sim_r}$ is decidable and the decomposition mapping $\mathrm{dec}$ for $\big(\rmMon(\Sigma,\B)/_{\sim_r}\big)\setminus \{[\widetilde{\0}]_{\sim_r}\}$ is computable.
   Then we can construct a finite scalar-basis $H$ of $\sfMon(\Sigma,\B)/_{\sim_r}$
 and we can construct the minimal  bu-deterministic $(\Sigma,\B)$-wta $\budwta(r,\sim_r,H)$.
     \end{theorem} 
   \begin{proof} By Theorem \ref{lm:construction-of-simple-det}, we can construct a slim wta which is equivalent to $\cA$. Hence we assume that $\cA$ itself is slim.
     By  using Lemma \ref{lm:tq-has-minimal-height}, we can construct a finite generating set  $\mathrm{H}_{\cA}$ of $\sfMon(\Sigma,\B)/_{\sim_{r}}$. 
     Since dependency in $\sfMon(\Sigma,\B)/_{\sim_r}$ is decidable, by Lemma~\ref{lm:crucial-canc-pair-ind-imply-uniqueness-det}(1), we can construct a finite scalar-basis $H\subseteq \mathrm{H}_{\cA}$ of $\sfMon(\Sigma,\B)/_{\sim_r}$.
  By Lemma \ref{cor:min-du-det-wta-exists} the $(\Sigma,\B)$-wta $\budwta(r,\sim_r,H)$ is defined and minimal.
       Since the decomposition mapping $\mathrm{dec}$ is computable, we can construct $\budwta(r,\sim_r,H)$ (cf. Lemma~\ref{lm:wta(r,sim,H)}(3)).
\end{proof}

Finally, we give an example of the construction of the minimal bu-deterministic $(\Sigma,\B)$-wta. 

  \begin{example}\rm \label{ex:B-implies-A} Here we show an example of the construction of a  minimal bu-deterministic wta as it is described in the proof of Theorem \ref{thm:minimization-theorem-new-2}.

We consider the ranked alphabet  $\Sigma=\{\gamma^{(1)},\alpha^{(0)}\}$. Each tree in $\T_\Sigma$ has the form $\gamma(\ldots\gamma(\alpha)\ldots)$ with $n$ occurrences of $\gamma$ for some $n\in\mathbb{N}$. We denote this tree by $\gamma^n(\alpha)$.
    Let us consider the  $(\Sigma,\Ratnum)$-wta $\cA=(Q,\delta,F)$ (cf. Figure~\ref{fig:making-wta-minimal}(left)) where
\begin{compactitem}
\item $Q=\{q_1,q_2,q_3\}$,
\item $\delta_0(\varepsilon,\alpha,q_1)=1$,  $\delta_1(q_1,\gamma,q_2)=\delta_1(q_2,\gamma,q_3)=\delta_1(q_3,\gamma,q_2)=1$, and $\delta_1(p,\gamma,q)=0$ for any other combination $p,q\in Q$, and
\item $F_{q_1}=F_{q_3}=2$ and $F_{q_2}=3$.
\end{compactitem}
Obviously, $\cA$ is total and bu-deterministic.  Moreover, $\cA$ is slim, because, e.g.,  $\state(\alpha)=q_1$, $\state(\gamma(\alpha))=q_2$, and $\state(\gamma^2(\alpha))=q_3$.
Also, for each $n \in \mathbb{N}$, we have
\[\sem{\cA}(\gamma^n(\alpha))=\begin{cases}
2 & \text{ if $n$ is even,} \\
3 & \text{ otherwise.} 
\end{cases}
\]

 \begin{figure}[t]
   \centering
     \begin{tikzpicture}

       \hspace*{-30mm}
       \begin{scope}
\tikzset{node distance=7em, scale=0.5, transform shape}
\node[state, rectangle] (alpha) {\Large $\alpha$};
\node[state, right of=alpha] (q1) {\Large $q_1$};
\node[state, rectangle, right of=q1] (left-gamma){\Large $\gamma$};
\node[state, right of=left-gamma] (q2){\Large $q_2$};
\node[state, rectangle, right of=q2] (right-gamma) {\Large $\gamma$};
\node[state, right of=right-gamma] (q3){\Large $q_3$};
\node[state, rectangle, above of=right-gamma] (up-gamma-right) {\Large $\gamma$};

\tikzset{node distance=2em}
\node[above of=q1] (w0)[yshift=-5em] {\Large 2};
\node[above of=alpha] (w1)[yshift=0.7em] {\Large 1};
\node[above of=q2] (w2)[yshift=-5em] {\Large 3};
\node[above of=left-gamma] (w3)[yshift=0.7em] {\Large 1};
\node[above of=right-gamma] (w4)[yshift=0.7em] {\Large 1};
\node[above of=q3] (w5)[yshift=-5em] {\Large 2};
\node[above of=up-gamma-right] (w7)[yshift=0.7em] {\Large 1};

\draw[->,>=stealth] (alpha) edge (q1);
\draw[->,>=stealth] (q1) edge (left-gamma);
\draw[->,>=stealth] (q2) edge (right-gamma);
\draw[->,>=stealth] (left-gamma) edge (q2);
\draw[->,>=stealth] (q2) edge (right-gamma);
\draw[->,>=stealth] (right-gamma) edge (q3);
\draw[->,>=stealth] (q3) edge[out=120, in=-20, looseness=1] (up-gamma-right);
\draw[->,>=stealth] (up-gamma-right) edge[out=200, in=60, looseness=1] (q2);
       \end{scope}

\hspace*{80mm}
       
          \begin{scope}
\tikzset{node distance=7em, scale=0.5, transform shape}
\node[state, rectangle] (alpha){\Large $\alpha$};
\node[state, right of=alpha] (1alpha){\Large $[1.\alpha]_{\sim_r}$};
\node[state, rectangle, right of=1alpha] (gamma) {\Large $\gamma$};
\node[state, right of=gamma] (1gammaalpha)[xshift=2em] {\Large $[1.\gamma(\alpha)]_{\sim_r}$};
\node[state, rectangle, above of=gamma] (up-gamma) {\Large $\gamma$};

\tikzset{node distance=2em}
\node[above of=alpha] (w1)[yshift=0.7em] {\Large 1};
\node[above of=1alpha] (w2)[yshift=-6em] {\Large 2};
\node[above of=gamma] (w3)[yshift=0.7em] {\Large 1};
\node[above of=up-gamma] (w4)[yshift=0.7em] {\Large 1};
\node[above of=1gammaalpha] (w5)[yshift=-7em] {\Large 3};

\draw[->,>=stealth] (alpha) edge (1alpha);
\draw[->,>=stealth] (1alpha) edge (gamma);
\draw[->,>=stealth] (gamma) edge (1gammaalpha);
\draw[->,>=stealth] (1gammaalpha) edge[out=130, in=-20, looseness=1] (up-gamma);
\draw[->,>=stealth] (up-gamma) edge[out=200, in=60, looseness=1] (1alpha);
       \end{scope}
       
\end{tikzpicture} 
\caption{\label{fig:making-wta-minimal} The $(\Sigma,\Ratnum)$-wta $\cA$ of Example \ref{ex:B-implies-A} (left) and $\budwta(r,\sim_r,\mathrm{H}_{\cA})$ (right).}
   \end{figure}

Let us abbreviate $\sem{\cA}$ by $r$. It is easy to show that, for every $m,n \in \mathbb{N}$ and $b_1,b_2\in \mathbb{Q}$, we have
\[b_1.\gamma^m(\alpha) \sim_r b_2.\gamma^n(\alpha) \ \text{ iff } \ \Big(b_1=b_2 \text{\ \ and \ \ } \big( m\!\!\!\mod (2) = n\!\!\!\mod (2)\big)\Big) \text{ or } b_1=b_2=0.\]
 Hence, for every $n\in \mathbb{N}$ and $b\in \mathbb{Q}$, we have
 \begin{align}\label{ex:what-is-bgamman}
   [b.\gamma^n(\alpha)]_{\sim_r}=\begin{cases}  \{b.\gamma^m(\alpha) \mid m \in \mathbb{N} \ \text{ and } \  m\!\!\!\!\mod (2) = n\!\!\!\!\mod (2)\} & \text{ if $b\ne 0$}\\
 \{\widetilde{0}\} & \text{otherwise.}
  \end{cases}
  \end{align}
Using the above, it is easy to see that dependency in $\sfMon(\Sigma,\B)/_{\sim_r}$ is decidable.

By Lemma \ref{lm:tq-has-minimal-height}, we construct the mapping $\mathrm{t}_\cA: Q \to \T_\Sigma$ with 
\[\mathrm{t}_\cA(q_1)=\alpha, \ \ \mathrm{t}_\cA(q_2)=\gamma(\alpha), \text{ and} \ \ \mathrm{t}_\cA(q_3)=\gamma^2(\alpha) \enspace.
\]

We have $1.\alpha \sim_r 1.\gamma^2(\alpha)$, i.e., $[1.\alpha]_{\sim_r}=[1.\gamma^2(\alpha)]_{\sim_r}$. Moreover, there does not exist $b\in \mathbb{Q}$ with $1.\alpha \sim_r b.\gamma(\alpha)$, hence the elements  $[1.\alpha]_{\sim_r}$ and $[1.\gamma(\alpha)]_{\sim_r}$ are independent. We obtain that $\mathrm{H}_{\cA}=\{[1.\alpha]_{\sim_r}, [1.\gamma(\alpha)]_{\sim_r}\}$, and this set is pair-independent. Thus there is no need to apply Lemma~\ref{lm:crucial-canc-pair-ind-imply-uniqueness-det}(1) in this example.
  
By using the characterization in \eqref{ex:what-is-bgamman}, we can compute the mapping $\mathrm{dec}$. In fact, for every $b\in \mathbb{Q}$ and $n \in \mathbb{N}$, we have
\begin{align}\label{eq:congruence-classes-another}
[b.\gamma^n(\alpha)]_{\sim_r} = 
\begin{cases}
[b.\alpha]_{\sim_r}=b \cdot[1.\alpha]_{\sim_r} & \text{ if \ $\big(n\!\!\mod(2)\big)=0$}\\
[b.\gamma(\alpha)]_{\sim_r}= b\cdot [1.\gamma(\alpha)]_{\sim_r} & \text{ if \ $\big(n\!\!\mod(2)\big)=1$}.
\end{cases}
\end{align}
In particular, for each $n \in \mathbb{N}$, we have
\[
  \mathrm{dec}([1.\gamma^n(\alpha)]_{\sim_r}) =
  \begin{cases}
(1,[1.\alpha]_{\sim_r}) & \text{ if \ $\big(n\!\!\mod(2)\big)=0$}\\
(1,[1.\gamma(\alpha)]_{\sim_r}) & \text{ if \ $\big(n\!\!\mod(2)\big)=1$}.
    \end{cases}
  \]

Finally, we construct the bu-deterministic $(\Sigma,\Ratnum)$-wta $\budwta(r,\sim_r,\mathrm{H}_{\cA})=(\mathrm{H}_{\cA},\delta',F')$ (cf. Figure~\ref{fig:making-wta-minimal}(right)), where
\begin{compactitem}
  \item $\mathrm{H}_{\cA}=\{[1.\alpha]_{\sim_r}, [1.\gamma(\alpha)]_{\sim_r}\}$,
\item $(\delta')_0(\varepsilon,\alpha,[1.\alpha]_{\sim_r}) = (\delta')_1([1.\alpha]_{\sim_r},\gamma,[1.\gamma(\alpha)]_{\sim_r})
  = (\delta')_1([1.\gamma(\alpha)]_{\sim_r},\gamma,[1.\alpha]_{\sim_r}) = 1$ and\\
  $(\delta')_0(\varepsilon,\alpha,[1.\gamma(\alpha)]_{\sim_r}) = 0$, and
  \item $(F')_{[1.\alpha]_{\sim_r}} = 2$ and $(F')_{[1.\gamma(\alpha)]_{\sim_r}} = 3$.
\end{compactitem}
By Lemma \ref{lm:wta(r,sim,H)}, we have $\sem{\budwta(r,\sim_r,\mathrm{H}_{\cA})}=r$.
By Theorem~\ref{thm:minimization-theorem-new-2}, the wta $\budwta(r,\sim_r,\mathrm{H}_{\cA})$ is minimal.
\hfill$\Box$
  \end{example}

  \subsection{A characterization of minimality}
  \label{subsect:characterization-of-minimlity}

Here we characterize the property of a bu-deterministic wta $\cA$ of being minimal by the properties that $\cA$ is slim and the degrees of $\sfMon(\Sigma,\B)/_{\ker(\sfh_\cA)}$ and $\sfMon(\Sigma,\B)/_{\sim_{\sem{\cA}}}$ coincide. We recall that, for each bu-deterministic wta $\cA$ (i.e., also for a non-minimal one), we have $\deg(\sfMon(\Sigma,\B)/_{\ker(\sfh_\cA)}) \ge \deg(\sfMon(\Sigma,\B)/_{\sim_{\sem{\cA}}})$.

  {\em In this subsection $\cA=(Q,\delta,F)$ denotes an arbitrary bu-deterministic $(\Sigma,\B)$-wta.}

\begin{lemma}\rm\label{lm:slim-then-h-A-surjective} If $\cA$ is slim, then $\im(\sfh_\cA)=B_{\le 1}^Q$ and $\deg(\sfMon(\Sigma,\B)/_{\ker(\sfh_\cA)})=|Q|$.
\end{lemma}
\begin{proof} Let $\cA$ be slim. Since $\im(\sfh_\cA) \subseteq B_{\le 1}^Q$, we prove $B_{\le 1}^Q \subseteq \im(\sfh_\cA)$.
  For each $\xi \in \T_\Sigma$, we have $\sfh_\cA(\widetilde{\0}) = \0^Q$. Hence $\0^Q \in \im(\sfh_\cA)$. Now let $b \in B^{-\0}$ and $q \in Q$. We prove that $b_q \in \im(\sfh_\cA)$.  Since $\cA$ is slim, there exists $\xi \in \T_\Sigma$ such that $q = \state_\cA(\xi)$. By Lemma~\ref{lm:properties-hA-of-budet-wta-det-new}(2), we have $\h_\cA(\xi)_q \ne \0$. Then
     \begin{align*}
    b_q &= b \cdot \1_q = (c \otimes \h_\cA(\xi)_q) \cdot \1_q
          \tag{with $c=b \otimes (\h_\cA(\xi)_q)^{-1}$}\\
    &= c \cdot (\h_\cA(\xi)_q \cdot \1_q) = c \cdot \h_\cA(\xi)
          \text{\ (by Lemma~\ref{lm:properties-hA-of-budet-wta-det-new}(3))} = c \cdot \sfh_\cA(\1.\xi)
          \tag{by \eqref{equ:sfh-related-to-h}}\\
    &= \sfh_\cA(c.\xi) \tag{because $\sfh_\cA$ is a scalar-linear mapping}\enspace.
    \end{align*}
    Hence $b_q \in \im(\sfh_\cA)$.  Thus $B_{\le 1}^Q \subseteq \im(\sfh_\cA)$.

Since $\im(\sfh_\cA)=B_{\le 1}^Q$, we have $\sfM_{\mathrm{im}}(\cA) =(B_{\le 1}^Q,\0^Q,\delta_\cA)$.
By~\eqref{eq:isomorphism}, this implies $\deg(\sfMon(\Sigma,\B)/_{\ker(\sfh_\cA)})=\deg((B_{\le 1}^Q,\0^Q))=|Q|$.
\end{proof}

\begin{lemma}\rm\label{lm:minimal-then-degrees-equal}
 If $\cA$ is minimal, then   $\deg(\sfMon(\Sigma,\B)/_{\ker(\sfh_\cA)})=\deg(\sfMon(\Sigma,\B)/_{\sim_{\sem{\cA}}})$.
 \end{lemma}
 \begin{proof} Let $\cA$ be minimal. By Lemma~\ref{lm:minimal-implies-slim}, $\cA$ is slim and hence,  by Lemma~\ref{lm:slim-then-h-A-surjective},   $\deg(\sfMon(\Sigma,\B)/_{\ker(\sfh_\cA)})=|Q|$.
Moreover, we also have $\deg(\sfMon(\Sigma,\B)/_{\sim_{\sem{\cA}}})=|Q|$ because, by   Lemma \ref{cor:min-du-det-wta-exists}, 
$ \deg(\sfMon(\Sigma,\B)/_{\sim_{\sem{\cA}}})$ is the number of the states of a minimal bu-deterministic $(\Sigma,\B)$-wta which recognizes $\sem{\cA}$.
\end{proof}

\begin{lemma}\rm\label{lm:degrees-equal-then-minimal}
 If $\cA$ is slim and $\deg(\sfMon(\Sigma,\B)/_{\ker(\sfh_\cA)})=\deg(\sfMon(\Sigma,\B)/_{\sim_{\sem{\cA}}})$, then $\cA$ is minimal.
  \end{lemma}
\begin{proof} Let $\cA$ be slim. By Lemma~\ref{lm:slim-then-h-A-surjective},   $\deg(\sfMon(\Sigma,\B)/_{\ker(\sfh_\cA)})=|Q|$. Then, by our assumption, we have 
$\deg(\sfMon(\Sigma,\B)/_{\sim_{\sem{\cA}}})=|Q|$. Hence $\cA$ is minimal because, by Lemma \ref{cor:min-du-det-wta-exists},
$\deg(\sfMon(\Sigma,\B)/_{\sim_{\sem{\cA}}})$ is the number of the states of a minimal bu-deterministic wta which recognizes $\sem{\cA}$.
\end{proof}

Now by Lemmas~\ref{lm:minimal-implies-slim} and \ref{lm:minimal-then-degrees-equal}   and Lemma \ref{lm:degrees-equal-then-minimal} we obtain the following characterization of minimality.

\begin{corollary}\rm\label{minimal-iff-degrees-equal} Let $\cA$ be a  bu-deterministic $(\Sigma,\B)$-wta. Then $\cA$  is minimal if and only if $\cA$ is slim and $\deg(\sfMon(\Sigma,\B)/_{\ker(\sfh_\cA)})=\deg(\sfMon(\Sigma,\B)/_{\sim_{\sem{\cA}}})$.
\end{corollary}

\bibliographystyle{alpha}
\phantomsection
\addcontentsline{toc}{chapter}{Bibliography} 
\bibliography{fund1.bbl}


\end{document}